\documentclass[11pt]{article}

\usepackage[shortlabels]{enumitem}
\usepackage{titling}

\usepackage{palatino} 
\usepackage{mathpazo}
\usepackage{braket}
\usepackage{amsfonts}
\usepackage{amssymb}
\usepackage{amsmath}
\usepackage{mathrsfs} 
\usepackage{bbm}
\usepackage{latexsym}
\usepackage{amsthm}
\usepackage[usenames]{color}
\usepackage{hyperref, etoolbox}
\usepackage{diagbox}
\usepackage{subcaption}
\usepackage{verbatim}
\usepackage{ragged2e, graphicx, pifont}
\usepackage{indentfirst}
\usepackage{algorithm}
\usepackage[noend]{algpseudocode}
\usepackage{tikz-cd}
\usetikzlibrary{decorations.pathmorphing}
\usetikzlibrary{arrows.meta}
\usepackage{mdframed}
\usepackage{thm-restate}

\hypersetup{
colorlinks = true,
citecolor= blue,
linkcolor= blue
}

\floatname{algorithm}{Protocol}

\makeatletter
\renewcommand\paragraph{\@startsection{paragraph}{4}{\z@}%
                                      {\parskip}
                                      {-1em}%
                                      {\normalfont\normalsize\bfseries}}
\makeatother

\usepackage{tabularx}

\hypersetup{pdfpagemode=UseNone}

\newtheorem{theorem}{Theorem}
\newtheorem{lemma}[theorem]{Lemma}

\newtheorem{cor}[theorem]{Corollary}
\newtheorem{definition}{Definition}

\newtheorem{remark}{Remark}

\newtheorem{fact}{Fact}

\newcommand{\norm}[1]{\ensuremath{\left\lVert #1 \right\rVert}}
\newcommand{\eps}{\varepsilon}
\newcommand{\matel}[3]{\langle #1 | #2 | #3\rangle}

\newcommand{\clA}{\mathcal{A}}
\newcommand{\clB}{\mathcal{B}}

\newcommand{\clE}{\mathcal{E}}

\newcommand{\clO}{\mathcal{O}}
\newcommand{\clP}{\mathcal{P}}
\newcommand{\clQ}{\mathcal{Q}}
\newcommand{\clR}{\mathcal{R}}
\newcommand{\clS}{\mathcal{S}}
\newcommand{\clT}{\mathcal{T}}
\newcommand{\clU}{\mathcal{U}}
\newcommand{\clV}{\mathcal{V}}

\newcommand{\clX}{\mathcal{X}}
\newcommand{\clY}{\mathcal{Y}}
\newcommand{\clZ}{\mathcal{Z}}

\newcommand{\sfB}{\mathsf{B}}
\newcommand{\sfD}{\mathsf{D}}
\newcommand{\sfF}{\mathsf{F}}
\newcommand{\sfH}{\mathsf{H}}
\newcommand{\sfI}{\mathsf{I}}
\newcommand{\sfP}{\mathsf{P}}

\newcommand{\sfV}{\mathsf{V}}

\newcommand{\Id}{\mathbbm{1}}
\newcommand{\bbF}{\mathbb{F}}
\newcommand{\bbR}{\mathbb{R}}

\newcommand{\bbE}{\mathop{\mathbb{E}}}
\newcommand{\bbP}{\mathbb{P}}

\newcommand{\A}{\mathrm{A}}
\newcommand{\B}{\mathrm{B}}

\newcommand{\Q}{\mathrm{Q}}
\newcommand{\V}{\mathrm{V}}

\newcommand{\tB}{\widetilde{B}}
\newcommand{\tE}{\widetilde{E}}
\newcommand{\tK}{\widetilde{K}}

\newcommand{\tX}{\widetilde{X}}
\newcommand{\tY}{\widetilde{Y}}

\newcommand{\vph}{\varphi}

\newcommand{\ketbra}[2]{\ket{#1}\!\!\bra{#2}}
\newcommand{\state}[1]{\ketbra{#1}{#1}}

\newcommand{\inprod}[2]{\langle #1 | #2\rangle}
\newcommand{\inner}[2]{\left\langle #1, #2\right\rangle}

\newcommand{\IP}{\mathrm{IP}}
\newcommand{\Tr}{\mathrm{Tr}}
\newcommand{\MS}{\mathrm{MS}}
\newcommand{\MSe}{\mathrm{MSE}}

\newcommand{\supp}{\mathrm{supp}}
\newcommand{\suc}{\mathrm{suc}}
\newcommand{\bad}{\mathrm{Bad}}

\newcommand{\err}{\mathrm{err}}

\newcommand{\eff}{\mathrm{eff}^*}

\newcommand{\polylog}{\mathrm{polylog}}

\allowdisplaybreaks

\title{A direct product theorem for quantum communication complexity with applications to device-independent cryptography}
\author{Rahul Jain \thanks{Centre for Quantum Technologies and Department of Computer Science, National University of Singapore and MajuLab, UMI 3654, Singapore. Email:~\tt{rahul@comp.nus.edu.sg}} \and
Srijita Kundu \thanks{Institute for Quantum Computing, University of Waterloo. Email:~\tt{srijita.kundu@uwaterloo.ca}}}
\date{}

\begin{document}
\maketitle
\vspace{-1.5cm}
\begin{abstract}
We give a direct product theorem for the entanglement-assisted interactive quantum communication complexity of an $l$-player predicate $\sfV$. In particular we show that for a distribution $p$ that is product across the input sets of the $l$ players, the success probability of any entanglement-assisted quantum communication protocol for computing $n$ copies of $\sfV$, whose communication is $o(\log(\eff(\sfV,p))\cdot n)$, goes down exponentially in $n$. Here $\eff(\sfV, p)$ is a distributional version of the quantum efficiency or partition bound introduced by Laplante, Lerays and Roland (2012), which is a lower bound on the distributional quantum communication complexity of computing a single copy of $\sfV$ with respect to $p$. For a two-input boolean function $f$, the best result for interactive quantum communication complexity known previously was due to Sherstov (2018), who showed a direct product theorem in terms of the generalized discrepancy, which is a lower bound on communication. Our lower bound on non-distributional communication complexity is in terms of $\max_{\text{product }p}\eff(\sfV,p)$, and there is no known relationship between this and the generalized discrepancy. But we define a distributional version of the generalized discrepancy bound and can show that for a given $p$, $\eff(\sfV,p)$ upper bounds it. Moreover, unlike Sherstov's result, our result works for two-input functions or relations whose outputs are non-boolean as well, and is a strong direct product theorem for functions or relations whose quantum communication complexity is characterized by $\eff(\sfV_f,p)$ for a product $p$.

Applying our direct product theorem for small communication and techniques related to $\eff$, we show that it is possible to do device-independent (DI) quantum cryptography without the assumption that devices do not leak any information. First, we analyze the parallel DI quantum key distribution protocol given by Jain, Miller and Shi (2020), and show that when the protocol is carried out with devices that are compatible with $n$ copies of the Magic Square game, it is possible to extract $\Omega(n)$ bits of key from it, even in the presence of $O(n)$ bits of leakage. Second, we show that it is possible to do sequential versions of the Jain, Miller and Shi protocol, which give a better key rate for QKD with leakage, and let us do sequential DI randomness expansion with leakage (it is not known how to do parallel DI randomness expansion even without leakage). Third, we show that proofs of quantumness with two entangled provers are resistant to leakage, i.e., classical players who communicate $O(n)$ bits with each other cannot convince the verifier that they share entanglement.
\end{abstract}

\section{Introduction}\label{sec:intro}
Communication complexity is an important model of computation with connections to many parts of theoretical computer science \cite{KN96}. In this paper, we consider the communication complexity of computing a predicate $\sfV$ on $(\clA^1\times\ldots\times\clA^l)\times(\clX^1\times\ldots\times\clX^l)$ by $l (\geq 2)$ players who receive inputs $x^1\ldots x^l \in \clX^1\times\ldots\times\clX^l$, and after communicating interactively, are required to produce outputs $a^1\ldots a^l$ such that $\sfV(a^1\ldots a^l, x^1\ldots x^l)$ is satisfied. The $l$ players cooperate and wish to minimize the total number of bits (in the classical model) or qubits (in the quantum model) communicated. The communication complexity of predicates generalizes the communication complexity of (total or partial) functions and relations that are most often considered in the literature.

In any model of computation, a fundamental question is: if we know how to do one copy of a task, what is the best way to do $n$ independent copies of it? One possible way is to simply do each copy independently; if we have an algorithm that successfully does a single copy of the task with probability $1-\eps$, the success probability of this product strategy is $(1-\eps)^n$ and its cost is $n$ times the cost of doing a single copy. For many tasks, this is the best one can do, and a direct product theorem for the task proves so. That is, a direct product theorem proves that any protocol for doing $n$ copies of the task that has cost at most $cn$, where $c$ is some lower bound on the cost of doing one copy with success probability less than 1, has success probability exponentially small in $n$. When $c$ is the exact cost of doing a single copy of the task, we call such a result a strong direct product theorem.

Direct product theorems are known in a number of computational models. In classical communication complexity, there is a long line of works showing direct product and weaker direct sum theorems (which show that the success probability of a protocol that uses $cn$ resources is at most constant, instead of exponentially small) in the two-party setting \cite{Raz92,CSWY01,BJKS02,JRS03,KSW07,VW08,LSS08,HJMR10,BR11,JY12,BBCR13,BRWY13a,BRWY13b,JPY16}.

For quantum communication, a direct sum theorem for one-way quantum communication for general functions was shown by \cite{JRS03}. \cite{BRW08} showed a direct product theorem for functions whose one-way quantum communication complexity is characterized by VC dimension, and \cite{JK20} showed a direct product theorem for one-way quantum communication complexity of general functions. In the interactive quantum setting however, direct product theorems are known only for special classes of functions, for example \cite{KSW07} showed a direct product theorem for symmetric functions. \cite{She12} showed a direct product theorem for the generalized discrepancy method, which is one of the strongest lower bound techniques on quantum communication complexity --- this gives a strong direct product theorem for functions whose quantum communication complexity is exactly characterized by the generalized discrepancy method.

Direct product theorems in communication are related to parallel repetition theorems for non-local games. A non-local game with $l$ players is defined by a predicate $\sfV$ and a distribution $p$. The players are given inputs $x^1\ldots x^l$ from distribution $p$ on $\clX^1\times\ldots\times\clX^l$, and they are required to produce outputs $a^1\ldots a^l$ in $\clA^1\times\ldots\times\clA^l$ so that $\sfV(a^1\ldots a^l,x^1\ldots x^l)$ is satisfied, without communicating. In the classical model, the players are allowed to share randomness, and in the quantum model they are allowed to share entanglement. The maximum winning probability of the game over all strategies is called the value of the game, which may be quantum or classical. A parallel repetition theorem shows that the value of $n$ independent instances of a non-local game is $(1-\eps)^{\Omega(n)}$, if the value a single instance is $(1-\eps)$.

A parallel repetition theorem for the classical value of general two-player non-local games was first shown by Raz \cite{Raz95}, and the proof was subsequently simplified by Holenstein \cite{Hol09}. A strong parallel repetition theorem for the quantum value of a general two-player non-local game is not known. Parallel repetition theorems were shown for special classes of two-player games such as XOR games \cite{CSUU08}, unique games \cite{KRT10} and projection games \cite{DSV15}. When the type of game is not restricted but the input distribution is, parallel repetition theorems have been shown under product distributions \cite{JPY14} and anchored distributions \cite{BVY17} --- both of these results can be extended to $l$ players. For general two-player games, the best current result is due to Yuen \cite{Yuen16}, which shows that the quantum value of $n$ parallel instances of a general game goes down polynomially in $n$, if the quantum value of the original game is strictly less than 1. The situation for more than two players is much less understood.

\paragraph{Device-independent cryptography.} Quantum cryptography lets us do a number of tasks with information theoretic security, i.e., security without any computational assumptions, that are not possible classically. Two basic examples are quantum key distribution (QKD) \cite{BB84} and randomness expansion (RE). In a key distribution scenario, two honest parties Alice and Bob want to share a key, i.e., a uniformly random string of a given length, which is secret from a third party eavesdropper Eve. If Alice and Bob have access to secure private randomness and an authenticated classical channel, it is possible to do the key distribution task quantumly with information theoretic security, but not classically. In randomness expansion, depending on setting, a single party holds a quantum device or quantum devices, as well as some private randomness, and wishes to get more randomness that is secure against Eve using these. In a conventional security proof for these tasks (or any other quantum cryptographic protocol), one needs to have a complete description of the quantum devices, i.e., the states and measurements used by Alice and Bob. However, in practice quantum devices are often not fully characterized, and protocols that rely on complete characterization of quantum devices often have loopholes.

A way around this problem is the framework of device-independent cryptography, which tries to give quantum protocols for cryptographic tasks that are secure even when the devices used by the honest parties are not fully characterized, and in fact can be arbitrarily manipulated by dishonest parties. All known device-independent protocols with information theoretic security use non-local games and rely on the property of self-testing or rigidity displayed by some non-local games. Suppose we play a non-local game with devices implementing some unknown state and measurements, and in fact even the dimension of the systems are unspecified. If these state and measurements regardless achieve a winning probability for the game that is close to its optimal winning probability, then self-testing tells us that the state and measurements are close to the ideal state and measurements for that game, up to trivial isometries. For device-independent QKD (DIQKD), this means in particular that the measurement outputs of the devices given the inputs are random, i.e., they cannot be predicted by a third party even if they have access to the inputs used. This lets us use the outputs of the devices to produce a secret key.

A number of protocols and security proofs for DIQKD and DIRE have been given over the years, in the sequential \cite{PAB+09, ADF+18, VV19} as well as parallel setting \cite{JMS17,Vid17}. Aside from assuming that Alice and Bob's devices are modelled by quantum mechanics however, all these proofs require the assumption that Alice and Bob's devices do not leak any information, i.e., do not communicate with each other or with Eve, unbeknownst to Alice and Bob. Although there have been some works studying non-local games in the presence of communication \cite{TZB+20,TZWP20}, and an argument showing device-independent QKD may be possible in the presence of a specific model of information leakage in \cite{SPM13}, none of these approaches have been developed into a full-fledged proof of security when there is leakage.

\paragraph{Proofs of quantumness.} A proof of quantumness is a protocol between a classical verifier and a prover or provers who claim to be able to do quantum operations, that they indeed can do quantum operations. Proofs of quantumness fall in the device-independent framework by default, since the verifier is entirely classical, and does not trust the operations the provers claim they are doing. It has long been folklore that two provers who cannot communicate with each other can prove that they are quantum, i.e., they share entanglement, to a verifier. This is simply because there are non-local games whose quantum values are higher than their classical values: the verifier simply takes on the role of the referee of the non-local game, and accepts the players' `proof' if their outputs satisfy the winning condition of the game with the inputs provided. The soundmess of this protocol can easily be increased by parallel repetition. This simple protocol has a number of advantages, namely that it is plausibly implementable with current quantum devices, does not require any computational assumptions on the provers, and the verification procedure is efficient. However, ensuring that the provers indeed do not communicate in such a setting is challenging; even if we spatially separate the two provers, the verifier has to be able to communicate back and forth with them, and in the time that this happens, the players may also be able to communicate with each other.

For this reason, more recently there has been more interest in proofs of quantumness with a single prover, with a number of methods for doing this proposed. The sampling method \cite{AA10, BJS10, AAB+19} simply requires the prover to sample from a distribution that was conjectured to be hard for a polynomial-time classical prover (this conjecture is non-standard, and has come under question \cite{LLL+21}); this method is feasible with current quantum devices, but the verifier needs to perform an exponential-time computation to check that the output is correct. Recently, a non-interactive proof of quantumness in the quantum random oracle model (QROM) was proposed \cite{YZ22}; this method has the advantage of being provably hard for polynomial-time classical provers, and also efficiently (and publicly) verifiable, but it works in the QROM rather than the real world, and also not with near-term quantum devices. A number of interactive protocols to prove quantumness based on more standard computational assumptions (such as the hardness of the learning with errors (LWE) problem), which also have efficient verification procedures, have been proposed \cite{BCM+18, BKVV20, KCVY21, KLVY22}; these protocols are not quite feasible with current quantum devices, but the devices required are more practical than those in \cite{YZ22}. Interestingly, some of these protocols \cite{KCVY21, KLVY22} are based on proofs of quantumness from non-local games. The protocol in \cite{KLVY22} can be thought of as a proof of quantumness with two provers adapted to work with one prover; the single prover is forced to act as the two provers would, by sending the inputs in two separate rounds, and using fully homomorphic encryption (which requires the LWE hardness assumption to construct).

\subsection{Our results}
\subsubsection{Direct product theorem}
Let $\sfV(a^1\ldots a^l,x^1\ldots, x^l)$ be a predicate on $(\clA^1\times\ldots\times\clA^l)\times(\clX^1\times\ldots\times\clX^l)$. We shall use $\sfV^n(a^1_1\ldots a^l_1\ldots a^1_n\ldots a^l_n,x^1_1\ldots x^l_1\ldots x^1_n\ldots x^l_n)$ to denote $n$ independent copies of $\sfV$, i.e., the predicate which is satisfied when all $n$ $(a^1_i\ldots a^l_i, x^1_i\ldots x^l_i)$-s satisfy $\sfV$.

For a probability distribution $p$ on $\clX^1\times\ldots\times\clX^l$, a (quantum) communication protocol $\clP$ between $l$ parties that takes inputs from $\clX^1\times\ldots\times\clX^l$ and produces outputs in $\clA^1\times\ldots\times\clA^l$, produces a conditional probability distribution on $\clA^1\times\ldots\times\clA^l$ conditioned on $\clX^1\times\ldots\times\clX^l$, and along with $p$ there is an induced distribution on $(\clA^1\times\ldots\times\clA^l)\times(\clX^1\times\ldots\clX^l)$. Let $\suc(p, \sfV, \clP)$ be the probability that the predicate $\sfV$ is satisfied according to this distribution.

Let $\eff_\eps(\sfV,p)$ denote the distributional quantum partition bound with error $\eps$ for $\sfV$ with respect to input distribution $p$, which we shall define formally in Section \ref{sec:part-bound}. $\eff_\eps(\sfV,p)$ is a lower bound on the quantum communication complexity of $\sfV$. Let $\omega^*(G(p,\sfV))$ denote the quantum value of the non-local game $G = (p,\clX^1\times\ldots\times\clX^l,\clA^1\times\ldots\times\clA^l,\sfV)$.

With this notation, our direct product theorem is stated below.
\begin{restatable}{theorem}{dpt}
\label{thm:dpt}
For any $\eps, \zeta > 0$, any predicate $\sfV$ on $(\clA^1\times\ldots\times\clA^l)\times(\clX^1\times\ldots\times\clX^l)$ and any product probability distribution $p$ on $\clX^1\times\ldots\times\clX^l$, if $\clP$ is an interactive entanglement-assisted quantum communication protocol between $l$ parties which has total communication $cn$. 
\begin{enumerate}[(i)]
\item If $c < 1$, then
\[ \suc(p^n, \sfV^n, \clP) \leq \left(1-\frac{\nu}{2} + \sqrt{2lc}\right)^{\Omega\left(\nu^2n/\left(l^2\cdot\log(|\clA^1|\cdot\ldots\cdot|\clA^l|)\right)\right)}\]
where $\nu=1-\omega^*(G(p,\sfV))$.
\item If $1 \leq c = \delta\cdot\frac{\zeta^2}{l^3}\eff_{\eps+\zeta}(\sfV,p)$ for small enough $\delta$, then
\[\suc(p^n, \sfV^n,\clP) \leq (1-\eps)^{\Omega\left(n/\left(\log(|\clA^1|\cdot\ldots\cdot|\clA^l|)\right)\right)}.\]
\end{enumerate}
\end{restatable}

The two cases in Theorem \ref{thm:dpt} should be interpreted as follows: $c < 1$ means there is less than one qubit of communication per copy of $\sfV$, and we are close to the non-local game situation where there is no communication. Therefore we get an upper bound on the success probability for computing $\sfV^n$ in terms of the winning probability of the corresponding game. The theorem in this case is essentially saying that parallel-repeated non-local games under product distributions are resistant to communication, i.e., if the winning probability of $n$ copies of the game goes does exponentially in $n$, then it also goes down exponentially in $n$ if there is a small amount of communication. We remark however that the result for case (i) as stated here is not necessarily optimal. It is possible to upper bound the winning probability with $cn$ communication by $2^{cn}$ times the winning probability without communication, which is simply parallel repetition value of $G(p,\sfV)$, $\omega^*(G^n(p,\sfV))$, with a somewhat different argument (which we describe in some of the cryptographic applications). For product $p$, $\omega^*(G^n(p,\sfV))$ is upper bounded by $(1-\nu/2)^{\Omega(\nu^2n/l^2\cdot\log(|\clA^1|\cdot\ldots\cdot|\clA^l|))}$ \cite{JPY14}, which results in a better bound. However, we have presented this result as is in order to have a unified framework for cases (i) and (ii).

The case $c \geq 1$ means on average at least one qubit is communicated per copy of $\sfV$. This corresponds to the true communication scenario, and thus if $c$ is less than a lower bound on the per copy communication complexity of $\sfV$, we get that the probability of success for computing $\sfV^n$ goes down exponentially in $n$.  By Yao's Lemma, case (ii) of Theorem \ref{thm:dpt} has the following corollary for communication complexity.

\begin{cor}\label{cor:yao-dpt}
For a predicate $\sfV$ on $(\clA^1\times\ldots\times\clA^l)\times(\clX^1\times\ldots\times\clX^l)$, let $\Q_\eps(\sfV)$ denote the interactive entanglement-assisted quantum communication complexity of computing it, and $\eff_{\eps+\zeta}(\sfV,p)$ be its distributional quantum partition bound for distribution $p$, and any $\eps,\zeta > 0$. Then,
\[ \Q_{1-(1-\eps)^{\Omega\left(n/\left(\log(|\clA^1|\cdot\ldots\cdot|\clA^l|)\right)\right)}}(\sfV^n) = \Omega\left(\frac{\zeta^2n}{l^3}\left(\max_{\mathrm{product } \, p}\log\eff_{\eps+\zeta}(\sfV,p)\right)\right).\]
\end{cor}

Corollary \ref{cor:yao-dpt} is a strong direct product theorem for predicates whose interactive entanglement-assisted communication complexity is characterized by $\max_{\text{product }p}\eff_\eps(\sfV,p)$. This result can be seen as a quantum version of the result of \cite{JY12}, who proved a direct product theorem for randomized communication complexity in terms of the smooth rectangle bound, which is a relaxation of the classical partition bound (although we note that the result of \cite{JY12} worked for all distributions rather than only product ones).

\subsubsection{Applications in two-party communication complexity of functions}
In the communication complexity setting for a two-input function or relation $f\subseteq \clX\times\clY\times\clZ$, we normally require that only one party gives an output. Nevertheless, we can define a predicate $\sfV_f$ for it in which one party has a singleton output set, say $\{\top\}$, and the other party's output set is $\clZ$.  We define
\[ \sfV_f(\top z, xy) = 1 \quad \iff \quad z\in f(x,y).\]
It is clear then that the two-party communication complexity of $f$ is equal to the communication complexity of $\sfV_f$.

In \cite{ABJO21}, it is shown that a large class of functions exists, whose quantum communication complexity is characterized by $\eff_\eps(\sfV_f,p)$ for a product $p$. In particular, they show that a class of functions known as two-wise independent functions, $\eff_\eps(\sfV_f,p)$ takes the maximum possible value of the uniform distribution, which is product.
\begin{fact}[\cite{ABJO21}]
Let $f:\clX\times\clY\to\clZ$ be a two-wise independent function with $|\clX|=|\clY|$, and let $p_U$ be the uniform distribution on $\clX\times\clY$. Then for any $\eps>0$,
\[ \eff_\eps(\sfV_f,p_U) \geq \frac{|\clX|}{|\clZ|}\left(1-\gamma-\frac{1}{|\clZ|}\right)^2.\]
\end{fact}
An example of a two-wise independent function is the generalized inner product $\IP_q^n: \bbF_q^n\times\bbF_q^n \to \bbF_q$ defined by:
\[ \IP_q^n(x,y) = \sum_{i=1}^nx_iy_i \mod q.\]
This makes our result the first strong direct product theorem for generalized inner product that we are aware of. The direct product theorem in terms of the generalized discrepancy method by Sherstov \cite{She12} works only for boolean-output functions, and gives a strong direct product theorem for quantum communication of $\IP_2^n$.

For further comparison between our direct product theorem and Sherstov's, we prove Theorem \ref{thm:gamma_2-eff}. For a total function $f : \clX\times\clY \to \{-1,+1\}$, let $F$ denote the $|\clX|\times|\clY|$ matrix whose $[x,y]$-th entry is given by $f(x,y)$. The generalized discrepancy method lower bounds communication in terms of $\log\gamma^\alpha_2(F)$, where $\gamma^\alpha_2(M)$ is the $\alpha$-approximate factorization norm of a matrix $M$. For a function $f$, $\gamma^\alpha_2(F)$ can be expressed as $\max_p\gamma^\alpha_2(F,p)$ where $\gamma^\alpha_2(F,p)$ is a distributional version of $\gamma^\alpha_2(F)$ with respect to $p$ over $\clX\times\clY$.
\begin{restatable}{theorem}{gamm}
\label{thm:gamma_2-eff}
For a total function $f:\clX\times\clY\to\{-1,+1\}$, let $\sfV_f$ denote the predicate on $(\{-1,+1\})^2\times (\clX\times\clY)$ given by
\[ \sfV(ab,xy) = 1 \quad \iff \quad a\cdot b = f(x,y).\]
Then for any distribution $p$ on $\clX\times\clY$,
\[ \eff_\eps(\sfV_f,p) \geq (1-2\eps)\gamma_2^{\alpha}(F,p)\]
with $\alpha = \frac{1+2\eps}{1-2\eps}$.
\end{restatable}
This shows that $\eff_\eps(\sfV_f,p)$ is a stronger lower bound technique than $\gamma_2(F,p)$ for boolean $f$. However, since our direct product theorem is in terms of $\max_{\text{product } p}\log\eff_\eps(\sfV_f,p)$, and Sherstov's in terms of $\max_p\log\gamma_2^{}(F,p)$, the two results cannot be directly compared.

\subsubsection{DIQKD and DIRE secure against leakage}\label{sec:QKDRE-results}
In DI protocols that involve playing multiple copies of a non-local game, two settings can be considered: sequential and parallel. We prove that it is possible to do DIQKD in the parallel setting and DIQKD and DIRE in the sequential setting in the presence of leakage. Our result in the parallel setting is proved by applying case (i) of Theorem \ref{thm:dpt} to the parallel DIQKD protocol of \cite{JMS17} --- the original security proof of \cite{JMS17} had been done by using parallel repetition of the corresponding game without leakage. It is often possible to get a DIRE protocol by making some minor changes to a DIQKD protocol, and very similar security proofs work for both. However, we are unable to get a security proof for DIRE in the parallel setting with leakage. In fact, there is no known security proof for DIRE in the parallel setting even without leakage. The \cite{JMS17} protocol uses a lot of private randomness, and the key rate obtained from it is quite low with or without leakage, and therefore it cannot be modified to give a DIRE protocol.

To mitigate this situation, we study DIQKD and DIRE with leakage in the sequential setting. The sequential setting is somewhat easier to analyze than the parallel setting, and a number of approaches have been developed to prove security for DIQKD and DIRE in this setting \cite{VV19, DFR17, ADF+18}. However, none of these approaches seem to be possible to generalize to the sequential setting with leakage. Our result in this case is obtained by using a new (though quite straightforward) approach to do sequential security proofs, which has some similarities with the proof approach for DIRE in \cite{PM13} (although the proof in that work only works for adversaries that have classical side information --- our approach can handle quantum side information as well). Leakage can be incorporated in this approach by using an argument inspired by the quantum efficiency lower bound for quantum communication (given in Lemma \ref{lem:eff-lb}) to handle communication. This is the same argument that can potentially be used to improve the upper bound in case (i) of Theorem \ref{thm:dpt}, although in this case we need to make sure that the sequential structure of the protocol is preserved. The key rates we get in the sequential setting are much better than that in the parallel setting. To demonstrate the usefulness of our sequential approach, we later qualitatively compare the key rates obtained by it in the setting \emph{without} leakage, to that obtained by the Entropy Accummulation Theorem (EAT) \cite{DFR17, ADF+18}, which is the most general and widely used approach in the sequential setting.

\paragraph{Leakage in parallel DIQKD.} In the parallel device-independent setting, each honest party's device is modelled as a single black box, into which the party provides inputs and from which they get outputs to play several copies of a non-local game. The parties may only enter inputs for all $n$ copies of the game at once in this setting, which means there may be arbitrary correlations between the inputs to the $i$-th game and the outputs of the $j$-th game, for any $i, j$. Ideally the boxes play $n$ independent copies of the non-local game, although they may do so noisily, i.e., each game is won with probability $\delta$-close to its optimal quantum value. For DIQKD, the honest parties are Alice and Bob and we assume their boxes are supplied by the eavesdropper Eve. The states and measurements implemented by these boxes may be very far from those corresponding to the two-player non-local game that each of Alice and Bob's boxes ideally play. In fact, instead of Alice and Bob sharing an entangled state that is uncorrelated with anything else, Eve may hold a purification of Alice and Bob's state, which we also model as a box.

As mentioned before, known DIQKD protocols rely on the assumption that Alice and Bob and Eve's boxes do not communicate with each other. We relax the assumption in a strong way: we assume Alice, Bob and Eve's boxes can all send classical messages to each other (since they share entanglement, this means they can also effectively exchange quantum states via teleportation) after Alice and Bob have entered their inputs into their boxes and before they receive their outputs. The communication between Alice, Bob and Eve's boxes may be arbitrarily interactive: we do not put any bound on the number of rounds of communication, only on the total number of bits communicated. Since all the inputs are entered at once, all messages communicated in the parallel leakage model can depend on all the inputs of the device that is leaking the message.

\begin{remark}
In practice Alice and Bob’s boxes can also continue sending messages to Eve after their outputs are produced (there can also be communication to Alice and Bob’s boxes, but the key rate depends on Eve’s probability of guessing Alice and Bob’s outputs, which cannot change due to communication to Alice and Bob’s boxes after they have produced their outputs, so we ignore these at this time). But as far as security analysis is concerned, this communication is equivalent to Eve gaining some information about Alice and Bob’s outputs after they have been produced, e.g. from communication between Alice and Bob over a public channel, which is a standard part of QKD protocols and can be handled by standard DIQKD proof techniques. Using standard techniques, the amount of communication after the outputs are produced would just be subtracted from the key rate, and after a certain threshold of communication, key rate would just be zero. Communication before Alice and Bob’s outputs are produced cannot be handled by standard techniques, however, and hence we focus on the leakage model described above in our work.
\end{remark}

For the sake of concreteness, we analyze the parallel DIQKD protocol given by \cite{JMS17}, based on the Magic Square non-local game, under this leakage model, but in principle the same analysis could be applied to any DIQKD protocol that is based on a non-local game that has: (i) a product input distribution, and (ii) a common bit that Alice and Bob can ideally both know given their outputs $a$ and $b$, and both parties' inputs $x$ and $y$ (and this bit is their shared key). Using case (i) of Theorem \ref{thm:dpt}, we prove the following theorem.
\begin{restatable}{theorem}{qkd}
\label{thm:leaky-qkd}
If the \cite{JMS17} DIQKD protocol (given in Protocol \ref{prot:QKD}) is carried out with boxes that are compatible with $n$ parallel copies of the Magic Square game, and have at most $\delta$ noise in the honest case, for $n = \Omega\left(\frac{1}{\delta^2\alpha\gamma}\log(1/\lambda)\right)$, it is possible to extract $r^{\mathrm{QKD}}_{\mathrm{par}}(\delta,c)n - \log(1/\lambda) - O(1)$ bits of secret key that are $\lambda$-secure, in the interactive leakage model, with the total communication between Alice, Bob and Eve's boxes being $cn$ bits, for
\[r^{\mathrm{QKD}}_{\mathrm{par}}(\delta,c) = \alpha\left(\nu - \beta(\sqrt{c}+ \sqrt{\alpha}) -4\delta - 2h(4\delta) - \gamma\right).\]
Here $\alpha$ and $\gamma$ are protocol parameters that can be optimized\footnote{Note that $\alpha$ and $\gamma$ appear in both the lower bound on $n$, and the key rate.}, and $\nu$ and $\beta$ are constants in $(0,1)$ (given by Fact \ref{fc:MSE-w} and Corollary \ref{cor:MSE-dpt}), and $h$ is the binary entropy function.
\end{restatable}
In order for $r^{\mathrm{QKD}}_{\mathrm{par}}(\delta, c)$ to be positive, one has to pick $\alpha = O(\nu^2)$, which gives $r^{\mathrm{QKD}}_{\mathrm{par}}(\delta, c) = \Omega(\nu^3) - O(\nu^2(\sqrt{c} + h(4\delta) + \gamma))$. For comparison, in the sequential protocol, we get the dependence on $\nu$ to be $\log(1/(1-\nu)) = \Omega(\nu)$ instead. As mentioned earlier, it is possible the bound on $\suc(p^n, \sfV^n, \clP)$ in case (i) of Theorem \ref{thm:dpt} can be slightly improved, and using this, $r^{\mathrm{QKD}}_{\mathrm{par}}$ could also be improved. In particular, it is possible to replace the factor of $-\alpha\beta\sqrt{c}$ with $-c$, which may be better for some parameter ranges.

\paragraph{Leakage in sequential DIRE and DIQKD.} DIRE has only has one honest party, but this party has two quantum devices, which are both supplied by Eve; for convenience we shall refer to the two devices as Alice and Bob.
In the sequential setting, the honest parties have separate black boxes corresponding to each copy of the non-local game and they are able to enter inputs into these one by one. In particular, the parties will enter inputs into the $i$-th boxes and receive outputs from them before entering inputs into the $(i+1)$-th boxes. This restricts the kinds of correlations between inputs and outputs that are possible: the shared state in the boxes after the $i$-th game does not depend on the $(i+1)$-th inputs, so it is possible to analyze the $(i+1)$-th game as a fresh single copy of the game.

In the sequential setting, the leakage can again be interactive, but it respects the sequential structure of the protocol. We shall assume that the $i$-th block of leakage happens after inputs are entered and before the outputs are produced for the $i$-th game, and clearly the bits leaked here cannot depend on the inputs for the $(i+1)$-th game. Note that for sequential RE the leakage is between the two devices of the single honest party and Eve's device, and in sequential QKD, the leakage is between the three parties Alice, Bob and Eve as before.

We give protocols for DIQKD and DIRE in the sequential setting, which can be seen as a sequential versions of the parallel DIQKD protocol in \cite{JMS17} (with the DIRE version having further modifications because it is RE and not QKD). We prove the following theorem about their security.
\begin{restatable}{theorem}{re}
\label{thm:leaky-re}
There is a sequential DIQKD protocol (given in Protocol \ref{prot:seq-QKD}) which if carried out with boxes that are compatible with $n$ sequential copies of the Magic Square game, and having $\delta$ noise in the honest case, for $n =\Omega\left(\frac{1}{\delta^2\gamma}\log(1/\lambda)\right)$, results in $r^{\mathrm{QKD}}_{\mathrm{seq}}(\delta, c)n - \log(1/\lambda) - O(1)$ bits of secret key that are $\lambda$-secure, in the presence of $cn$ bits of interactive sequential leakage, for
\[ r^{\mathrm{QKD}}_{\mathrm{seq}}(\delta, c) = \log\left(\frac{1}{1-\nu}\right) - 4\delta - h(4\delta) - \gamma - c, \]
where $h$ is the binary entropy function, $\nu$ is a constant given by Fact \ref{fc:MSE-w}, and $\gamma$ is a free parameter that can be optimized.

There is a sequential DIRE protocol (given in Protocol \ref{prot:RE}), which if carried out with boxes that are compatible with $n$ sequential copies of the Magic Square game, and having $\delta$ noise in the honest case, for $n =\Omega\left(\frac{1}{\delta^2}\log(1/\lambda)\right)$, is an $n\log 9 + \polylog(n) \to r^{\mathrm{RE}}_{\mathrm{seq}}(\delta,c)n - 2\log(1/\lambda) - O(1)$ RE protocol that is $\lambda$-secure, in the presence of $cn$ bits of interactive sequential leakage, for
\[ r^{\mathrm{RE}}_{\mathrm{seq}}(\delta, c) = \log\left(\frac{1}{1-\nu}\right) - 2\delta - c. \]
\end{restatable}
The $\nu$ in the above theorem is the same $\nu$ that appears in Theorem \ref{thm:leaky-qkd}. However, since the proof approach we use here is slightly different, it may be possible to slightly improve this constant. We have not attempted to do this in this work, but we shall comment more later on what the improvement could be.

\paragraph{Comparison to standard sequential approaches without leakage.} The key rates for our sequential protocols \emph{without} leakage can simply be obtained by setting $c=0$.\footnote{Here we are informally using `key rate' to refer to both the number of bits of secret key obtained in QKD, and the number of \emph{new} random bits obtained in RE, as a function of $n$ and $\lambda$, as the key rate. Conventionally, the key rate is defined to be this number divided by $n$, in the large-$n$, small-$\lambda$ limit.} We shall use $r^{\mathrm{QKD}}_{\mathrm{seq}}(\delta)$ and $r^{\mathrm{RE}}_{\mathrm{seq}}(\delta)$ to refer to the corresponding values of $r^{\mathrm{QKD}}_{\mathrm{seq}}(\delta, c)$ and $r^{\mathrm{RE}}_{\mathrm{seq}}(\delta, c)$. We now compare these key rates to those obtained by the Entropy Accummulation approach without leakage. To do this, we first give some intuition for what the quantity $\nu$ that appears in our key rates is. Our security proof is done by considering a 3-player game between Alice, Bob and Eve in which Alice and Bob get the same inputs as in the standard Magic Square game, and Eve gets both of their inputs and has to guess the bit that is always equal for the both of them. For technical reasons, the game we actually analyze is more complicated than what we just described, but if the winning probability of the game just described is $1-\nu'$, then $\nu$ in Theorem \ref{thm:leaky-re} is $\frac{\nu'}{2}$. In Protocol \ref{prot:seq-QKD}, the raw secret key for each sequential round is set to be the bit that is supposed to be equal for Alice and Bob in the Magic Square Game. Therefore, $1-2\nu$ is the maximum probability that Alice and Bob win a single copy of the Magic Square game with their boxes, and Eve guesses the secret key bit. As stated earlier, it may be possible to replace $\nu$ with a slightly larger constant. Specifically, we could instead consider the winning probability of the 3-player game with the following constrained strategy: Eve has the canonical purification of Alice and Bob's state, and does the same measurement as them to guess the common bit. This constrained winning probability is a natural interpretation of the Renyi-2 entropy of the common bit given Eve's system, which is the quantity we work with in our proof. If this constrained winning probability is $1-2\kappa$, then we could potentially replace $\log(1/(1-\nu))$ in our key rate with $\log(1/(1-\kappa))$.

Entropy Accumulation, as far as we know, has not been used to analyze the specific sequential protocols that we consider in this work. In order to compare our approach to the EAT approach, we shall compare the key rates obtained from our approach and the EAT approach for Protocols \ref{prot:seq-QKD} and \ref{prot:RE}. However, there are some caveats to this. Protocols \ref{prot:seq-QKD} and \ref{prot:RE} are different from protocols usually analyzed by the EAT in an important way: in the latter protocols, the non-local game is played with its actual input distribution in only a small fraction $\gamma$ of the boxes, on which the testing happens, and deterministic inputs are supplied in the rest of the boxes. Our approach does not seem to be able to handle such protocols --- we require that the non-local game is played with its correct input distribution on all boxes. This is not a problem for QKD, since private randomness is essentially free in QKD, but in randomness expansion, our approach requires many more bits of seed randomness to be put into the protocol per new bit of randomness gained. Therefore, although the key rates (as a function of $n$ and $\lambda$) obtained from our approach and the EAT approach may be comparable for the RE protocol, it is fair to say that Entropy Accumulation has an advantage, since it can also handle protocols where fewer bits of seed randomness are needed to produce the same key rate. Of course, as stated before, our approach does have the advantage that it can be generalized to work in the setting with leakage, which Entropy Accumulation does not seem to be able to handle.

The number of bits of secret key obtained from a protocol using the EAT approach can be expressed in a generic way regardless of the specific game used in a protocol. We clarify here that we are going to be using the original EAT from \cite{DFR17}. There is a recent more generalized version of the EAT \cite{MFSR22} which works for general types of protocols, but it provides no advantages over the original EAT in analyzing our protocol. For Protocol \ref{prot:seq-QKD}, it gives the number of bits of secret key extracted to be
\[ r_{\mathrm{EAT}}(\delta)n - O\left(\sqrt{n\cdot\log(1/\lambda)}\right) - O(\log(1/\lambda)).\]
The constants in the two $O$-s in the above expression are large (whereas the constant in the $O(1)$ for our protocol is $\sim 10$), so it is easy to see that the behaviour w.r.t. $\lambda$ is better for our proof approach. We now compare the behaviour w.r.t. $n$. The function $r_{\mathrm{EAT}}$ is given by
\[ r_{\mathrm{EAT}}(\delta) = \min_{\substack{\sigma: \text{ Alice and Bob win} \\ \MS \text{ with probability } 1-2\delta}}\sfH(K|\tE)_\sigma - \gamma - h(4\delta).\]
The first term in the above expression is the minimization of a conditional von Neumann entropy  over tripartite states $\sigma$ that are shared between Alice, Bob and Eve, with $K$ being Alice's raw key bit, and $\tE$ being the quantum register held by Eve. In particular, we minimize over all such states that allow Alice and Bob to win a single copy of the Magic Square game with probability at least $1-2\delta$, and the register $K$ holds the output bit of Alice in the game which is supposed to be equal to one of Bob's output bits. Note that the $-(\gamma + h(4\delta))$ is identical in $r^{\mathrm{QKD}}_{\mathrm{seq}}(\delta)$ and $r_{\mathrm{EAT}}(\delta)$, so we need to compare $(\log(1/(1-\nu)) - 4\delta)n$ with $(\min_\sigma\sfH(K|\tE)_\sigma) n - O(\sqrt{n})$.

The $(\min_\sigma\sfH(K|\tE)_\sigma)$ term is somewhat hard to interpret, but $\sfH(K|\tE)_\sigma$ is lower bounded by the corresponding conditional min-entropy $\sfH_{\min}(K|\tE)_\sigma$. Min-entropy has the operational interpretation of being the log of the inverse of a guessing probability, so $(\min_\sigma\sfH_{\min}(K|\tE)_\sigma)$ here is the log of the inverse of Eve's maximum guessing probability for a single bit of the key, on a state with which Alice and Bob win a single copy of Magic Square with probability $1-2\delta$. Recalling that $1-2\nu$ is the maximum overall probability of Alice and Bob winning a single copy of Magic Square and Eve guessing the key bit, we think $\log(1/(1-\nu)) - 4\delta$ is comparable in value to $\min_\sigma\sfH_{\min}(K|\tE)_\sigma$. Since $\min_\sigma\sfH(K|\tE)_\sigma$ is bigger than $\min_\sigma\sfH_{\min}(K|\tE)_\sigma$, but the EAT bound also has the $-O(\sqrt{n})$ term (with a large constant in the $O$), we think our bound is better for smaller values of $n$, but EAT outperforms it for large $n$. A more exact comparison can be made by numerically computing the values of $\min_\sigma\sfH(K|\tE)_\sigma$ (this value is usually computed numerically in analyses done with the EAT) and $\nu$ (although there is an analytic estimate for $\nu$ given in \cite{JMS17}, it is quite small, and numerical estimates perform much better for these purposes).

The key rate given by the EAT for RE is similar, except without the $-\gamma - h(4\delta)$ terms, which our bound also does not have. Therefore, a similar comparison holds for RE. However, as noted before, EAT-based approaches can perform much better in terms of number of bits of seed randomness used, compared to our approach.

\subsection{Proofs of quantumness with two players secure against leakage}
The only disadvantage of proofs of quantumness with two provers compared to those with one prover is the fact that we have to assume there is no communication between two provers; as we discussed earlier, there are in fact several advantages to the approach with two provers. Our final application is helping relax the assumption in proofs of quantumness with two provers. Since the provers in this setting are potentially dishonest, their devices need not be modelled as black boxes here, and "leakage" here is actually intentional communication between two classical provers who are trying to convince the verifier that they are quantum. The communication between the two provers may be arbitrarily interactive as before. Moreover, since the protocol is 2-round, the verifier gives the inputs for all copies of the game to the provers at once, and all the messages leaked between the provers can depend on all inputs. But we need not worry about leakage to or from the verifier in this case, since the verifier is honest and does not have any quantum devices.

If the provers are allowed to communicate an arbitrary amount, it is of course possible for classical players to always win the non-local game and hence convince the verifier that they are quantum. However, in the case where the interactive leakage is bounded, we prove the following theorem.
\begin{restatable}{theorem}{quantum}
\label{thm:leaky-quantum}
There is an interactive (2-round) proof of quantumness protocol (given in Protocol \ref{prot:quantum}) between a verifier and two provers which is $\lambda$-secure against classical provers who communicate that most $C(\lambda) = O(\log(1/\lambda))$ bits.
\end{restatable}
Theorem \ref{thm:leaky-quantum} is proved using a parallel repetition theorem for the classical value of 2-player non-local games, and a classical version of the argument used in proving Theorem \ref{thm:leaky-re} (which was used to show that the classical partition bound lower bounds classical communication). Since we use the Magic Square non-local game for our other applications, we also use that game in Protocol \ref{prot:quantum} for convenience. But no properties of the Magic Square game are actually used in the protocol, and we only need parallel repetition to hold for the classical value of the game. Since parallel repetition holds for the classical value of all 2-player games, a proof of quantumness could be done with any such game.

Aside from the specific results for DIQKD, DIRE, and proofs of quantumness, our proof technique can be seen as a general framework for making device-independent protocols that prove security using $n$ copies of a non-local game, secure against leakage. For example, this technique can also be applied to the device-independent protocol for encryption with certified deletion given by \cite{KT20}. The security proof for that protocol uses a parallel repetition theorem for an anchored two-round game (where players receive two rounds of inputs and give two rounds of outputs). A version of Theorem \ref{thm:dpt} in case (i) also applies to anchored distributions instead of product distributions for one-round games, and it is not difficult to generalize to two-round games by considering an appropriate round-by-round leakage model.

\subsection{Organization of the paper}
In Section \ref{sec:overview} we give an overview of our proofs. In Section \ref{sec:prelim} we provide definitions and known results about the quantities used in our proofs. In Section \ref{sec:part-bound}, we introduce variants of the quantum partition bound, prove that they lower bound communication and also Theorem \ref{thm:gamma_2-eff}. In Section \ref{sec:perturb-D}, we prove a lemma called the Substate Perturbation Lemma, which is a main tool for our direct product theorem. In Section \ref{sec:dpt-proof}, we give the proof of our main direct product theorem. Finally, in Section \ref{sec:QKD} we show the applications of our direct product theorem to prove security of DIQKD, DIRE and proofs of quantumness with leakage.

\section{Proof overview}\label{sec:overview}
\subsection{Direct product theorem}
We follow the information-theoretic framework for parallel repetition and direct product theorems introduced by \cite{Raz95} and \cite{Hol09}. The idea is this: take a protocol $\clP$ for $\sfV^n$ that is ``too good''. We condition on the success in some $t$ coordinates in this protocol, and show that either the probability of success in these coordinates is already small, or there is an $i$ in the other $n-t$ coordinates such that the probability of success of $i$ conditioned on success event $\clE$ is bounded away from 1. This is done by showing that if the probability of $\clE$ and the probability of success in $i$ conditioned on $\clE$ are both large, we can give a protocol $\clP'$ for $\sfV$ that is ``too efficient''. Now our lower bound in the $c \geq 1$ case is in terms of $\eff_\eps(\sfV,p)$, which intuitively speaking, corresponds to the inverse of the maximum probability of not aborting in a zero-communication protocol in which the $l$ parties either abort, or produce outputs that satisfy $\sfV$ with probability at least $1-\eps$ (conditioned on not aborting). Therefore, $\clP'$ for us will be a zero-communication protocol with aborts that computes $\sfV$ with high probability conditioned on not aborting, whose probability of not aborting is too high.

For simplicity, we shall give an overview of the proof with only two parties Alice and Bob; the proof for $l$ parties follows similarly. When Alice and Bob's inputs are $x_i$ and $y_i$ respectively at the $i$-th coordinates in $\clP$, we define a state $\ket{\vph}_{x_iy_i}$ that represents the state at the end of $\clP$ conditioned on $\clE$. Considering the state at the end instead of round by round is the same approach as that taken in \cite{JRS03}, who use it to show a direct sum theorem. On input $(x_i,y_i)$ in $\clP'$, Alice and Bob will try to either abort, or get a shared state close to $\ket{\vph}_{x_iy_i}$. Once they have this state, they can perform measurements on the $i$-th output registers to give their outputs $(a_i,b_i)$. Their output distribution will be close to the output distribution in the $i$-th coordinate of $\clP$ conditioned on $\clE$; hence if the probability of success on $i$ conditioned on $\clE$ is too large, the probability of Alice and Bob correctly computing $\sfV$ in $\clP'$ conditioned on not aborting is also large. Hence our proof mainly consists of showing how Alice and Bob can get the shared state close to $\ket{\vph}_{x_iy_i}$ with probability of aborting $2^{-O(c)}$, where $cn$ is the communication in $\clP$. Since the probability of aborting in $\clP'$ cannot be smaller than $\eff$, this gives the desired lower bound on the communication of $\clP$ in terms of $\eff$.

In the $c < 1$ case, our proof is very similar to the proof of a parallel repetition theorem for non-local games with product distributions due to \cite{JPY14}. The main difference between that $c \geq 1$ case and the parallel repetition of $c < 1$ case is that in the latter, we need to show that Alice and Bob can get the shared state $\ket{\vph}$ by local unitaries (without aborting).  We briefly describe their proof below.

\paragraph{Parallel repetition for games under product distribution.} Let $\ket{\vph}_{x_i}$ be the superposition of $\ket{\vph}_{x_iy_i}$ over the distribution of $Y_i$, $\ket{\vph}_{y_i}$ be the superposition over the distribution of $X_i$,  and $\ket{\vph}$ be the superposition over both. If the probability of $\clE$ is large, then conditioning on it, the following can be shown:
\begin{enumerate}
\item By chain rule of mutual information, there is an $X_i$ whose mutual information with Bob's registers in $\ket{\vph}$ is small. Hence by Uhlmann's theorem, there exist unitaries $U_{x_i}$ acting on Alice's registers that take $\ket{\vph}$ close to $\ket{\vph}_{x_i}$.
\item Similarly, the mutual information between $Y_i$ and Alice's registers in $\ket{\vph}$ is small, and hence there exist unitaries $V_{y_i}$ acting on Bob's registers that take $\ket{\vph}$ close to $\ket{\vph}_{y_i}$.
\item By applying the quantum operation that measures the $X_i$ register and records the outcome, it can be shown that $V_{y_i}$ also takes $\ket{\vph}_{x_i}$ to $\ket{\vph}_{x_iy_i}$.
\item Since $U_{x_i}$ and $V_{y_i}$ act on disjoint registers, $U_{x_i}\otimes V_{y_i}$ then takes $\ket{\vph}$ close to $\ket{\vph}_{x_iy_i}$.
\end{enumerate}
Alice and Bob can thus share $\ket{\vph}$ as entanglement, and get close to $\ket{\vph}_{x_iy_i}$ by local unitariess $U_{x_i}$ and $V_{y_i}$. In case (i) of our proof, everything is similar to this, except that the distance between $\ket{\vph}$ and $\ket{\vph}_{x_i}$ also accounts for $c^\A$, $c^\A n$ being Alice's total communication to Bob, and the distance between $\ket{\vph}$ and $\ket{\vph}_{y_i}$ also accounts for $c^\B$, $c^\B n$ being Bob's communication.

If we wish a give a proof for case (i) with anchored distributions instead of product distributions, we would need to follow the equivalent steps in the proof of the parallel repetition theorem for anchored games given in \cite{BVY17} or the alternative proof given in \cite{JK20} instead, and account for communication there.

\paragraph{Direct product for communication under product distribution.} In case $c \geq 1$, we cannot use Uhlmann unitaries to go from $\ket{\vph}$ to $\ket{\vph}_{x_i}$ and $\ket{\vph}_{y_i}$, as there is a lot of dependence between Alice's registers and Bob's registers due to communication. But we can use a compression scheme due to \cite{JRS02,JRS03} which says that if the mutual information between $X_i$ and Bob's registers is $c$, then there exist measurement operators $M_{x_i}$ acting on Alice's registers which succed on $\ket{\vph}$ with probability $2^{-c}$, and on success take it close to $\ket{\vph}_{x_i}$. Following parallel repetition proof we can show:
\begin{enumerate}
\item If the total communication from Alice to Bob in $\clP$ is $c^\A n$, then the mutual information between $X_1\ldots X_n$ and Bob's registers in $\ket{\vph}$ is $O(c^\A n)$. By chain rule of mutual information, there exists an $i$ such that the mutual information between $X_i$ and Bob's registers is $O(c^\A)$, and hence there exist measurement operators $M_{x_i}$ acting on Alice's registers which succeed with probability $2^{-O(c^\A)}$ on $\ket{\vph}$ and on success take $\ket{\vph}$ close to $\ket{\vph}_{x_i}$.
\item Similarly, if the total communication from Bob to Alice in $\clP$ is $c^\B n$, then there exist measurement operators $N_{y_i}$ acting on Bob's registers which succeed with probability $2^{-O(c^\B)}$ on $\ket{\vph}$ and on success take $\ket{\vph}$ close to $\ket{\vph}_{y_i}$.
\item By applying the same argument with the operation measuring the $X_i$ register and recording the outcome, it can be shown that $N_{y_i}$ succeeds on $\ket{\vph}_{x_i}$ with probability $2^{-O(c^\B)}$ and on success takes it close to $\ket{\vph}_{x_iy_i}$.
\end{enumerate}
However, unlike in the case of unitaries, even though $M_{x_i}$ and $N_{y_i}$ commute, there is a problem in combining items 2 and 3 above to say that $M_{x_i}\otimes N_{y_i}$ succeed on $\ket{\vph}$ with probability $2^{-O(c^\A + c^\B)}$ and on success take it close to $\ket{\vph}_{x_iy_i}$. Since $\sqrt{\frac{1}{2^{-O(c^\A)}}}M_{x_i}\ket{\vph}$ (i.e., the normalized state on success of $M_{x_i}$ on $\ket{\vph}$) is only close to $\ket{\vph}_{x_i}$ rather than exactly equal to it, acting $N_{y_i}$ on this state cannot take it close to $\ket{\vph}_{x_iy_i}$, unless the distance between $\sqrt{\frac{1}{2^{-O(c^\A)}}}M_{x_i}\ket{\vph}$ and $\ket{\vph}_{x_i}$ is of the same order as the success probability of $M_{y_i}$ on $\ket{\vph}_{x_i}$. This distance figures in the exponent in the success probability $2^{-O(c^\A)}$, so we cannot afford to make it that small.

Instead we shall directly try to get projectors $N'_{y_i}$ that succeed with high probability on $\ket{\rho}$, which is we what we call the superposition over $X_i$ of $\sqrt{\frac{1}{2^{-O(c^\A)}}}M_{x_i}\ket{\vph}$, and on success take it close to $\ket{\vph}_{y_i}$ (these will also take $\ket{\rho}_{x_i}$ close to $\ket{\vph}_{x_iy_i}$). Since we do not have a bound on the mutual information between $Y_i$ and Alice's registers in $\rho$, we prove what we call the Substate Perturbation Lemma in order to do this. The quantity that is actually of relevance in the \cite{JRS03} compression scheme is the smoothed relative min-entropy $\sfD^\eps_\infty$ between $\vph_{Y_iA}$ and $\vph_{Y_i}\otimes\vph_A$ ($A$ being Alice's registers), which is $O(c^\B/\eps^2)$ if the mutual information between $Y_i$ and $A$ is $O(c^\B)$, due to the Quantum Substate Theorem \cite{JRS02, JRS09, JN12}. In the Substate Perturbation Lemma, which is one of our main technical contributions, we show that if $D^\eps_\infty(\vph_{Y_iA}\Vert\vph_{Y_i}\otimes\vph_A)$ is $c'$ and $\rho_A$ and $\vph_A$ are $\delta$-close, then $\sfD^{3\eps+\delta}_\infty(\vph_{Y_iA}\Vert\vph_{Y_i}\otimes\rho_A)$ is $O(c')$. Using the \cite{JRS03} compression scheme, this lets us get projectors $N'_{y_i}$ on Bob's registers that succeed with probability $2^{-O(c^\B)}$ on $\ket{\rho}$ and on success take it close to $\ket{\vph}_{y_i}$.

The protocol $\clP'$ will thus involve the following: Alice and Bob share $\ket{\vph}$ as entanglement and on inputs $(x_i,y_i)$, apply the measurements $\{M_{x_i}, \Id - M_{x_i}\}$ and $\{N'_{y_i},\Id - N'_{y_i}\}$ on it. They abort if the $M_{x_i}$ or $N'_{y_i}$ projector does not succeed. Since $M_{x_i}\otimes N'_{y_i}$ succeeds on $\ket{\vph}$ with probability $2^{-O(c^\A+c^\B)} = 2^{-O(c)}$, $\clP'$ does not abort with probability $2^{-O(c)}$ and on not aborting, gets a state close to $\ket{\vph}_{x_iy_i}$.

\subsection{Security of DIQKD and DIRE with leakage}
The \cite{JMS17} protocol is based on the Magic Square non-local game. In a single copy of the Magic Square game, henceforth denoted by $\MS$, Alice and Bob receive trits $x$ and $y$ and are required to output 3-bit strings $a$ and $b$ which respectively have even and odd parity; they win the game if their outputs satisfy the condition $a[y] = b[x]$. In the \cite{JMS17} protocol, Alice and Bob have boxes which are compatible with $n$ copies of $\MS$. Using trusted private randomness, Alice and Bob generate i.i.d. inputs $x_i,y_i$ for each game and generate outputs $a_i,b_i$. The inputs $x_i, y_i$ are then publicly communicated. Alice and Bob select a small random subset of $[n]$ to test the $\MS$ winning condition on, i.e., they check if $a_i[y_i]=b_i[x_i]$ for $i$ in that subset (up to error tolerance). If the test passes, they select $K^\A = (a_i[y_i])_i$ and $K^\B = (b_i[x_i])_i$ as their raw secret keys ; otherwise the protocol aborts. Due to error correction and privacy amplification, we can get a linear amount of secret key from this scheme if we can show \cite{Ren-th}
\[ \sfH^\eps_2(K^\A|\tE')_\rho - \sfH^\eps_0(K^\A|K^\B)_\rho = \Omega(n), \]
where $\sfH^\eps_2$ is the $\eps$-smoothed conditional Renyi-2 entropy and $\sfH^\eps_0$ is the $\eps$-smoothed conditional Hartley entropy, $\rho$ is the shared state of Alice, Bob and Eve conditioned on not aborting, and $\tE'$ is everything Eve holds at the end of the protocol, including a quantum purification of Alice and Bob's systems and also the classical information $X_iY_i$ that Alice and Bob have communicated publicly. $\sfH^\eps_2(K^\A|\tE')_\rho$ is lower bounded by the conditional min-entropy $\sfH^\eps_\infty(K^\A|K^\B)_\rho$, which is what we shall actually be working with for the parallel security proof. 

\paragraph{Challenges in a standard sequential security proof approaches.}

As stated before, the most widely-used tool in sequential security proofs is the Entropy Accumulation Theorem \cite{DFR17,AFRV19}. Suppose the information released to Eve when Alice and Bob play the $i$-th game is $T_i$, and Eve's quantum register is $\tE$. Then in order to apply the Entropy Accumulation Theorem to bound $\sfH_\infty^\eps(K^\A|T_1\ldots T_n\tE)_\rho$, we require the Markov condition $(A_1\ldots A_{i-1}) - (T_1\ldots T_{i-1}\tE)-T_i$ for all $i$, i.e., the information leaked in the $i$-th round is independent of the Alice's outputs of the rounds before $i$, given Eve's side information before the $i$-th round. In the setting without leakage, $T_i$ is just Alice and Bob's inputs $X_iY_i$ for the $i$-th round, which are picked with trusted private randomness, and thus can be made independent of everything else. In the setting with leakage however, $T_i$ would include the information leaked by Alice and Bob's boxes in the $i$-th round as well. Once we allow the boxes to leak information, there is nothing stopping them from leaking information about the outputs of the $(i-1)$-th round in the $i$-th round. Thus imposing the Markov condition here feels fairly unnatural, and closes off the possibility of using Entropy Accumulation in the model with leakage.

\paragraph{Parallel security proof.} Instead we closely follow the approach of \cite{JMS17} in giving a parallel security proof for their protocol, where no Markov condition is required.
The security proof of \cite{JMS17} is based on the parallel repetition theorem for non-local games under product distributions \cite{JPY14}. Since we are working in the setting with leakage, instead of a parallel repetition theorem for games, we use our direct product theorem for communication. The communication setting with 3 players exactly corresponds to the leakage model between the parties Alice, Bob and Eve in QKD. Case (i) of our direct product theorem says that if communication is $cn$ for for sufficiently small $c<1$, then the probability of computing $n$ copies of a non-local game's predicate correctly goes down exponentially in $n$.

The game we consider is a three-player version of $\MS$, which is a hybrid of the games considered by \cite{JMS17} and \cite{Vid17}, and this gives a simplified version of the \cite{JMS17} proof. In this game, which we call $\MSe$, Alice and Bob play $\MS$ between them, and in addition Eve, who has no input, has to guess both their inputs $x, y$, and Alice's output bit $a[y]$ (note that this makes the input ditribution product). Due to technical reasons, we also need to include the following feature in the game: Alice and Eve get additional independent input bits $z$ and $z'$, and Alice and Bob's winning condition $a[y]=b[x]$ not being satisfied is forgiven if $z=z'$, but we shall ignore the effects of introducing this condition for now. The winning probability of this game is strictly smaller than $\frac{1}{9}$ (which is Eve's probability of correctly guessing $x,y$). Due to our direct product result, in the presence of a bounded amount of communication before the outputs are produced, the winning probability of $n$ copies of this game is $\left(\frac{1}{9}(1-\nu)\right)^{\Omega(n)}$ for some $\nu > 0$.

Since Alice and Bob have performed the test to see that $a_i[y_i]=b_i[x_i]$ on a random subset, this condition is satisfied in most locations with high probability conditioned on not aborting. Therefore, $\MSe$ is won if Eve can correctly guess $x_i,y_i,a_i[y_i]$. Now, suppose $\vph_{K^\A K^\B X_1\ldots X_nY_1\ldots Y_n\tE}$ is the shared quantum state before $x_1\ldots x_n,y_1\ldots y_n$ are communicated, conditioned on not aborting\footnote{Alice and Bob cannot actually check the abort condition before $x_1\ldots x_n,y_1\ldots y_n$ are communicated, but the aborting condition is a well-defined event on $K^\A K^\B XY$ and thus can be conditioned on before this.}, with $\tE$ being Eve's quantum register. Operationally $\sfH_\infty(X_1\ldots X_nY_1\ldots Y_nK^\A|\tE)_\vph$ is the negative logarithm of Eve's probability of guessing $x_1\ldots x_ny_1\ldots y_nk^\A$, which is the probability of winning $n$ instances of $\MSe$, since Alice and Bob's winning condition is satisfied with high probability (adding smoothing to this definition allows us to consider the guessing probability in states close to $\vph$ instead). Hence by the direct product theorem, in the presence of a bounded amount of communication, $\sfH^\eps_\infty(X_1\ldots X_nY_1\ldots Y_nK^\A|\tE)_\vph$ is $\Omega(n(\log 9 + \log(1/(1-\nu)))$. By the chain rule of conditional min-entropy, this means that $\sfH^\eps_\infty(K^\A|XY\tE)_\vph$ is $\Omega(n\log(1/(1-\nu)))$. We remark that since our direct product theorem is not ``perfect'', i.e., the exponent we have is $\Omega(n)$ instead of $n$, we can only have Alice and Bob communicate a subset of $x_1\ldots x_ny_1\ldots y_n$ here instead of all of them (and $XY$ in the notation refers to the subset), and use those for key generation, so as not to make $\sfH^\eps_\infty(K^\A|XY\tE)_\vph$ negative.

In the actual state $\rho$ after $xy$ is released, Eve can do some local operations on $XY\tE$, but these do not change $\sfH^\eps_\infty(K^\A|XY\tE)_\vph$, and hence we have the same lower bound for $\sfH^\eps_\infty(K^\A|XY\tE)_\rho$. In order to upper bound $\sfH^\eps_0(K^\B|K^\A)_\rho$, we use the operational interpretation of $\sfH^\eps_0(K^\B|K^\A)_\rho$ as the maximum number of possible values of $K^\B$ given $K^\A$. As mentioned before, conditioned on not aborting, $K^\A$ and $K^\B$ differ in very few locations with high probability, and hence we can bound this quantity.

\begin{remark}
An alternate security proof of the \cite{JMS17} protocol was given in \cite{Vid17} by using the parallel repetition of anchored games instead of product games. A version of case (i) of Theorem \ref{thm:dpt} with anchored games could also be used to follow this proof instead, to prove security against leakage.
\end{remark}

In our sequential security proofs for DIQKD and DIRE, we shall actually need to work with $\sfH_2^\eps(K^\A|XY\tE)_\vph$ instead. $\sfH_2(K^\A|XY\tE)_\vph$ has the operational interpretation that it is Eve's guessing probability for $K^\A$, when she is constrained to holding a canonical purification of Alice and Bob's state and doing the same measurements as them. If Eve is doing the same measurements as Alice and Bob, that in particular means that she is acting sequentially --- if we were working with $\sfH_\infty$ instead, even though Alice and Bob were acting sequentially, we could not force Eve to. With this interpretation in mind, we need to upper bound the probability of winning $n$ sequential copies of $\MSe$ in the presence of leakage. Without leakage, this probability can be upper bounded by a much simpler argument with exactly $\left(\frac{1}{9}(1-\nu)\right)^n$ --- the improvement in key rate comes from the exponent being $n$ instead of $\Omega(n)$. Once we have this upper bound, it is easy to incorporate leakage in it as well (by the same argument which could improve case (i) of Theorem \ref{thm:dpt}. This is because we can get a protocol without leakage from a protocol with leakage by making each party guess the leaked transcript from shared randomness, and their guesses are correct with probability $2^{-cn}$. This factor in the guessing probability subtracts $cn$ from the key rate due to leakage.

\section{Preliminaries}\label{sec:prelim}
\subsection{Probability theory}
We shall denote the probability distribution of a random variable $X$ on some set $\clX$ by $\sfP_X$. For any event $\clE$ on $\clX$, the distribution of $X$ conditioned on $\clE$ will be denoted by $\sfP_{X|\clE}$. For joint random variables $XY$, $\sfP_{X|Y=y}(x)$ is the conditional distribution of $X$ given $Y=y$; when it is clear from context which variable's value is being conditioned on, we shall often shorten this to $\sfP_{X|y}$. We shall use $\sfP_{XY}\sfP_{Z|X}$ to refer to the distribution
\[ (\sfP_{XY}\sfP_{Z|X})(x,y,z) = \sfP_{XY}(x,y)\cdot\sfP_{Z|X=x}(z).\]
For two distributions $\sfP_X$ and $\sfP_{X'}$ on the same set $\clX$, the $\ell_1$ distance between them is defined as
\[ \Vert\sfP_X - \sfP_{X'}\Vert_1 = \sum_{x\in\clX}|\sfP_X(x) - \sfP_{X'}(x)|.\]

\begin{fact}
For joint distributions $\sfP_{XY}$ and $\sfP_{X'Y'}$ on the same sets,
\[ \Vert\sfP_X -  \sfP_{X'}\Vert_1 \leq \Vert\sfP_{XY} - \sfP_{X'Y'}\Vert_1.\]
\end{fact}
\begin{fact}\label{l1-dist}
For two distributions $\sfP_X$ and $\sfP_{X'}$ on the same set and an event $\clE$ on the set,
\[ |\sfP_X(\clE) - \sfP_{X'}(\clE)| \leq \frac{1}{2}\Vert\sfP_X - \sfP_{X'}\Vert_1.\]
\end{fact}
The following result is a consequence of the well-known Serfling bound.
\begin{fact}[\cite{TL17}]\label{fc:serfling}
Let $Z=Z_1\ldots Z_n$ be $n$ binary random variables with an arbitrary joint distribution, and let $T$ be a random subset of size $\gamma n$ for $0 \leq \gamma \leq 1$, picked uniformly among all such subsets of $[n]$ and independently of $Z$. Then,
\[ \Pr\left[\left(\sum_{i\in T}Z_i \geq (1-\eps)\gamma n\right) \land \left(\sum_{i\in[n]}Z_i < (1-2\eps)n\right)\right] \leq 2^{-2\eps^2\gamma n}.\]
\end{fact}

\subsection{Quantum information}
The $\ell_1$ distance between two quantum states $\rho$ and $\sigma$ is given by
\[ \Vert\rho-\sigma\Vert_1 = \Tr\sqrt{(\rho-\sigma)^\dagger(\rho-\sigma)} = \Tr|\rho-\sigma|.\]
The fidelity between two quantum states is given by
\[ \sfF(\rho,\sigma) = \Vert\sqrt{\rho}\sqrt{\sigma}\Vert_1 = \max_U\Tr(U\sqrt{\rho}\sqrt{\sigma}).\]
The purified distance based on fidelity is given by
\[ \Delta(\rho,\sigma) = \sqrt{1-\sfF(\rho,\sigma)^2}.\]
The Bures distance which is also based on fidelity is given by
\[ \sfB(\rho,\sigma) = \sqrt{1-\sfF(\rho,\sigma)}.\]

$\ell_1$ distance, $\Delta$ and $\sfB$ are all metrics that satisfy the triangle inequality.
\begin{fact}[Uhlmann's theorem]\label{fc:uhlmann}
Suppose $\rho$ and $\sigma$ are states on register $X$ which are purified to $\ket{\rho}_{XY}$ and $\ket{\sigma}_{XY'}$ with $Y$ ad $Y'$ not necessarily being of the same dimension, then it holds that
\[ \sfF(\rho, \sigma) = \max_U|\matel{\rho}{\Id_X\otimes U}{\sigma}|\]
where the maximization is over isometries taking $Y'$ to $Y$.
\end{fact}
\begin{fact}[Fuchs-van de Graaf inequality]\label{fc:fvdg}
For any pair of quantum states $\rho$ and $\sigma$,
\[ 2(1-\sfF(\rho,\sigma)) \leq \Vert\rho-\sigma\Vert_1\leq 2\sqrt{1-\sfF(\rho,\sigma)^2}.\]
Consequently,
\[ 2\sfB(\rho,\sigma)^2 \leq \norm{\rho-\sigma}_1 \leq 2\sqrt{2}\cdot\sfB(\rho,\sigma).\]
For two pure states $\ket{\psi}$ and $\ket{\phi}$, we have
\[ \Vert\state{\psi} - \state{\phi}\Vert_1 = \sqrt{1 - \sfF\left(\state{\psi},\state{\phi}\right)^2} = \sqrt{1-|\inprod{\psi}{\phi}|^2}.\]
\end{fact}
\begin{fact}[\cite{Tom16}]\label{F-concave}
The square of the fidelity is jointly concave in both arguments, i.e.,
\[ \sfF(\eps\rho + (1-\eps)\rho', \eps\sigma + (1-\eps)\sigma')^2 \geq \eps\sfF(\rho,\sigma)^2 + (1-\eps)\sfF(\rho',\sigma')^2.\]
\end{fact}
\begin{fact}[Data-processing inequality]\label{fc:chan-l1}
For a quantum channel $\clO$ and states $\rho$ and $\sigma$,
\[ \Vert\clO(\rho) - \clO(\sigma)\Vert_1 \leq \Vert\rho-\sigma\Vert_1 \quad \quad \text{and} \quad \quad \sfF(\clO(\rho),\clO(\sigma)) \geq \sfF(\rho,\sigma).\]
\end{fact}

The entropy of a quantum state $\rho$ on a register $Z$ is given by
\[ \sfH(\rho) = -\Tr(\rho\log \rho).\]
We shall also denote this by $\sfH(Z)_\rho$. For a state $\rho_{YZ}$ on registers $YZ$, the entropy of $Y$ conditioned on $Z$ is given by
\[ \sfH(Y|Z)_\rho = \sfH(YZ)_\rho - \sfH(Z)_\rho\]
where $\sfH(Z)_\rho$ is calculated w.r.t. the reduced state $\rho_Z$. The relative entropy between two states $\rho$ and $\sigma$ of the same dimensions is given by
\[ \sfD(\rho\Vert \sigma) = \Tr(\rho\log\rho) - \Tr(\rho\log\sigma).\]
The relative min-entropy between $\rho$ and $\sigma$ is defined as
\[ \sfD_\infty(\rho\Vert\sigma) = \min\{\lambda : \rho \leq 2^\lambda\sigma\}.\]
It is easy to see that for all $\rho$ and $\sigma$,
\[ 0 \leq \sfD(\rho\Vert\sigma) \leq \sfD_\infty(\rho\Vert\sigma).\]
\begin{fact}[Pinsker's inequality]\label{fc:pinsker}
For any two states $\rho$ and $\sigma$,
\[ \Vert\rho-\sigma\Vert_1^2 \leq 2\ln 2\cdot\sfD(\rho\Vert\sigma) \quad \text{ and } \quad \sfB(\rho,\sigma)^2 \leq \ln 2\cdot\sfD(\rho\Vert\sigma).\]
\end{fact}
\begin{fact}\label{fc:u-inv}
For any unitary $U$, and states $\rho, \sigma$, $\sfD(U\rho U^\dagger\Vert U\sigma U^\dagger) = \sfD(\rho\Vert\sigma)$, and $\sfD_\infty(U\rho U^\dagger\Vert U\sigma U^\dagger) = \sfD_\infty(\rho\Vert\sigma)$.
\end{fact}
\begin{fact}\label{fc:event-prob}
If $\sigma = \eps\rho + (1-\eps)\rho'$, then $\sfD_\infty(\rho\Vert \sigma) \leq \log(1/\eps)$.
\end{fact}
\begin{fact}\label{fc:Sinfty-tri}
For any three quantum states $\rho, \sigma, \vph$ such that $\supp(\rho) \subseteq \supp(\vph) \subseteq \supp(\sigma)$,
\[ \sfD_\infty(\rho\Vert\sigma) \leq \sfD_\infty(\rho\Vert\vph) + \sfD_\infty(\vph\Vert\sigma). \]
\end{fact}

The conditional min-entropy of $Y$ given $Z$ is defined as
\[ \sfH_\infty(Y|Z)_\rho = \inf\{\lambda: \exists \sigma_Z \text{ s.t. } \rho_{YZ} \leq 2^{-\lambda}\Id_Y\otimes\sigma_Z\}.\]
The conditional Renyi-2 entropy of $Y$ given $Z$ is defined as
\[ \sfH_2(Y|Z)_\rho = -\log\Tr\left(\rho_{YZ}(\Id\otimes\rho_Z^{-1/2})\rho_{YZ}(\Id\otimes\rho_Z^{-1/2})\right). \]
The conditional Hartley entropy of $Y$ given $Z$ is defined as
\[ \sfH_0(Y|Z)_\rho = \log\left(\sup_{\sigma_Z}\Tr(\supp(\rho_{YZ})(\Id_Y\otimes\sigma_Z))\right)\]
where $\supp(\rho_{YZ})$ is the projector on to the support of $\rho_{YZ}$. For a classical distribution $\sfP_{YZ}$, this reduces to 
\[ \sfH_0(Y|Z)_{\sfP_{YZ}} = \log\left(\sup_z\left|\{y: \sfP_{YZ}(y,z) > 0\}\right|\right).\]
The conditional entropies satisfy
\[ \sfH_0(Y|Z)_\rho \geq \sfH(Y|Z)_\rho \geq \sfH_2(Y|Z)_\rho \geq \sfH_\infty(Y|Z)_\rho.\]
\begin{fact}\label{fc:local-Hmin}
All the conditional entropies ($\sfH, \sfH_\infty, \sfH_2, \sfH_0$) are invariant under isometries on one of the systems. That is, if $\rho_{YZ'} = (\Id_Y\otimes U)\sigma_{YZ}(\Id_Y\otimes U^\dagger)$, and $\vph_{Y'Z} = (U\otimes\Id_Z)\sigma_{YZ}(U^\dagger\otimes\Id)$ then
\[ \sfH_\infty(Y|Z)_\sigma = \sfH_\infty(Y|Z')_\rho = \sfH_\infty(Y'|Z)_\vph,\]
and similar statements hold for the other conditional entropies as well.
\end{fact}

For any distance measure (not necessarily a metric) $d$ between states, the $\eps$-smoothed relative min-entropy between $\rho$ and $\sigma$ w.r.t. $d$ is defined as
\[ \sfD^{\eps,d}_\infty(\rho\Vert \sigma) = \inf_{\rho': d(\rho,\rho') \leq\eps}\sfD_\infty(\rho'\Vert\sigma).\]
When $d$ is the $\ell_1$ distance, we often omit the superscript.
\begin{fact}[Quantum Substate Theorem, \cite{JRS02,JRS09,JN12}]\label{fc:substate}
For any two states $\rho$ and $\sigma$ such that the support of $\rho$ is contained in the support of $\sigma$, and any $\eps > 0$,\footnote{Since $1-\sfF$ is the distance measure rather than $\sfF$ itself, the closeness condition for $\sfD^{\eps,\sfF}_\infty(\rho\Vert\sigma)$ is $\sfF(\rho, \rho') \geq 1-\eps$.}
\[ \sfD^{\eps,\sfF}_\infty(\rho\Vert\sigma) \leq \frac{\sfD(\rho\Vert\sigma)+1}{\eps} + \log\left(\frac{1}{1-\eps}\right).\]
Consequently,
\[ \sfD^\eps_\infty(\rho\Vert\sigma) \leq \frac{4\sfD(\rho\Vert\sigma)+1}{\eps^2} + \log\left(\frac{1}{1-\eps^2/4}\right).\]
\end{fact}
\begin{fact}[\cite{JRS02}]\label{fc:substate-proj}
For two states $\rho_X$ and $\sigma_X$, if $\sfD^{\eps,\Delta}_\infty(\rho_X\Vert\sigma_X) = c$, then for any purifications $\ket{\rho}_{XY}$ and $\ket{\sigma}_{XY'}$, there exists a measurement operator $M$ taking $Y'$ to $Y$, such that $\Id\otimes M$ succeeds on $\ket{\sigma}_{XY'}$ with probability $2^{-c}$, and
\[ \Delta\left(2^c(\Id\otimes M)\state{\sigma}_{XY'}(\Id\otimes M^\dagger), \state{\rho}_{XY}\right) \leq \eps.\]
\end{fact}

\begin{fact}\label{dim-ub}
For any quantum state $\rho_{YZ}$,
\[ \inf_{\sigma_Z}\sfD_\infty(\rho_{YZ} \Vert \rho_Y\otimes\sigma_Z) \leq 2\min\{\log|\clY|,\log|\clZ|\}.\]
\end{fact}

The $\eps$-smoothed version of the conditional min-entropy w.r.t. some distance measure $d$ are defined as
\[ \sfH^{\eps, d}_\infty(Y|Z)_\rho = \sup_{\rho':d(\rho,\rho') \leq \eps}\sfH_\infty(Y|Z)_{\rho'}.\]
$\eps$-smoothed versions of the conditional Renyi-2 and Hartley entropies are defined similarly. In this case as well, we shall emit the superscript when the distance measure is the $\ell_1$ distance.
\begin{fact}\label{fc:Hinf-ch-rule}
For any state $\rho_{XYZ}$,
\[ \sfH^\eps_\infty(Y|Z)_\rho \geq \sfH^\eps_\infty(Y|XZ)_\rho \geq \sfH^\eps_\infty(YX|Z)_\rho - \log|\clX|.\]
The equivalent statements hold for $\sfH_2$ and $\sfH_0$ as well.
\end{fact}
\begin{fact}[\cite{Tom16}]\label{fc:H2-Hmin}
For any state $\rho_{YZ}$,
\[ \sfH_\infty^{\eps, \Delta}(Y|Z)_\rho \geq \sfH_2(Y|Z)_\rho - \log(2/\eps^2).\]
By the Fuchs-van de Graaf inequality, this implies
\[ \sfH_\infty^{\eps}(Y|Z)_\rho \geq \sfH_2(Y|Z)_\rho - \log(2/\eps).\]
\end{fact}

The mutual information between $Y$ and $Z$ with respect to a state $\rho$ on $YZ$ can be defined in the following equivalent ways:
\[  \sfI(Y:Z)_\rho = \sfD(\rho_{YZ}\Vert\rho_Y\otimes\rho_Z) = \sfH(Y)_\rho - \sfH(Y|Z)_\rho = \sfH(Z)_\rho - \sfH(Z|Y)_\rho.\]
The conditional mutual information between $Y$ and $Z$ conditioned on $X$ is defined as
\[ \sfI(Y:Z|X)_\rho = \sfH(Y|X)_\rho - \sfH(Y|XZ)_\rho = \sfH(Z|X)_\rho - \sfH(Z|XY)_\rho.\]
Mutual information can be seen to satisfy the chain rule
\[  \sfI(XY:Z)_\rho =  \sfI(X:Z)_\rho +  \sfI(Y:Z|X)_\rho.\]

\begin{fact}[Quantum Gibbs' inequality, see e.g. - \cite{BVY17}]\label{fc:quantum-gibbs}
For any three states $\rho_{XY}, \sigma_X, \vph_Y$,
\[ \sfD(\rho_{XY}\Vert\sigma_X\otimes\vph_Y) \geq \sfD(\rho_{XY}\Vert\sigma_X\otimes\rho_Y) \geq \sfI(X:Y)_\rho.\]
\end{fact}

A state of the form
\[ \rho_{XY} = \sum_x \sfP_X(x)\state{x}_X\otimes\rho_{Y|x}\]
is called a CQ (classical-quantum) state, with $X$ being the classical register and $Y$ being quantum. We shall use $X$ to refer to both the classical register and the classical random variable with the associated distribution. As in the classical case, here we are using $\rho_{Y|x}$ to denote the state of the register $Y$ conditioned on $X=x$, or in other words the state of the register $Y$ when a measurement is done on the $X$ register and the outcome is $x$. Hence $\rho_{XY|x} = \state{x}_X\otimes \rho_{Y|x}$. When the registers are clear from context we shall often write simply $\rho_x$.

For CQ states where $X$ is the classical register, relative entropy has the chain rule
\[ \sfD(\rho_{XY}\Vert\sigma_{XY}) = \sfD(\rho_X\Vert\sigma_X) + \bbE_{\rho_X}\sfD(\rho_{Y|x}\Vert\sigma_{Y|x}).\]
Using this, the following fact follows by expanding out the relative entropies.
\begin{fact}\label{fc:cond-dec}
For CQ states $\rho_{XY}$ and $\sigma_{XY}$,
\[ \bbE_{\rho_X}\sfD(\rho_{Y|x}\Vert\sigma_Y) - \sfD(\rho_Y\Vert\sigma_Y) = \bbE_{\rho_X}\sfD(\rho_{Y|x}\Vert\rho_Y) \geq 0. \]
\end{fact}
\begin{fact}[\cite{KRS09}]\label{fc:guess-prob}
For a CQ state $\rho_{XY}$ where $X$ is the classical register, $\sfH_\infty(X|Y)_\rho$ is equal to the negative logarithm of the maximum probability of guessing $X$ from the quantum system $\rho_{Y|x}$, i.e.,
\[ \sfH_\infty(X|Y)_\rho = 	-\log\left(\sup_{\{M_x\}_x}\sum_x\sfP_X(x)\Tr(M_x\rho_{Y|x}M^\dagger_x)\right)\]
where the maximization is over the set of measurements with measurement operators indexed by $x$.
\end{fact}
\begin{fact}[\cite{JMS17}]\label{fc:H2-prob}
Let $\sigma_{Z}$ be any state and let $\sigma_{YZ}$ be its canonical purification. Let $\rho_{XYZ}$ be the CQ state produced by applying the channel which does the measurement with measurement operators $\{M_x\}_x$ on the $Z$ register of $\sigma_{YZ}$, and records the outcome in the register $X$. Let $\vph_{XX'YZ}$ be the state obtained by further applying the channel which does the measurement with measurement operators $\{M^*_x\}_x$\footnote{Here $M^*_x$ denotes the measurement operator obtained by replacing each entry in $M_x$ with its complex conjugate. Note that this is a basis-dependent operation, and it appears because of a technicality regarding the canonical purification, which is also basis-dependent. However, the POVM elements are the same $\{M_xM^\dagger_x\}_x$ for both $\{M_x\}_x$ and $\{M^*_x\}_x$, so this technicality does not really matter.} on the $Y$ register of $\rho_{XYZ}$, and records the outcome in register $X'$. Then we have,
\[ \sfH_2(X|Y)_\rho = -\log\left(\Pr_\vph[X=X']\right). \]
\end{fact}
The above two facts give operational interpretaions of $\sfH_\infty$ and $\sfH_2$ for CQ states. The first fact simply says that $\sfH_\infty(X|Y)_\rho$ is the the log of the inverse of the best probability of guessing $X$ given $Y$ in $\rho$. The second fact is somewhat harder to interpret. Here we start with a CQ state $\rho_{XY}$, and we think of it as having been produced by starting with the state $\sigma_{YZ}$ which was a canonical purification, and then doing a measurement on $Z$ and recording the outcome (and tracing out $Y$ itself). Then $\sfH_2(X|Z)_\rho$ is the log of the inverse of the probability of guessing $X$ by doing the same measurement (up to complex conjugation) on $Z$. Note that this guessing probability is obviously smaller than the best guessing probability of $X$ given $Y$, which is consistent with $\sfH_2(X|Y)_\rho \geq \sfH_\infty(X|Y)_\rho$.

\subsection{Quantum communication \& non-local games}
An interactive entanglement-assisted quantum communication protocol $\clP$ between $l$ parties goes as follows: before the start of the protocol, the $l$ parties share a joint entangled state, and at the start parties 1 through $l$ receive inputs $x^1,\ldots, x^l$ respectively from $\clX^1\times\ldots\times\clX^l$. We assume that only the $j$-th party communicates in rounds $\{j,j+l, j+2l,\ldots\}$, and sends messages to all the other parties. For $i \in \{j, j+l, \ldots, \}$, in the $i$-th round the $j$-th party has a memory register $E_{i-l}$ from the previous round in which they communicated (when $i=j$, this is just the $j$-th party's part of the initial shared entangled state), as well as message registers $M^j_{i-l+1}, \ldots, M^j_{i-1}$ that they have received from all the other parties in the $(i-l+1)$-th to $(i-1)$-th rounds. The $j$-th party applies a unitary depending on their input $x^j$ on all these registers, to generate a register $E_i$ that they keep as memory, and a message $M_i = M^1_i\ldots M^{j-1}_iM^{j+1}_i\ldots M^l_i$, where $M^{j'}_i$ is sent to the $j'$-th party in this round. After all the communication rounds are done, the $j$-th party applies a final unitary on the memory and message registers they currently have, and then measures in the computational basis to produce their answer $a^j \in \clA^j$. We shall denote the outputs of $\clP$ on inputs $x^1\ldots x^l$ to $\clP$ by $\clP(x^1\ldots x^l)$ --- this is a random variable, as $\clP$'s outputs are not necessarily deterministic.

The following lemma about the final state of a quantum communication protocol is proved in Appendix \ref{ap:int-holevo}.
\begin{lemma}\label{int-holevo}
Let $\ket{\sigma}_{A^1\ldots A^l|x^1\ldots x^l}$ be the pure state shared by the $l$ parties at the end of a quantum communication protocol, on inputs $x^1, \ldots x^l$, with party $j$ holding register $A^j$. For any product input distribution $\sfP_{X^1\ldots X^l}$ on $\clX^1\times\ldots\times\clX^l$, define
\begin{align*}
\ket{\sigma}_{X^1\tX^1 \ldots X^l\tX^l A^1\ldots A^l} & = \sum_{x_1\ldots x_l}\sqrt{\sfP_{X^1\ldots X^l}(x_1\ldots x_l)}\ket{x^1x^1\ldots x^lx^l}_{X^1\tX^1 \ldots X^l\tX^l}\ket{\sigma}_{A^1\ldots A^l|x^1\ldots x^l}.
\end{align*}
If $c^j$ is the total communication from the $j$-th party in the protocol, then there for all $j \in [l]$, there exists a state $\rho^j_{X^{-j}\tX^{-j} A^{-j}}$ such that
\[ \sfD_\infty\left(\sigma_{X^jX^{-j}\tX^{-j} A^{-j}}\middle\Vert\sigma_{X^j}\otimes\rho^j_{X^{-j}\tX^{-j} A^{-j}}\right) \leq 2c^j \]
where $X^{-j}$ denotes $X^1\ldots X^{j-1}X^{j+1}\ldots X^l$, and $\tX^{-j}$ and $A^{-j}$ are defined analogously.
\end{lemma}

\begin{definition}
For a predicate $\sfV$ on $(\clA^1\times\ldots\times\clA^l)\times(\clX^1\times\ldots\times\clX^l)$, its entanglement-assisted $l$-party quantum communication complexity with error $0 < \eps < 1$, denoted by $\Q_\eps(\sfV)$, is the minimum total communication in an interactive entanglement-assisted quantum protocol such that for all $x^1\ldots x^l \in \clX^1\times\ldots\times\clX^l$,
\[ \Pr\left[\sfV\left(\clP(x^1\ldots x^l),x^1\ldots x^l\right) = 1\right] \geq 1 - \eps.\]
\end{definition}
\begin{definition}
For a predicate $\sfV$ on $(\clA^1\times\ldots\times\clA^l)\times(\clX^1\times\ldots\times\clX^l)$ and a distribution $p$ on $\clX^1\times\ldots\times\clX^l$, the distributional entanglement-assisted $l$-party quantum communication complexity of $\sfV$ with error $0 < \eps < 1$ w.r.t. distribution $p$, denoted by $\Q_\eps(\sfV,p)$, is the minimum total communication in an interactive entanglement-assisted quantum protocol such that,
\[ \Pr\left[\sfV\left(\clP(x^1\ldots x^l),x^1\ldots x^l\right) = 1\right] \geq 1 - \eps\]
where the probability is taken over the distribution $p$ for $x^1\ldots x^l$, as well as the internal randomness of $\clP$.
\end{definition}

\begin{fact}[Yao's Lemma, \cite{Yao79}]\label{yao}
For any $0< \eps < 1$, and any predicate $\sfV$, $\Q_{\eps}(V) = \max_p\Q_{\eps}(\sfV,p)$.\footnote{Note that this statement of Yao's lemma is not contradicted by the results in \cite{dGdW02}, which are about exact quantum communication protocols. For completeness, a proof of this lemma is provided in Appendix \ref{ap:yao}.}
\end{fact}

An $l$-player non-local game $G$ is described as $(p, \clX^1\times\ldots\times\clX^l,\clA^1\times\ldots\times\clA^l, \sfV)$ where $p$ is a distribution over the input set $\clX^1\times\ldots\times\clX^l$, $\clA^1\times\ldots\times\clA^l$ is the output set, and $\sfV$ is a predicate on the outputs and inputs. In an entangled strategy for a non-local game, the players are allowed to share an $l$-partite entangled state. Player $j$ gets input $x^j$ and performs a unitary and a measurement depending on their input on their part of the entangled state, to give their output $a^j$. The value achieved by a strategy on $G$ is the probability over $p$ and the internal randomness of the strategy that $\sfV(a^1\ldots a^l, x^1\ldots x^l)=1$.
\begin{definition}
The entangled value of a game $G=(p, \clX^1\times\ldots\times\clX^l,\clA^1\times\ldots\times\clA^l, \sfV)$, denoted by $\omega^*(G)$, is the maximum value achieved by any strategy for $G$.
\end{definition}


\section{Quantum partition bound}\label{sec:part-bound}
For sets $\clX^1,\ldots, \clX^l$ and $\clA^1, \ldots, \clA^l$, let $\clQ(\clA^1\times\ldots\times\clA^l,\clX^1\times\ldots\times\clX^l)$ denote the set of conditional probability distributions $q(a^1\ldots a^l|x^1\ldots x^l)$ that can be obtained by $l$ parties who share an $l$-partite entangled state, receive inputs $x^j\in\clX^j$ respectively, and perform measurements on their parts of the entangled state to obtain outputs $a^j$, without communicating. That is, $\clQ(\clA^1\times\ldots\times\clA^l,\clX^1\times\ldots\times\clX^l)$ is the following set:
\[ \left\{\left(\matel{\psi}{M^1_{a^1|x^1}\otimes \ldots \otimes M^l_{a^l|x^l}}{\psi}\right)_{a^1\ldots a^l,x^1\ldots x^l} \middle| \ket{\psi} \text{ is a state}, \forall a^j, x^j, j, \sum_{a^j\in\clA^j}M^j_{a^j|x^j} = \Id, M^j_{a^j|x^j} \geq 0\right\}.\]

We state definitions for three variants of the quantum partition bound, the first of which is non-distributional and was given by \cite{LLR12}. The second two are distributional modifications which we shall use.
\begin{definition}
For a predicate $\sfV$ on $(\clA^1\times\ldots \times \clA^l)\times(\clX^1\times\ldots\times\clX^l)$, and $0<\eps<1$, let $\bot$ be a special symbol not in any $\clA^j$. The quantum partition bound for $\sfV$ with $\eps$ error, denoted by $\eff_{\eps}(\sfV)$, is defined as the optimal value of the following optimization problem:
\begin{align*}
\min & \quad \frac{1}{\eta} \\
\text{s.t.} &  \sum_{a^1\ldots a^l: \sfV(a^1\ldots a^l,x^1\ldots x^l)=1} q(a^1\ldots a^l|x^1\ldots x^l) \geq (1-\eps)\eta \quad \forall x^1\ldots x^l \in \clX^1\times\ldots\times\clX^l\\
 & \quad \sum_{a^1\ldots a_l\in \clA^1\times\ldots\times\clA^l}q(a^1\ldots a^l|x^1\ldots x^l) = \eta \quad \forall\,x^1\ldots x^l\in\clX^1\times\ldots\times\clX^l\\
 & \quad q(a'^1\ldots a'^l|x^1\ldots x^l) \in \clQ\left((\clA^1\cup\{\bot\})\times\ldots\times(\clA^l\cup\{\bot\}), \clX^1\times\ldots\times\clX^l\right).
\end{align*}
\end{definition}

\begin{definition}
For a predicate $\sfV$ on $(\clA^1\times\ldots \times \clA^l)\times(\clX^1\times\ldots\times\clX^l)$, a distribution $p(x^1\ldots x^l)$ on $\clX^1\times\ldots\times\clX^l$, and $0<\eps<1$, let $\bot$ be a special symbol not in any $\clA^j$. The quantum partition bound for $\sfV$ with $\eps$ error with respect to $p$, denoted by $\widetilde{\eff}_{\eps}(\sfV,p)$, is defined as the optimal value of the following optimization problem:
\begin{align*}
\min & \quad \frac{1}{\eta} \\
\text{s.t.} &  \quad \sum_{x^1\ldots x^l\in \clX^1\times\ldots\times\clX^l}p(x^1\ldots x^l)\sum_{a^1\ldots a^l: \sfV(a^1\ldots a^l,x^1\ldots x^l)=1} q(a^1\ldots a^l|x^1\ldots x^l) \geq (1-\eps)\eta \\
 & \quad \sum_{a^1\ldots a_l\in \clA^1\times\ldots\times\clA^l}q(a^1\ldots a^l|x^1\ldots x^l) = \eta \quad \forall\,x^1\ldots x^l\in\clX^1\times\ldots\times\clX^l\\
 & \quad q(a'^1\ldots a'^l|x^1\ldots x^l) \in \clQ\left((\clA^1\cup\{\bot\})\times\ldots\times(\clA^l\cup\{\bot\}), \clX^1\times\ldots\times\clX^l\right).
\end{align*}
\end{definition}
\begin{definition}
For a predicate $\sfV$ on $(\clA^1\times\ldots \times \clA^l)\times(\clX^1\times\ldots\times\clX^l)$, a distribution $p(x^1\ldots x^l)$ on $\clX^1\times\ldots\times\clX^l$, and $0<\eps<1$, let $\bot$ be a special symbol not in any $\clA^j$. The average quantum partition bound for $\sfV$ with $\eps$ error with respect to $p$, denoted by $\eff_{\eps}(\sfV,p)$, is defined as the optimal value of the following optimization problem:
\begin{align*}
\min & \quad \frac{1}{\eta} \\
\text{s.t.} &  \quad \sum_{x^1\ldots x^l\in \clX^1\times\ldots\times\clX^l}p(x^1\ldots x^l)\sum_{a^1\ldots a^l: \sfV(a^1\ldots a^l,x^1\ldots x^l)=1} q(a^1\ldots a^l|x^1\ldots x^l) \geq (1-\eps)\eta \\
 & \sum_{x^1\ldots x^l \in \clX^1\times\ldots\times\clX^l}p(x^1\ldots x^l)\sum_{a^1\ldots a^l\in \clA^1\times\ldots\times\clA^l}q(a^1\ldots a^l|x^1\ldots x^l) = \eta \\
 & \quad q(a'^1\ldots a'^l|x^1\ldots x^l) \in \clQ\left((\clA^1\cup\{\bot\})\times\ldots\times(\clA^l\cup\{\bot\}), \clX^1\times\ldots\times\clX^l\right).
\end{align*}
\end{definition}

Operationally, $\eff_\eps(\sfV)$, $\widetilde{\eff}_\eps(\sfV,p)$ and $\eff_{\eps}(\sfV,p)$ are connected to zero-communication protocol (with aborts) to compute $\sfV$. A zero-communication protocol is one which any player is allowed to abort (indicated by them outputting the $\bot$ symbol), but if nobody aborts they need to compute $\sfV$ correctly. A zero-communication protocol for $\sfV$ is basically a strategy for a non-local game version of $\sfV$, with the output alphabet extended to $\clA^1\cup\{\bot\})\times\ldots\times(\clA^l\cup\{\bot\})$. Now we can have different conditions on the abort and success probability conditioned on not aborting for such protocols.
\begin{itemize}
\item Suppose the protocol is required to not abort on every input $x^1\ldots x^l$ with the same probability $\eta$, and conditioned on not aborting, every input is required to compute $\sfV$ correctly with probability $(1-\eps)$. $\eff_\eps(\sfV)$ corresponds to the efficiency, i.e., the inverse of the maximum probability of not aborting, in such a protocol.
\item Suppose the protocol is required to not abort with the same probability $\eta$ on every input $x^1\ldots x^l$, but conditioned on not aborting, the probability of computing $\sfV$ correctly, averaged over the inputs from $p$, is at least $(1-\eps)$. $\widetilde{\eff}_{\eps}(\sfV)$ is the inverse of the maximum  probability of not aborting in such a protocol.
\item Suppose the protocol aborts on input $x^1\ldots x^l$ with probability $\eta_{x^1\ldots x^l}$, and we require that the average over $x^1\ldots x^l$ from $p$ is $\eta$. Moreover, we require that the average probability of computing $\sfV$ correctly is at least $(1-\eps)\eta$, i.e., the average probability of correctness conditioned on not aborting is at least $(1-\eps)$.
$\eff_\eps(\sfV,p)$ is the inverse of the maximum probability of not aborting in such a protocol.
\end{itemize}
Because the requirements from the protocols are successively relaxed, it is easy to see that for any $p$,
\[ \eff_\eps(\sfV) \geq \widetilde{\eff}_\eps(\sfV,p) \geq \eff_\eps(\sfV,p).\]

The following lemma shows that $\widetilde{\eff}(\sfV,p)$, and hence $\eff_\eps(\sfV,p)$ lower bounds communication. The proof of this is a slight modification the proof in \cite{LLR12} which lower bounded $\Q_\eps(f)$ by $\eff_\eps(\sfV)$.

\begin{lemma}
\label{lem:eff-lb}
For any predicate $\sfV$ on $(\clA^1\times\ldots\times\clA^l)\times(\clX^1\times\ldots\times\clX^l)$, any distribution $p$ on $\clX^1\times\ldots\times\clX^l$ and error $\eps$,
\[ \Q_{\eps}(\sfV,p) \geq \frac{1}{2}\log\widetilde{\eff}_{\eps}(\sfV,p).\]
\end{lemma}
\begin{proof}
We shall show that if there is a quantum interactive protocol $\clP$ for $\sfV$ with $c$ qubits of communication and error probability at most $\eps$, over input distribution $p$, then there is a zero-communication quantum protocol $\clP''$ which does not abort with probability $2^{-2c}$ worst case over all inputs, and when it does not abort it computes $\sfV$ with the same error probability over $p$.

Firstly, we can use entanglement and teleportation to get a protocol $\clP'$ from $\clP$, which only involves at most $2c$ bits of classical communication (with the players doing measurements according to the classical messages they receive and their inputs, on their parts of a shared entangled state). We assume that the number of bits communicated in $\clP'$ is of some fixed length every round for every input, with the total communication being $2c$ (this can be done by padding dummy bits if necessary).

Now in the zero-communication protocol $\clP''$, the players will share the same initial entangled state as in $\clP'$, and also $2c$ uniformly random classical bits. If player $j$ communicates in the $i$-th round, let $r_i=r^1_i\ldots r^{j-1}_ir^{j+1}_i\ldots r^l_i$ denote the portion of the shared randomness that corresponds to the bits in the $i$-th round of communication in $\clP'$, with $r^k_i$ corresponding to the message to the $k$-th player. On inputs $x^1\ldots x^l$, the players do the following in $\clP''$:
\begin{itemize}
\item For each round $i$, if player $j$ is the one communicating in that round, player $j$ assumes $r^j_{i-l+1}\ldots r^j_{i-1}$ are the classical messages they have received from the other $l-1$ players between the $(i-l)$-th and the $i$-th round. Player $j$ does a measurement on their part of the entangled state as they do in the $i$-th round of $\clP'$, depending on $x^j$, their previous measurement outcomes, and messages from the other players. If $r_i$ is not compatible with their input and these measurement outcomes and previous messages, then player $j$ outputs $\bot$.
\item At the end, if a player has not output $\bot$ yet, they output according to $\clP'$.
\end{itemize}
Once the outputs of the measurements are fixed, the protocol is deterministic. So a transcript that is separately compatible for all the players, is compatible for all of them, and there is exactly one such transcript. $\{r_i\}_i$ is equal to this transcript with probability $2^{-2c}$, and hence no player outputs $\bot$ with probability $2^{-2c}$. When they do not output $\bot$, the trancript is correct for input $x^1\ldots x^l$, and hence $\clP''$ is correct with probability at least $1-\eps$ over the distribution $p$ on $x^1\ldots x^l$, due to the correctness of $\clP'$.
\end{proof}

Yao's Lemma and Lemma \ref{lem:eff-lb} imply that for any $p$, $\log\widetilde{\eff}_\eps(\sfV,p)$ and therefore $\log\eff_\eps(\sfV,p)$ are lower bounds on $\Q_\eps(\sfV)$.

\subsection{Relationship between $\eff$ and the generalized discrepancy method}
In this section we shall prove Theorem \ref{thm:gamma_2-eff}, recalled below.
\gamm*
We shall not define $\gamma_2^\alpha$ and its dual norm $\gamma_2^*$ for general matrices. Instead, we shall use an exact characterization of $\gamma_2^*(F)$ for a boolean $f$ in terms of non-local games given by Tsirelson, and then use a duality relation to express $\gamma^\alpha_2$ in terms of $\gamma_2^*$.
\begin{fact}[\cite{Tsi87}]\label{fc:gamma-xor}
For total $f: \clX\times\clY\to\{-1,+1\}$, let $V_f$ denote its corresponding predicate as given in the statement of Theorem \ref{thm:gamma_2-eff}, and let $p$ be any distribution on $\clX\times\clY$. Then,
\[ \omega^*(G(p,\sfV_f)) = \frac{1}{2}(1+\gamma_2^*(F\circ p)),\]
where $F\circ p$ denotes the entry-wise product of $F$ and $p$.
\end{fact}
\begin{fact}[see e.g. - Theorem 64 in \cite{LS08}]\label{fc:gamma-dual}
For any matrix $A$, $\alpha \geq 1$, $\gamma_2^\alpha(A)$ and $\gamma_2^*(A)$ are related as
\[ \gamma_2^\alpha(A) = \max_M\frac{(\alpha+1)\inner{A}{M} - (\alpha-1)\Vert M\Vert_1}{2\gamma_2^*(M)}.\]
When $A$ is the matrix corresponding to a boolean function $f$, this can also be expressed as
\[ \gamma_2^\alpha(F) = \max_{F',p}\frac{(\alpha+1)\inner{F}{F'\circ p} - (\alpha-1)}{2\gamma_2^*(F'\circ p)}\]
where the maximization is taken over matrices $F'$ with $\pm 1$ entries, and distributions $p$.
\end{fact}
Using this characterization, we give the following definition of $\gamma_2^\alpha(F,p)$.
\begin{definition}
For matrix $F$ with $\pm 1$ entries, $\gamma_2^\alpha(F,p)$ with respect to distribution $p$ is defined as
\[ \gamma_2^\alpha(F,p) = \max_{F'}\frac{(\alpha+1)\inner{F}{F'\circ p} - (\alpha-1)}{2\gamma_2^*(F'\circ p)}.\]\end{definition}

\begin{proof}[Proof of Theorem \ref{thm:gamma_2-eff}]
Our proof closely follows the lower bound for $\Q_\eps(f)$ in terms of $\log\gamma_2^\alpha$ as described in Section 5.3.2 of \cite{LS08}, which is credited to Harry Buhrman.

Suppose $\eff_\eps(\sfV_f,p)= \frac{1}{\eta}$ for some $\eta$. Let $\clP$ be a zero-communication protocol for $V_f$ with constraints as required in the definition of $\eff_\eps(\sfV_f,p)$. Let $\eta_{xy}$ denote the probability of the protocol aborts on input $(x,y)$. Let $O(x,y)$ denote the average (over internal randomness) output given by $\clP$ conditioned on not aborting on inputs $x,y$. Here we are calling $a\cdot b$ the output of the protocol, if Alice outputs $a$ and Bob outputs $b$, which means $O(x,y)$ is some number in $[-1,1]$. Note that $O(x,y)$ is defined conditioned on not aborting, so it is in fact normalized by the quantity $\eta$. From the definition of $\eff_\eps(\sfV_f,p)$, the following condition holds
\[
\sum_{x,y} p(x,y)f(x,y)O(x,y) \geq 1-2\eps.
\]
The above expression is actually the difference between the probability of computing $f$ correctly and the probability of computing it incorrectly, which is why we get $1-2\eps$.

Now let $f':\clX\times\clY\to\{-1,+1\}$ be an arbitrary boolean function, and define $\sfV_{f'}$ the same way as $\sfV_f$. We shall give a strategy $\clS$ for the game $G(p,\sfV_{f'})$ using the zero-communication protocol $\clP$. $\clS$ works as follows:
\begin{itemize}
\item On inputs $x,y$ for $G(p,\sfV_{f'})$, Alice and Bob run the protocol $\clP$ on $x,y$.
\item If $\clP$ gives output $\bot$ for either player, they output $\pm 1$ uniformly at random.
\item If $\clP$ does not abort, then Alice and Bob both output according to $\clP$.
\end{itemize}
Note that conditioned on $\clP$ not aborting, the average output produced by Alice and Bob on inputs $x,y$ is also $O(x,y)$. Strategy $\clS$ thus wins with probability $\frac{1}{2}(1+\delta)$ ($\delta$ may be negative), where
\[ \delta = \eta\sum_{x,y}p(x,y) f'(x,y)O(x,y).\]
$f(x,y), f'(x,y)$ are in $\{-1,+1\}$, and $O(x,y)$ is in $[-1,1]$. For three numbers $\alpha,\beta \in \{-1,+1\}$, $\theta \in [-1,1]$, the following condition is true, and can be checked by putting in the four possible values of $(\alpha,\beta)$:
\[ \beta \theta \geq \alpha\beta + \alpha\theta - 1.\]
Using the above on $f(x,y), f'(x,y), O(x,y)$ we get,
\begin{align*}
\sum_{x,y}p(x,y) f'(x,y)O(x,y) & \geq \sum_{x,y}p(x,y)\left(f(x,y)f'(x,y) + f(x,y)O(x,y) - 1\right) \\
& \geq \sum_{x,y}p(x,y)f(x,y)f'(x,y) + (1-2\eps) - 1 \\
&  = \inner{F}{F'\circ p} - 2\eps.
\end{align*}
By Fact \ref{fc:gamma-xor} we have,
\begin{align*}
\gamma_2^*(F'\circ p) & \geq \delta \geq \eta(\inner{F}{F'\circ p} - 2\eps)
\end{align*}
which gives us
\[ \frac{1}{\eta} \geq \max_F\frac{\inner{F}{F'\circ p} - 2\eps}{\gamma_2^*(F'\circ p)} = (1-2\eps)\gamma_2^\alpha(F,p)\]
with $\alpha = \frac{1+2\eps}{1-2\eps}$.
\end{proof}


\section{Substate Perturbation Lemma}\label{sec:perturb-D}

To prove the Substate Perturbation Lemma, we use the following result due to \cite{ABJT18}. This result is stated in terms of $\sfI_{\max}$ for general states in \cite{ABJT18}, where some of the states involved are optimized over. Moreover, the distance between $\sigma'_{XB}$ and $\sigma_{XB}$ is considered, rather than just the distance between $\sigma'_B$ and $\sigma_B$. However, for the purposes of the proof the optimal state does not matter, and only the distance between $\sigma'_B$ and $\sigma_B$ is relevant; so we state in the form below. Our proof of the Substate Perturbation Lemma is also heavily inspired by their proof of this result.
\begin{fact}[\cite{ABJT18}, Theorem 2]\label{fc:Imax-marginal}
Suppose there are states $\sigma'_{XB}, \sigma_B$ and $\psi_X$ satisfying $\Delta(\sigma_{B}, \sigma'_{B}) \leq \eps$ and
\[ \sigma'_{XB} \leq 2^c(\psi_X\otimes\sigma_B).\]
Then for any $\delta >0$, there exists a state $\sigma''_{XB}$ satisfying $\Delta(\sigma'_{XB}, \sigma''_{XB}) \leq \eps+\delta$, $\sigma''_B = \sigma_B$, and
\[ \sigma''_{XB} \leq 2^c\left(1 + \frac{8}{\delta^2}\right)\psi_X\otimes\sigma_B.\]
\end{fact}

\begin{lemma}[Substate Perturbation Lemma]\label{Imax-qu-close}
Suppose there are states $\sigma'_{XB},\sigma_{B}$ and $\psi_X$ satisfying $\Delta(\sigma_{B}, \sigma'_{B}) \leq \eps$,
\[ \sigma'_{XB} \leq 2^c (\psi_X\otimes\sigma_B)\]
and a state $\rho_B$ satisfying $\Delta(\sigma_B,\rho_B) \leq \delta_1$. Then for any $\delta_0 > 0$, there exists state $\rho'_{XB}$ satisfying $\Delta(\rho'_{XB},\sigma'_{XB}) \leq \eps + \delta_0 +\delta_1$, and
\[ \rho'_{XB} \leq 2^{c+1}\left(1+\frac{4}{\delta_0^2}\right)\psi_X\otimes\rho_B.\]
\end{lemma}
\begin{proof}
First we use Fact \ref{fc:Imax-marginal} to get a state $\sigma''_{XB}$ satisfying
\[ \sigma''_{XB} \leq 2^c\left(1+\frac{8}{\delta_0^2}\right)\psi_X\otimes\sigma_B\]
such that $\Delta(\sigma'_{XB},\sigma''_{XB}) \leq \eps+\delta_0$ and $\sigma''_B = \sigma_B$.

Let $U$ be the unitary such that
\[ \sfF\left(\rho_B,\sigma_B\right) = \Tr\left(U\rho_B^{1/2}\sigma_B^{1/2}\right).\]
Define
\[ \rho'_{XB} = \underbrace{(\Id\otimes\rho_B^{1/2}U\sigma_B^{-1/2})\sigma''_{XB}(\Id\otimes\sigma_B^{-1/2}U^\dagger\rho_B^{1/2})}_{\tilde{\vph}_{XB}} + \underbrace{\sigma_X\otimes\rho_B^{1/2}(\Id - U\Pi U^\dagger)\rho_B^{1/2}}_{\tilde{\psi}_{XB}}\]
where all the inverses are generalized and $\Pi$ is the projector onto the support of $\sigma_B$. Note that
\[ (\Id\otimes\rho_B^{1/2}U\sigma_B^{-1/2})\sigma''_{XB}(\Id\otimes\sigma_B^{-1/2}U^\dagger\rho_B^{1/2}) \leq 2^c\left(1 + \frac{8}{\delta_0^2}\right)\psi_X\otimes\rho_B^{1/2}U\sigma_B^{-1/2}\sigma_B\sigma_B^{-1/2}U^\dagger\rho_B^{1/2},\]
and hence
\begin{align*}
\rho'_{XB} & \leq 2^c\left(1+\frac{8}{\delta_0^2}\right)\psi_X\otimes\rho_B^{1/2}U\Pi U^\dagger\rho_B^{1/2} + \psi_X\otimes\rho_B^{1/2}(\Id-U\Pi U^\dagger)\rho_B^{1/2} \\
 & \leq 2^{c+1}\left(1+\frac{4}{\delta_0^2}\right)\psi_X\otimes\rho_B.
\end{align*}

Now we only have to show that $\Delta(\rho'_{XB},\sigma_{XB}) \leq 2\eps+\delta_0+\delta_1$. In order to do this, we note that
\begin{equation}\label{eq:Imax-tri}
\Delta(\rho'_{XB},\sigma'_{XB}) \leq \Delta\left(\rho'_{XB}, \sigma''_{XB}\right) + \Delta(\sigma''_{XB}, \sigma'_{XB}).
\end{equation}
Using Fact \ref{F-concave},
\begin{align}
\sfF\left(\rho'_{XB},\sigma''_{XB}\right)^2 & \geq \Tr(\tilde{\vph}_{XB})\cdot\sfF\left(\frac{\tilde{\vph}_{XB}}{\Tr(\tilde{\vph}_{XB})}, \sigma''_{XB}\right)^2 + \Tr(\tilde{\psi}_{XB})\cdot\sfF\left(\frac{\tilde{\psi}_{XB}}{\Tr(\tilde{\psi}_{XB})},\sigma''_{XB}\right)^2 \nonumber \\
& \geq \Tr(\tilde{\vph}_{XB})\cdot\sfF\left(\frac{\tilde{\vph}_{XB}}{\Tr(\tilde{\vph}_{XB})},\sigma''_{XB}\right)^2 \nonumber \\
& \geq \Tr(\tilde{\vph}_{XB})\cdot\sfF\left(\state{\vph}_{XBC},\state{\sigma''}_{XBC}\right)^2 \label{eq:F-sigma}
\end{align}
where in the last step $\ket{\vph}_{XBC}$ and $\ket{\sigma''}_{XBC}$ are arbitrary purifications of $\vph_{XB} = \tilde{\vph}_{XB}/\Tr(\tilde{\vph}_{XB})$ and $\sigma''_{XB}$, and we have used Fact \ref{fc:chan-l1} with the tracing out operation. Note that $\vph_{BC}$ is obtained from $\sigma''_{BC}$ by doing an operation $\clO_B$ only on $B$, which is akin to applying a measurement and conditioning on success. In particular this operation preserves purity of states. We let $\ket{\vph}_{XBC}$ be the state we get by applying $\clO_B$ on $\ket{\sigma''}_{XBC}$. Now let $\ket{\sigma''_1}_{B\tB}$ be the canonical purification of $\sigma''_B$ and $\ket{\vph_1}_{B\tB}$ be the state we get by applying $\clO_B$ on $\ket{\sigma''_1}_{B\tB}$. These are given by
\begin{align*}
\ket{\sigma''_1}_{B\tB} & = ((\sigma''_B)^{1/2}\otimes\Id)\sum_i\ket{i}_B\ket{i}_{\tB} = (\sigma_B^{1/2}\otimes\Id)\sum_i\ket{i}_B\ket{i}_{\tB} \\
\ket{\vph_1}_{B\tB} & = \frac{\rho_B^{1/2}U\sigma_B^{-1/2}\otimes\Id}{\Tr(\tilde{\vph}_{XB})^{1/2}}\ket{\sigma''}_{B\tB} = \frac{\rho_B^{1/2}U\Pi\otimes\Id}{\Tr(\tilde{\vph}_{XB})^{1/2}}\sum_i\ket{i}_B\ket{i}_{\tB}.
\end{align*}
Since $\ket{\sigma''}_{XBC}$ is also a purification of $\sigma''_B$, there exists an isometry $V$ acting only on $\tB$ such that $\Id_B\otimes V\ket{\sigma''}_{XBC} = \ket{\sigma''_1}_{B\tB}$. Hence,
\begin{align*}
\sfF\left(\state{\vph}_{XBC},\state{\sigma''}_{XBC}\right) & = \sfF\left(\clO_B(\state{\sigma''}_{XBC}), \state{\sigma''}_{XBC}\right) \\
& = \sfF\left(\Id_B\otimes V\left(\clO_B(\state{\sigma''}_{XBC})\right)\Id_B\otimes V^\dagger, \Id_B\otimes V\state{\sigma''}_{XBC}\Id_B\otimes V^\dagger\right) \\
& = \sfF\left(\clO_B\left(\Id_B\otimes V\state{\sigma''}_{XBC})\Id_B\otimes V^\dagger\right), \Id_B\otimes V\state{\sigma''}_{XBC}\Id_B\otimes V^\dagger\right) \\
& = \sfF\left(\state{\sigma''_1}_{B\tB},\state{\vph_1}_{B\tB}\right).
\end{align*}
Putting this in \eqref{eq:F-sigma} gives us
\begin{align*}
\sfF\left(\rho'_{XB},\sigma''_{XB}\right)^2 & \geq \left|\sum_i\sum_j\left(\bra{ii}(\Pi U\rho_B^{1/2}\otimes\Id)\right)\left((\sigma_B^{1/2}\otimes\Id)\ket{jj}\right)\right|^2 \\
 & = \left|\sum_i\matel{i}{\Pi U\rho_B^{1/2}\sigma_B^{1/2}}{i}\right|^2 \\
 & = \left|\Tr(\Pi U\rho_B^{1/2}\sigma_B^{1/2})\right|^2 \\
 & = \left|\Tr(U\rho_B^{1/2}\sigma_B^{1/2})\right|^2 = \sfF\left(\rho_B,\sigma_B\right)^2
\end{align*}
where we have used the fact that $\sigma_B^{1/2}\Pi=\sigma_B^{1/2}$, and the definition of $U$. Putting this in \eqref{eq:Imax-tri} we get,
\[ \Delta(\rho'_{XB},\sigma'_{XB}) \leq \Delta\left(\rho_B,\sigma_B\right) + \Delta(\sigma''_{XB}, \sigma'_{XB}) \leq \delta_1 + \eps + \delta_0. \qedhere \]
\end{proof}

\section{Proof of the direct product theorem}\label{sec:dpt-proof}
In this section, we prove Theorem \ref{thm:dpt}, whose statement is recalled below.
\dpt*
\subsection{Setup}
We consider an interactive quantum protocol $\clP$ for $n$ copies of $\sfV$ with player $j$ having input registers $X^j = X^j_1\ldots X^j_n$, and communicating $c^jn$ bits. The total communication of the protocol is $cn$, where $c = \sum_{j=1}^lc^j$. In the case $c \geq 1$, we shall also assume each $c^j \geq 1$; if some $c^j$ is smaller than 1, we can pad extra bits to it, and this increases total communication by a factor of at most $l$. Hence we have, $\sum_{j=1}^lc^j \leq cl$.

We define the following pure state
\[ \ket{\psi}_{X^1\tX^1\ldots X^l\tX^lE^1\ldots E^lA^1\ldots A^l} = \sum_{xy}\sqrt{\sfP_{X^1\ldots X^l}(x^1\ldots x^l)}\ket{x^1x^1\ldots x^lx^l}_{X^1\tX^1 \ldots X^l\tX^l}\ket{\psi}_{E^1\ldots E^lA^1\ldots A^l|x^1\ldots x^l}\]
where $\sfP_{X^1\ldots X^l}$ is the distribution $p^n$ on $(\clX^1\times\ldots\times\clX^l)^n$, and $\ket{\psi}_{E^1\ldots E^lA^1\ldots A^l|x^1\ldots x^l}$ being the state at the end of the protocol on inputs $x^1,\ldots,x^l$. In $\ket{\psi}_{E^1\ldots E^lA^1\ldots A^l|x^1\ldots x^l}$, $A^j=A^j_1\ldots A^j_n$ are the output registers of player $j$, and $E^j$ is some quantum register they have that they don't measure. We use $\sfP_{X^1\ldots X^lA^1\ldots A^l}$ to denote the distribution of $X^1\ldots X^lA^1\ldots A^l$ in $\ket{\psi}$. We shall use $X$ to denote $X^1\ldots X^l$, $X_i$ to denote $X^1_i\ldots X^l_i$, $X^{-j}$ to denote $X^1\ldots X^{j-1}X^{j+1}\ldots X^l$, and $X^{\leq j}$ to denote $X^1\ldots X^j$. Similar notation will be used for $\tX^j, E^j, A^j$. Also for a subset $C \subseteq [n]$, we shall use use $X_C$ to denote $(X_i)_{i\in C}$.

We shall show the following lemma, which can be applied inductively to get Theorem \ref{thm:dpt}.
\begin{lemma}\label{lm:induct}
For $i\in[k]$, let $T_i = \sfV(A^1_i\ldots A^l_i,X^1_i\ldots X^l_i)$ in $\clP$, and let $\clE$ denote the event $\prod_{i\in C}T_i=1$ for some $C\subseteq [n]$ such that $|C| \leq n/2$, 
\begin{enumerate}[(i)]
\item If $c < 1$,
\[ \bbE_{i\in\bar{C}}\Pr[T_i=1|\clE] \leq \omega^*(G(p,\sfV)) + \sqrt{2lc} + l\sqrt{2\delta},\]
\item If $1 \leq c < \frac{\zeta^2}{270l^3}\eff_{\eps+\zeta}(\sfV,p)$, and if $\delta < 1$, there exists an $i\in\bar{C}$ such that
\[ \Pr[T_i=1|\clE] \leq 1 - \eps,\]
\end{enumerate}
where
\[ \delta = \frac{|C|\log(|\clA^1|\cdot\ldots\cdot|\clA^l|) + \log(1/\Pr[\clE])}{n}.\]
\end{lemma}
In order to get the statement of case (i) of Theorem \ref{thm:dpt} from case (i) of Lemma \ref{lm:induct}, we start with $C=\emptyset$, and find some $i\in[n]$ such that $\Pr[T_i=1] \leq 1-\nu + \sqrt{2lc} + \frac{\nu}{2}$. As long as $l\sqrt{2\delta}$ is at most $\frac{\nu}{2}$ we can do this. When we have built up a non-empty set $C$ this way, if either $|C| = \Omega\left(\frac{\nu^2n}{l^2\log(|\clA^1|\cdot\ldots\cdot|\clA^l|)}\right)$, or $\Pr[\Pi_{i\in C}T_i=1] \leq \exp\left(-\Omega\left(\frac{\nu^2n}{l^2\log(|\clA^1|\cdot\ldots\cdot|\clA^l|)}\right)\right)$, we are already done. Otherwise, $l\sqrt{2\delta} < \frac{\nu}{2}$, and we can continue the process.

The bound on $\Pr[T_i=1|\clE]$ in case (ii) of Lemma \ref{lm:induct} does not depend on $\delta$, but it requires $\delta<1$ as a precondition. Hence following the same process there, we can go up to $C$ of size $|C| = \Theta\left(\frac{n}{\log(|\clA^1|\cdot\ldots\cdot|\clA^l|)}\right)$, or $\Pr[\Pi_{i\in C}T_i=1] = \exp\left(-\frac{n}{\log(|\clA^1|\cdot\ldots\cdot|\clA^l|)}\right)$.

Since in case (i) Lemma \ref{lm:induct} gives us a bound on $\bbE_{i\in \bar{C}}\Pr[T_i=1|\clE]$ rather than showing just that there exists an $i$ for which $\Pr[T_i=1|\clE]$ is bounded, we can use it to show the following corollary, which we shall later use in our DIQKD application. See Appendix C of \cite{JMS17} for a proof of how this follows from the lemma.
\begin{cor}\label{cor:rand-V}
Let $\sfV^{t/n}_{\mathrm{rand}}$ be the randomized predicate which is satisfied if $\sfV$ is satisfied on a random subset of size $t$ of $[n]$. If the communication cost of $\clP$ is $cn<n$, then\footnote{Note that $\suc(p^n,\sfV^{t/n}_\mathrm{rand},\clP)$ accounts for the randomness inherent in $\sfV^{t/n}_\mathrm{rand}$ in addition to $p^n$ and the protocol.}
\[ \suc(p^n,\sfV^{t/n}_{\mathrm{rand}},\clP) \leq \left(\omega^*(G(p,\sfV)) + O\left(\sqrt{lc} + l\sqrt{\frac{t\cdot\log(|\clA^1|\cdot\ldots\cdot|\clA^l|)}{n}}\right)\right)^t.\]
\end{cor}

\subsection{Proof of Lemma \ref{lm:induct}}
We define the following state which is $\ket{\psi}$ conditioned on success event $\clE$ in $C$:
\[ \ket{\vph}_{X\tX EA} = \frac{1}{\sqrt{\gamma}}\sum_{x_Cx_{\bar{C}}}\sqrt{\sfP_{X}(x_Cx_{\bar{C}})}\ket{x_Cx_{\bar{C}}x_Cx_{\bar{C}}}_{X\tX}\otimes\sum_{a_C:V^{|C|}(a_C,x_C)=1}\ket{a_C}_{A_C}\ket{\tilde{\vph}}_{EA_{\bar{C}}|x_Cx_{\bar{C}}a_C}\]
where $\ket{\tilde{\vph}}_{EA_{\bar{C}}|x_Cx_{\bar{C}}a_C}$ is a subnormalized state satisfying $\Vert\ket{\tilde{\vph}}_{EA_{\bar{C}}|x_Cx_{\bar{C}}a_C}\Vert_2^2 = \sfP_{A_C|x_Cx_{\bar{C}}}(a_C)$, and $\gamma = \Pr[\clE]$.

We shall use the following lemma, whose proof we give later.
\begin{lemma}\label{lem:i-conds}
Letting $R=X_CA_C$, the following conditions hold:
\begin{enumerate}[(i)]
\item If $c < 1$, for every $i\in\bar{C}$ and $j\in[l]$, there exist isometries $U^j_{i}$ taking registers $X^j_{\bar{C}}\tX^j_{\bar{C}}E^jA^j_{\bar{C}}$ to $\tX^j_{\bar{C}}E^jA^j_{\bar{C}}$ such that
\begin{align*}
& \bbE_{i\in\bar{C}}\bigg\Vert \Big(\bigotimes_{j\in[l]}U^j_{i}\Big)\left(\state{\psi}_{X'_iX_i}\otimes\state{\vph}_{X_{\bar{C}}\tX_{\bar{C}}EA_{\bar{C}}R}\right)\Big(\bigotimes_{j\in[l]}(U^j_{i})^\dagger\Big) - \state{\vph}_{X_{\bar{C}}\tX_{\bar{C}}EA_{\bar{C}}R}\bigg\Vert_1 \\
& \leq 2\sqrt{2lc} + 2l\sqrt{2\delta};
\end{align*}
\label{item:unit-i-close}
\item In case (ii): $1 \leq c < \frac{\zeta^2}{270l^3}\eff_{\eps+\zeta}(\sfV,p)$ and $\delta < 1$, there exists an $i\in\bar{C}$ such that for every $j\in[l]$, there exist measurement operators $M^j_i$ taking registers $X^j_i\tX^j_{\bar{C}}E^jA^j_{\bar{C}}$ to $\tX^j_{\bar{C}}E^jA^j_{\bar{C}}$ (with $M^j_i(M^j_i)^\dagger$ being the POVM element), such that each $\bigotimes_{j\in[l]} M^j_i$ succeeds on $\ket{\psi}_{X'_iX_i}\otimes\ket{\vph}_{\tX_{\bar{C}}EA_{\bar{C}}R}$ with probability $\alpha_i \geq 2^{-\frac{270l^3c}{\zeta^2}}$, and
\[
\bigg\Vert \frac{1}{\alpha_i}\Big(\bigotimes_{j\in[l]}M^j_i\Big)\left(\state{\psi}_{X'_iX_i}\otimes\state{\vph}_{\tX_{\bar{C}}EA_{\bar{C}}R}\right)\Big(\bigotimes_{j\in[l]}(M^j_i)^\dagger\Big) - \state{\vph}_{X'_i\tX_{\bar{C}}EA_{\bar{C}}R}\bigg\Vert_1 \leq 2\zeta
\]
where $\ket{\psi}_{X'_iX_i} = \sum_{x_i}\sqrt{\sfP_{X_i}(x_i)}\ket{x_ix_i}_{X'_iX_i}$, $\ket{\vph}_{X'_i\tX_{\bar{C}}EA_{\bar{C}}R}$ is the same state as $\ket{\vph}_{X'_i\tX_{\bar{C}}EA_{\bar{C}}R}$ with the $X_i$ register replaced by the $X'_i$ register. 
\label{state-i-close}
\end{enumerate}
\end{lemma}

In order to prove Lemma \ref{lm:induct} assuming Lemma \ref{lem:i-conds}, we shall view the inputs to a strategy for $G(p,\sfV)$ or a zero-communication protocol for $\sfV$ as being quantum, instead of coming from a classical distribution. If the players receive inputs from a distribution $\sfP_Y = \sfP_{Y^1\ldots Y^l}$, we can think of them as receiving registers $Y^1, \ldots, Y^l$ respectively of a pure state
\[ \ket{\sigma}_{Y'Y} = \sum_y\sqrt{\sfP_Y(y)}\ket{yy}_{Y'Y}\]
with say a referee holding the $Y'$ registers (which the players cannot touch). The players hold a shared entangled state $\ket{\rho}_{EA} = \ket{\rho}_{E^1\ldots E'^lA^1\ldots A^l}$, with player $j$ holding $E^jA^j$, $A^l$ being the answer register. Player $j$ now applies some isometry or measurement on registers $Y^jE^jA^j$ to determine their output. Strictly speaking, this isometry or measurement should only use $Y^j$ as a control register, since it is classical. But player $j$ can always copy over $Y^j$ to a different register $\tY^j$ and apply a general isometry or measurement on $\tY^jE^jA^j$ --- the effect of this will be the same as applying a general isometry or measurement on $Y^jE^jA^j$ that does not use $Y^j$ as a control register. So we shall assume that player $j$ can in fact apply a general isometry or measurement on $Y^jE^jA^j$. Note that in this case we are talking we are talking about isometries instead of unitaries because these need not preserve the $Y^j$ (or other) registers.

For zero-communication protocol with aborts, we shall also assume the players first apply a measurement to decide whether they will abort or not abort, and conditioned on not aborting, do another measurement to give outputs in $\clA^1\times\ldots\times\clA^l$ (in general they can do a single measurement to decide their output, which may be abort, or some element of $\clA^j$, but the protocol $\clP'$ we describe will have two measurements). In fact they do not need to actually do this last measurement in order for us to determine the average success probability: we can assume that the state conditioned on not aborting already has the correlations they want between the registers $Y'$ and $A$ (the $Y^j$ registers may have been modified by the measurement), and the average success probability is determined by computing $\sfV$ on $Y'A$ of the state conditioned on not aborting. That is, suppose the measurement operator corresponding to not abort for player $j$ is $M^j$. Then the average probability of not aborting in the protocol is the success probability $\alpha$ of $\bigotimes_{j\in[l]}M^j$ on $\ket{\sigma}_{Y'Y}\otimes\ket{\rho}_{EA}$. And the average success probability of the protocol conditioned on not aborting is determined by computing $\sfV$ on the $Y'A$ registers of $\frac{1}{\sqrt{\alpha}}\left(\bigotimes_{j\in[l]}M^j\right)\ket{\sigma}_{Y'Y}\otimes\ket{\rho}_{EA}$. Similarly, for games, the average success probability will be determined by computing $\sfV$ on the $Y'A$ registers of $\bigotimes_{j\in[l]}U^j\ket{\sigma}_{Y'Y}\otimes\ket{\rho}_{EA}$, where $U^j$ is the isometry applied by the $j$-th player.

\subsubsection{Case (i): $c < 1$}
In this case, we shall use the isometries given by Lemma \ref{lem:i-conds} to give a quantum strategy $\clS$ for $G(p,\sfV)$, whose winning probability is at least
\[ \bbE_{i\in\bar{C}}\Pr[T_i=1|\clE] - \sqrt{2lc} - l\sqrt{2\delta}.\]
By the definition of $\omega^*(G(p,\sfV))$, $\clS$ cannot have success probability more than $\omega^*(G(p,\sfV))$. This gives the required upper bound on $\bbE_{i\in\bar{C}}\Pr[T_i=1|\clE]$.

The strategy $\clS$ is as follows:
\begin{itemize}
\item The players share $\log|\bar{C}|$ uniformly random bits as randomness (which we shall treat as classical), and $\ket{\vph}_{\tX_{\bar{C}}EA_{\bar{C}}R}$ as shared entanglement, with player $j$ holding the registers $\tX^j_{\bar{C}}E^jA^j_{\bar{C}}$ (the extra $R$ register can go to any player, say the first, but they won't need to do anything on it).
\item The players receive inputs as the $X^j_i$ register of $\ket{\psi}_{X'_iX_i}$ (note that the distribution in this state is the correct one, $\mu$).
\item The players use their shared randomness to sample a uniformly random $i\in\bar{C}$.
\item Player $j$ applies the isometry $U^j_i$ on the registers $X^j_i\tX^j_{\bar{C}}E^jA^j_{\bar{C}}$ according to the sampled $i$.
\item Player $j$ provides $A^j_i$ of the resulting state as their answer register.
\end{itemize}
The state obtained by the players at the end of this strategy is obviously
\[ \bbE_{i\in\bar{C}}\Big(\bigotimes_{j\in[l]}U^j_{i}\Big)\left(\state{\psi}_{X'_iX_i}\otimes\state{\vph}_{X_{\bar{C}}\tX_{\bar{C}}EA_{\bar{C}}R}\right)\Big(\bigotimes_{j\in[l]}(U^j_{i})^\dagger\Big).\]
Note that if $i$ is sampled uniformly in $\bar{C}$, and $\sfV$ is computed on the $X'_iA_i$ registers of $\ket{\vph}_{X'_i\tX_{\bar{C}}EA_{\bar{C}}R}$, then the average success probability is $\bbE_{i\in\bar{C}}\Pr[T_i=1|\clE]$. By condition (ii) of Lemma \ref{lem:i-conds}, then the average success probability of $\clS$ is as claimed.

\subsubsection{Case (ii): $c \geq 1$}
In this case, we shall use the measurement operators given by Lemma \ref{lem:i-conds} to give a zero-communication protocol $\clP'$ for $\sfV$ whose average probability of not aborting is at least $2^{-\frac{270l^3 c}{\zeta^2}} > 1/\eff_{\eps+\zeta}(\sfV,p)$ (by the condition on $cn$), and conditioned on not aborting, is correct with probability at least
\[ \Pr[T_i=1|\clE] - \zeta \]
(with the $i$ provided by condition (ii)) averaged on inputs from $p$. By the definition of $\eff_{\eps+\zeta}(\sfV,p)$, $\clP'$ cannot be correct conditioned on not aborting with probability more than $1-(\eps+\zeta)$ when inputs come from $p$. This gives the required upper bound on $\Pr[T_i=1|\clE]$.

The protocol $\clP'$ is as follows:
\begin{itemize}
\item The players share $\ket{\vph}_{\tX_{\bar{C}}EA_{\bar{C}}R}$ as shared entanglement, with player $j$ holding the registers $\tX^j_{\bar{C}}E^jA^j_{\bar{C}}$ (the extra $R$ register can go to any player, say the first, but they won't need to do anything on it).
\item The players receive inputs as the $X^j_i$ register of $\ket{\psi}_{X'_iX_i}$ (note that the distribution in this state is the correct one, $p$).
\item Player $j$ applies measurements $\{M^j_i(M^j_i)^\dagger,\Id-M^j_i(M^j_i)^\dagger\}$ on the registers $X^j_i\tX^j_{\bar{C}}E^jA^j_{\bar{C}}$ and declares not abort if the $M^j_i$ measurement succeeds.
\item Conditioned on not aborting, player $j$ provides $A^j_i$ as their answer register.
\end{itemize}
By our description above, and condition \ref{state-i-close}, the average probability of not aborting in this protocol is $\alpha_i \geq 2^{-\frac{270l^3c}{\zeta^2}} > \frac{1}{\eff_{\eps+\zeta}(\sfV,p)}$ by the condition on $c$. Now note that if $\sfV$ is computed in the $X'_iA_i$ register of $\ket{\vph}_{X'_i\tX_{\bar{C}}EA_{\bar{C}}}$, the average success probability is by definition $\Pr[T_i=1|\clE]$. Since by condition \ref{state-i-close},
\[ \bigg\Vert \frac{1}{\alpha_i}\Big(\bigotimes_{j\in[l]}M^j_i\Big)\left(\state{\psi}_{X'_iX_i}\otimes\state{\vph}_{\tX_{\bar{C}}EA_{\bar{C}}R}\right)\Big(\bigotimes_{j\in[l]}(M^j_i)^\dagger\Big) - \state{\vph}_{X'_i\tX_{\bar{C}}EA_{\bar{C}}R}\bigg\Vert_1 \leq 2\zeta
\]
the average success probability on $\frac{1}{\sqrt{\alpha_i}}\Big(\bigotimes_{j\in[l]}M^j_i\Big)\ket{\psi}_{X'_iX_i}\otimes\ket{\vph}_{\tX_{\bar{C}}EA_{\bar{C}}R}$, that is, the average success probability of $\clP'$ conditioned on not aborting, is at least $\Pr[T_i=1|\clE] - \zeta$.

\subsection{Proof of Lemma \ref{lem:i-conds}}
The first part of the proof goes the same way for both cases (i) and (ii). We shall proceed with a common proof and then diverge when required.

Since player $j$'s communication in $\clP$ is $c^jn$ bits, by Lemma \ref{int-holevo} for the final state $\ket{\psi}$ of $\clP$, there exists a state $\rho^j_{X^{-j}\tX^{-j}E^{-j}A^{-j}}$ such that
\[ \sfD_\infty\left(\psi_{X^jX^{-j}\tX^{-j} E^{-j}A^{-j}}\middle\Vert\psi_{X^j}\otimes\rho^j_{X^{-j}\tX^{-j}E^{-j}A^{-j}}\right) \leq 2c^j n.\]
Using Facts \ref{fc:event-prob} and \ref{fc:Sinfty-tri}, this gives us
\begin{align}
& \bbE_{\sfP_{R|\clE}}\sfD\left(\vph_{X^j_{\bar{C}}X^{-j}_{\bar{C}}\tX^{-j}_{\bar{C}} E^{-j}A^{-j}_{\bar{C}}|r}\middle\Vert\psi_{X^j_{\bar{C}}}\otimes\rho^j_{X^{-j}_{\bar{C}}\tX^{-j}_{\bar{C}}E^{-j}A^{-j}_{\bar{C}}}\right) \nonumber \\
& = \bbE_{\sfP_{X_CA_C|\clE}}\sfD\left(\vph_{X^j_{\bar{C}}X^{-j}_{\bar{C}}\tX^{-j}_{\bar{C}} E^{-j}A^{-j}_{\bar{C}}|x_Ca_C}\middle\Vert\psi_{X^j_{\bar{C}}}\otimes\rho^j_{X^{-j}_{\bar{C}}\tX^{-j}_{\bar{C}}E^{-j}A^{-j}_{\bar{C}}}\right) \nonumber \\
& \leq \bbE_{\sfP_{A_C|\clE}}\sfD\left(\vph_{X^jX^{-j}\tX^{-j} E^{-j}A^{-j}|a_C}\middle\Vert\psi_{X^j}\otimes\rho^j_{X^{-j}\tX^{-j}E^{-j}A^{-j}}\right) \nonumber \\
& \leq \bbE_{\sfP_{A_C|\clE}}\sfD_\infty\left(\vph_{X^jX^{-j}\tX^{-j} E^{-j}A^{-j}|a_C}\middle\Vert\psi_{X^j}\otimes\rho^j_{X^{-j}\tX^{-j}E^{-j}A^{-j}}\right) \nonumber \\
& \leq \bbE_{\sfP_{X_CA_C|\clE}}\left[\sfD_\infty\left(\vph_{X^jX^{-j}\tX^{-j} E^{-j}A^{-j}|a_C}\middle\Vert\vph_{X^jX^{-j}\tX^{-j} E^{-j}A^{-j}}\right)\right. \nonumber \\
& \quad \left. + \sfD_\infty\left(\vph_{X^jX^{-j}\tX^{-j} E^{-j}A^{-j}}\middle\Vert\psi_{X^jX^{-j}\tX^{-j} E^{-j}A^{-j}}\right) \right. \nonumber \\
& \quad \left. + \sfD_\infty\left(\psi_{X^jX^{-j}\tX^{-j} E^{-j}A^{-j}}\middle\Vert\psi_{X^j}\otimes\rho^j_{X^{-j}\tX^{-j}E^{-j}A^{-j}}\right) \right] \nonumber \\
& \leq \bbE_{\sfP_{X_CA_C|\clE}}\left[\log(1/\sfP_{A_C|\clE}(a_C)) + \log(1/\Pr[\clE]) + 2c^jn\right] \nonumber \\
& \leq \bbE_{\sfP_{X_CA_C|\clE}}\left[|C|\cdot\log\left(|\clA^1|\cdot\ldots\cdot|\clA^l|\right) + \log(1/\Pr[\clE]) + 2c^jn\right] \nonumber \\
& = (\delta + 2c^j)n. \label{eq:X-Dinfty-1}
\end{align}

Similarly we also have,
\begin{equation}\label{eq:phi-XR}
\sfD\left(\vph_{X_{\bar{C}}R}\middle\Vert\psi_{X_{\bar{C}}}\otimes\vph_R\right) = \bbE_{\sfP_{R|\clE}}\sfD\left(\vph_{X_{\bar{C}}|r}\middle\Vert\psi_{X_{\bar{C}}}\right) \leq \delta n.
\end{equation}
Using the Quantum Gibb's inequality on \eqref{eq:X-Dinfty-1} we have,
\begin{align*}
\bbE_{\sfP_{X^j_{\bar{C}}R|\clE}}\sfD\left(\vph_{X^{-j}_{\bar{C}}\tX^{-j}_{\bar{C}} E^{-j}A^{-j}_{\bar{C}}|x^j_{\bar{C}}r}\middle\Vert\vph_{X^{-j}_{\bar{C}}\tX^{-j}_{\bar{C}}E^{-j}A^{-j}_{\bar{C}}|r}\right) & = \bbE_{\sfP_{R|\clE}}\sfD\left(\vph_{X^j_{\bar{C}}X^{-j}_{\bar{C}}\tX^{-j}_{\bar{C}} E^{-j}A^{-j}_{\bar{C}}|r}\middle\Vert\vph_{X^j_{\bar{C}}}\otimes\vph_{X^{-j}_{\bar{C}}\tX^{-j}_{\bar{C}}E^{-j}A^{-j}_{\bar{C}}|r}\right) \\
& \leq \bbE_{\sfP_{R|\clE}}\sfD\left(\vph_{X^j_{\bar{C}}X^{-j}_{\bar{C}}\tX^{-j}_{\bar{C}} E^{-j}A^{-j}_{\bar{C}}|r}\middle\Vert\psi_{X^j_{\bar{C}}}\otimes\vph_{X^{-j}_{\bar{C}}\tX^{-j}_{\bar{C}}E^{-j}A^{-j}_{\bar{C}}|r}\right) \\
& \leq (2c^j+\delta)n.
\end{align*}
Hence by the chain rule of relative entropy on this and \eqref{eq:phi-XR},
\[ \sfD\left(\vph_{X^j_{\bar{C}}X^{-j}_{\bar{C}}\tX^{-j}_{\bar{C}} E^{-j}A^{-j}_{\bar{C}}R}\middle\Vert\psi_{X^j_{\bar{C}}}\otimes\vph_{X^{-j}_{\bar{C}}\tX^{-j}_{\bar{C}}E^{-j}A^{-j}_{\bar{C}}R}\right)\leq 2(c^j+\delta)n.\]
Using $\bar{C}_{<i}$ to denote the set of coordinates if $\bar{C}$ less than $i$, we have by the chain rule of relative entropy again,
\begin{align}
4(c^j+\delta) & \geq \bbE_{i\in\bar{C}}\bbE_{\sfP_{X^j_{\bar{C}_{<i}}}}\sfD\left(\vph_{X^j_iX^{-j}_{\bar{C}}\tX^{-j}_{\bar{C}} E^{-j}A^{-j}_{\bar{C}}R|x_{\bar{C}_{<i}}}\middle\Vert\psi_{X^j_i}\otimes\vph_{X^{-j}_{\bar{C}}\tX^{-j}_{\bar{C}}E^{-j}A^{-j}_{\bar{C}}R}\right) \nonumber \\
 & \geq \bbE_{i\in\bar{C}}\sfD\left(\vph_{X^j_iX^{-j}_{\bar{C}}\tX^{-j}_{\bar{C}} E^{-j}A^{-j}_{\bar{C}}R}\middle\Vert\psi_{X^j_i}\otimes\vph_{X^{-j}_{\bar{C}}\tX^{-j}_{\bar{C}}E^{-j}A^{-j}_{\bar{C}}R}\right) \label{eq:XR-Dinfty}
\end{align}
where we have used Fact \ref{fc:cond-dec}.

\subsubsection{Case (i): $c < 1$}
We shall apply Pinsker's inequality (Fact \ref{fc:pinsker}) on \eqref{eq:XR-Dinfty} for all $j\in[l]$, after tracing out $X^{-j}_{\bar{C}}$. After also applying Jensen's inequality with respect to the expectation over $i$, this gives us
\[ \bbE_{i\in\bar{C}}\sfB\left(\vph_{X^j_i\tX^{-j}_{\bar{C}}E^{-j}A^{-j}_{\bar{C}}R}, \psi_{X^j_i}\otimes\vph_{\tX^{-j}_{\bar{C}}E^{-j}A^{-j}_{\bar{C}}R}\right) \leq 2\sqrt{c^j+\delta}.\]
Now since $X'^j_i$ as used in $\ket{\psi}_{X'^j_iX^j_i}$ and $\ket{\vph}_{X'^j_i\tX_{\bar{C}}EA_{\bar{C}}R}$ in Lemma \ref{lem:i-conds}, is identical to $X^j_i$, we also have,
\begin{equation}\label{eq:X-pinsker}
\bbE_{i\in\bar{C}}\sfB\left(\vph_{X'^j_i\tX^{-j}_{\bar{C}}E^{-j}A^{-j}_{\bar{C}}R}, \psi_{X'^j_i}\otimes\vph_{\tX^{-j}_{\bar{C}}E^{-j}A^{-j}_{\bar{C}}R}\right) \leq 2\sqrt{c^j+\delta}.
\end{equation}
Note that $\ket{\vph}_{X'^j_i\tX_{\bar{C}}EA_{\bar{C}}R}$ is a purification of the state in the first argument of the Bures distance in \eqref{eq:X-pinsker}, and $\ket{\psi}_{X'^j_iX^j_i}\otimes\ket{\vph}_{\tX_{\bar{C}}EA_{\bar{C}}R}$ is a purification of the state in the second argument (since $\sfP_{X^1_i\ldots X^l_i}$ is a product distribution, each $\psi_{X'^j_iX^j_i}$ is pure). Therefore, by Uhlmann's theorem, there exist isometries $\{U^j_i\}_i$ taking registers $X^j_i\tX^j_{\bar{C}}E^jA^j_{\bar{C}}$ to $\tX^j_{\bar{C}}E^jA^j_{\bar{C}}$ such that
\[ \bbE_{i\in\bar{C}}\sfB\left(U^j_i\otimes\Id\left(\state{\psi}_{X'^j_iX^j_i}\otimes\state{\vph}_{\tX_{\bar{C}}EA_{\bar{C}}R}\right)(U^j_i)^\dagger\otimes\Id, \state{\vph}_{X'^j_i\tX_{\bar{C}}EA_{\bar{C}}R}\right) \leq 2\sqrt{2c^j+\delta}.\]
Since $U^j_i$ do not act on the $X^{-j}_i$ registers, we in fact have,
\begin{align}
& \bbE_{i\in\bar{C}}\sfB\left(U^j_i\otimes\Id\left(\state{\psi}_{X'^{\geq j}_iX^{\geq j}_i}\otimes\state{\vph}_{X'^{<j}_i\tX_{\bar{C}}EA_{\bar{C}}R}\right)(U^j_i)^\dagger\otimes\Id, \state{\psi}_{X^{>j}_iX^{>j}_i}\otimes\state{\vph}_{X'^{\leq j}_i\tX_{\bar{C}}EA_{\bar{C}}R}\right) \nonumber \\
& \leq 2\sqrt{2c^j+\delta}. \label{eq:X-pinsker-2}
\end{align}
Now we need to find the action of $\bigotimes_{j\in[l]}U^j_i$ on $\ket{\psi}_{X'_iX_i}\otimes\ket{\vph}_{\tX_{\bar{C}}EA_{\bar{C}}R}$. By the triangle inequality, we see that
\begin{align*}
& \bbE_{i\in\bar{C}}\sfB\Bigg(\Big(\bigotimes_{j\in[l]}U^j_{i}\Big)\left(\state{\psi}_{X'_iX_i}\otimes\state{\vph}_{X_{\bar{C}}\tX_{\bar{C}}EA_{\bar{C}}R}\right)\Big(\bigotimes_{j\in[l]}(U^j_{i})^\dagger\Big), \state{\vph}_{X_{\bar{C}}\tX_{\bar{C}}EA_{\bar{C}}R}\Bigg) \\
& \leq \bbE_{i\in\bar{C}}\sum_{k=1}^l\sfB\Bigg(\Big(\bigotimes_{j>k}U^j_i\Big)\Big(U^k_i\otimes\Id\left(\state{\psi}_{X'^{\geq k}_iX^{\geq k}_i}\otimes\state{\vph}_{X'^{>k}_i\tX_{\bar{C}}EA_{\bar{C}}R}\right)(U^k_i)^\dagger\otimes\Id\Big)\Big(\bigotimes_{j>k}(U^j_i)^\dagger\Big), \\
& \qquad \qquad \quad \Big(\bigotimes_{j>k}U^j_i\Big)\left(\state{\psi}_{X^{>k}_iX^{>k}_i}\otimes\state{\vph}_{X'^{\leq k}_i\tX_{\bar{C}}EA_{\bar{C}}R}\right) \Big(\bigotimes_{j>k}(U^j_i)^\dagger\Big)\Bigg) \\
& = \sum_{k=1}^l\bbE_{i\in\bar{C}}\sfB\left(U^k_i\otimes\Id\left(\state{\psi}_{X'^{\geq k}_iX^{\geq k}_i}\otimes\state{\vph}_{X'^{<k}_i\tX_{\bar{C}}EA_{\bar{C}}R}\right)(U^k_i)^\dagger\otimes\Id, \right. \\
& \qquad \qquad \quad \left. \state{\psi}_{X^{>k}_iX^{>k}_i}\otimes\state{\vph}_{X'^{\leq k}_i\tX_{\bar{C}}EA_{\bar{C}}R}\right) \\
& \overset{{\color{red}a}}{\leq} \sum_{k=1}^l2\sqrt{c^k+\delta} \\
& \leq 2\sum_{k=1}^l\left(\sqrt{c^k} + \sqrt{\delta}\right) \\
& \overset{{\color{red}b}}{\leq} 2\sqrt{l\sum_{k=1}^lc^k} + 2l\sqrt{\delta} = 2\sqrt{lc} +2l\sqrt{\delta}
\end{align*}
where in ${\color{red}a}$ we have used \eqref{eq:X-pinsker-2} with $j=k$, and in ${\color{red}b}$ we have used the Cauchy-Schwarz inequality. Finally, by applying the Fuchs-van de Graaf inequality, we have,
\[ \bbE_{i\in\bar{C}}\norm{\Big(\bigotimes_{j\in[l]}U^j_{i}\Big)\left(\state{\psi}_{X'_iX_i}\otimes\state{\vph}_{X_{\bar{C}}\tX_{\bar{C}}EA_{\bar{C}}R}\right)\Big(\bigotimes_{j\in[l]}(U^j_{i})^\dagger\Big), \state{\vph}_{X_{\bar{C}}\tX_{\bar{C}}EA_{\bar{C}}R}}_1 \leq 2\sqrt{2lc} + 2l\sqrt{2\delta},\]
which proves condition (i) of Lemma \ref{lem:i-conds}.

\subsubsection{Case (ii): $c \geq 1$}
Using the Quantum Substate Theorem on \eqref{eq:XR-Dinfty} and tracing out $X^{-j}_{\bar{C}}$ we get for all $j\in[l]$,
\[
\bbE_{i\in\bar{C}}\sfD^{\sqrt{2\zeta'},\Delta}_\infty\left(\vph_{X^j_i\tX^{-j}_{\bar{C}}E^{-j}A^{-j}_{\bar{C}}R}\middle\Vert\psi_{X^j_i}\otimes\vph_{\tX^{-j}_{\bar{C}}E^{-j}A^{-j}_{\bar{C}}R}\right) \leq \frac{4c^j + 4\delta + 1}{\zeta'} + \log\left(\frac{1}{1-\zeta'}\right)
\]
for some $\zeta'$ to be fixed later. Now since $X'^j_i$ as used in $\ket{\psi}_{X'^j_iX^j_i}$ and $\ket{\vph}_{X'^j_i\tX_{\bar{C}}EA_{\bar{C}}R}$ in condition (ii) of Lemma \ref{lem:i-conds}, is identical to $X^j_i$, we also have,
\begin{equation}\label{eq:X-Dinfty}
\bbE_{i\in\bar{C}}\sfD^{\sqrt{2\zeta'},\Delta}_\infty\left(\vph_{X'^j_i\tX^{-j}_{\bar{C}}E^{-j}A^{-j}_{\bar{C}}R}\middle\Vert\psi_{X'^j_i}\otimes\vph_{\tX^{-j}_{\bar{C}}E^{-j}A^{-j}_{\bar{C}}R}\right) \leq \frac{4c^j + 4\delta + 1}{\zeta'} + \log\left(\frac{1}{1-\zeta'}\right)
\end{equation}

To find the measurement operators $M^j_i$, we shall do induction on the number of players. In particular we shall prove the following lemma. 
\begin{lemma}\label{lem:proj-induct}
Suppose we have measurement operators $\Big\{M^j_i\Big\}_i$ for $j\in[k], i \in \bar{C}$, $ 0 \leq k < l$, taking registers $X^j_i\tX^j_{\bar{C}}E^jA^j_{\bar{C}}$ to $\tX^j_{\bar{C}}E^jA^j_{\bar{C}}$ respectively, such that $\bigotimes_{j\in[k]}M^j_{i,r}$ succeeds on $\left(\bigotimes_{j\in[k]}\ket{\psi}_{X'^j_iX^j_i}\right)\otimes\ket{\vph}_{\tX_{\bar{C}}EA_{\bar{C}}R}$ with probability $\alpha^{\leq k}_i = 2^{-\sum_{j=1}^k\tilde{c}^j_i}$ where
\[ \bbE_{i\in\bar{C}}\tilde{c}^j_{i} \leq \frac{15c^j}{\zeta'},\]
and for all $i\in\bar{C}$,
\begin{align}\label{eq:Xk-dist}
& \Delta\Bigg(\frac{1}{\alpha^{\leq k}_{i}}\bigg(\bigotimes_{j\in[k]}M^j_{i}\otimes\Id\bigg)\bigg(\bigotimes_{j\in[k]}\state{\psi}_{X'^j_iX^j_i}\otimes\state{\vph}_{\tX_{\bar{C}}EA_{\bar{C}}R}\bigg)\bigg(\bigotimes_{j\in[k]}(M^j_{i})^\dagger\otimes\Id\bigg), \nonumber \\
& \state{\psi}_{X'^{>k}_iX^{>k}_i}\otimes\state{\vph}_{X'^{\leq k}_i\tX_{\bar{C}}EA_{\bar{C}}R}\Bigg) \leq (3k - 2)\sqrt{2\zeta'}.
\end{align}
Then there are measurement operators $\Big\{M^{k+1}_{i}\Big\}_{i}$ taking registers $X^{k+1}_i\tX^{k+1}_{\bar{C}}E^{k+1}A^{k+1}_{\bar{C}}$ to $\tX^{k+1}_{\bar{C}}E^{k+1}A^{k+1}_{\bar{C}}$, such that $\bigotimes_{j\in[k+1]}M^j_{i}$ succeeds on $\left(\bigotimes_{j\in[k+1]}\ket{\psi}_{X'^j_iX^j_i}\right)\otimes\ket{\vph}_{\tX_{\bar{C}}EA_{\bar{C}}R}$ with probability $\alpha^{\leq (k+1)}_{i} = \alpha^{k+1}_{i}\alpha^{\leq k}_{i}$ where $\alpha^{k+1}_{i} = 2^{-\tilde{c}^{k+1}_{i}}$, with
\[ \bbE_{i\in\bar{C}}\tilde{c}^{k+1}_{i} \leq \frac{15c^{k+1}}{\zeta'},\]
and for all $i\in\bar{C}$
\begin{align*}
& \Delta\Bigg(\frac{1}{\alpha^{\leq (k+1)}_{i}}\bigg(\bigotimes_{j\in[k+1]}M^j_{i}\otimes\Id\bigg)\bigg(\bigotimes_{j\in[k+1]}\state{\psi}_{X'^j_iX^j_i}\otimes\state{\vph}_{\tX_{\bar{C}}EA_{\bar{C}}R}\bigg)\bigg(\bigotimes_{j\in[k+1]}(M^j_{i})^\dagger\otimes\Id\bigg), \nonumber \\
& \state{\psi}_{X'^{>(k+1)}_iX^{>(k+1)}_i}\otimes\state{\vph}_{X'^{\leq (k+1)}_i\tX_{\bar{C}}EA_{\bar{C}}R}\Bigg) \leq (3k+1)\sqrt{2\zeta'}.
\end{align*}
\end{lemma}
We also clarify that if the distance in \eqref{eq:Xk-dist} is $\Delta[k]$, then the way we pick our parameters in the proof of the lemma gives us $\Delta[k+1] = \Delta[k] + 3\sqrt{2\zeta'}$. The expression $(3k-2)\sqrt{2\zeta'}$ is obtained by setting $\Delta[1] = \sqrt{2\zeta'}$.

\begin{proof}[Proof of Lemma \ref{lem:proj-induct}]
Let
\[ \ket{\rho}_{X'_iX^{>k}_i\tX_{\bar{C}}EA_{\bar{C}}R} = \frac{1}{\sqrt{\alpha^{\leq k}_{i}}}\bigg(\bigotimes_{j\in[k]}M^j_{i}\otimes\Id\bigg)\bigg(\bigotimes_{j\in[k]}\ket{\psi}_{X^j_iX'^j_i}\otimes\ket{\vph}_{\tX_{\bar{C}}EA_{\bar{C}}|r}\bigg).\]
Note that $\ket{\rho}$ has an $i$ dependence, but we are not writing it explicitly. By \eqref{eq:Xk-dist},
\[
\bbE_{i\in\bar{C}}\Delta\left(\rho_{\tX^{-(k+1)}_{\bar{C}}E^{-(k+1)}A^{-(k+1)}_{\bar{C}}R}, \vph_{\tX^{-(k+1)}_{\bar{C}}E^{-(k+1)}A^{-(k+1)}_{\bar{C}}R}\right) \leq \Delta[k].
\]
Moreover, since none of the operators $M^j_i$ for $j\in[k]$ act on the $X^{k+1}_i$ register,
\[ \rho_{X'^{k+1}_i\tX^{-(k+1)}_{\bar{C}}E^{-(k+1)}A^{-(k+1)}_{\bar{C}}R} = \rho_{X'^{k+1}_i}\otimes\rho_{\tX^{-(k+1)}_{\bar{C}}E^{-(k+1)}A^{-(k+1)}_{\bar{C}}R} = \psi_{X'^{k+1}_i}\otimes\rho_{\tX^{-(k+1)}_{\bar{C}}E^{-(k+1)}A^{-(k+1)}_{\bar{C}}R}.\]
Using the Substate Perturbation Lemma on the above and \eqref{eq:X-Dinfty} with $j=k+1$, picking parameters $\eps = \delta_0 = \sqrt{2\zeta'}, \delta_1= \Delta[k]$ we get,
\begin{align*}
& \bbE_{i\in\bar{C}}\sfD^{\Delta[k+1],\Delta}_\infty\left(\vph_{X'^{k+1}_i\tX^{-(k+1)}_{\bar{C}}E^{-(k+1)}A^{-(k+1)}_{\bar{C}}R}\middle\Vert\rho_{X'^{k+1}_i\tX^{-(k+1)}_{\bar{C}}E^{-(k+1)}A^{-(k+1)}_{\bar{C}}R}\right) \\
& = \bbE_{i\in\bar{C}}\sfD^{3\sqrt{2\zeta'}+\Delta[k],\Delta}_\infty\left(\vph_{X'^{k+1}_i\tX^{-(k+1)}_{\bar{C}}E^{-(k+1)}A^{-(k+1)}_{\bar{C}}R}\middle\Vert\psi_{X'^{k+1}_i}\otimes\rho_{X'^{k+1}_i\tX^{-(k+1)}_{\bar{C}}E^{-(k+1)}A^{-(k+1)}_{\bar{C}}R}\right) \\
& \leq \frac{4c^{k+1} + 4\delta + 1}{\zeta'} + \log\left(\frac{1}{1-\zeta'}\right) + 1 + \log\left(1 + \frac{2}{\zeta'}\right) \\
& \leq \frac{4c^{k+1} + 4\delta + 1}{\zeta'} + 3\zeta' + 1 + \frac{2}{\zeta'} \leq \frac{15c^{k+1}}{\cdot\zeta'}.
\end{align*}
Now note that $\ket{\psi}_{X'^{>(k+1)}_iX^{>(k+1)}_i}\otimes\ket{\vph}_{X'^{\leq (k+1)}_i\tX_{\bar{C}}EA_{\bar{C}}R}$ is a purification of the state in the first argument in the above smoothed entropy, and $\ket{\rho}_{X'_iX_i\tX_{\bar{C}}EA_{\bar{C}}R}$ is obviously a purification of the state in the second. Therefore, by Fact \ref{fc:substate-proj}, there exist measurement operators $\Big\{M^{k+1}_{i}\Big\}_{i}$ taking registers $X^{k+1}_i\tX^{k+1}_{\bar{C}}E^{k+1}A^{k+1}_{\bar{C}}$ to $\tX^{k+1}_{\bar{C}}E^{k+1}A^{k+1}_{\bar{C}}$, that succeed on $\ket{\rho}_{X'_iX^{>k}_i\tX_{\bar{C}}EA_{\bar{C}}R}$ with probability $\alpha^{k+1}_{i} = 2^{-\tilde{c}^{k+1}_i}$, where
\[ \bbE_{i\in\bar{C}}\tilde{c}^{k+1}_{i} \leq \frac{15c^{k+1}}{\zeta'},\]
and for all $i$,
\begin{align*}
& \Delta\Bigg(\frac{1}{\alpha^{k+1}_{i}\alpha^{\leq k}_{i}}\bigg(\bigotimes_{j\in[k+1]}M^j_{i}\otimes\Id\bigg)\bigg(\bigotimes_{j\in[k+1]}\state{\psi}_{X'^j_iX^j_i}\otimes\state{\vph}_{\tX_{\bar{C}}EA_{\bar{C}}R}\bigg)\bigg(\bigotimes_{j\in[k+1]}(M^j_{i})^\dagger\otimes\Id\bigg), \nonumber \\
& \quad \state{\psi}_{X'^{>(k+1)}_iX^{>(k+1)}_i}\otimes\state{\vph}_{X'^{\leq (k+1)}_i\tX_{\bar{C}}EA_{\bar{C}}R}\Bigg) \\
& = \Delta\Bigg(\frac{1}{\alpha^{k+1}_{i}}M^{k+1}_{i}\otimes\Id\bigg(\state{\rho}_{X'_iX^{>k}_i\tX_{\bar{C}}EA_{\bar{C}}R}\bigg)(M^{k+1}_{i})^\dagger\otimes\Id, \state{\psi}_{X'^{>(k+1)}_iX^{>(k+1)}_i}\otimes\state{\vph}_{X'^{\leq (k+1)}_i\tX_{\bar{C}}EA_{\bar{C}}R}\Bigg) \\
& \leq \Delta[k+1].
\end{align*}
This proves the lemma.
\end{proof}

After the induction process, we have measurement operators $\Big\{M^{j}_{i}\Big\}_{i}$ for $j\in[l]$ and the conditions in the statement of Lemma \ref{lem:proj-induct} hold with $k=l$. Therefore, by the Fuchs-van de Graaf inequality,
\begin{align*}
& \bigg\Vert \frac{1}{\alpha_i}\Big(\bigotimes_{j\in[l]}M^j_i\Big)\left(\state{\psi}_{X'_iX_i}\otimes\state{\vph}_{\tX_{\bar{C}}EA_{\bar{C}}R}\right)\Big(\bigotimes_{j\in[l]}(M^j_i)^\dagger\Big) - \state{\vph}_{X'_i\tX_{\bar{C}}EA_{\bar{C}}R}\bigg\Vert_1 \\
& \leq 2(3l-2)\sqrt{2\zeta'}.
\end{align*}
Setting $(3l-2)\sqrt{2\zeta'} = \zeta$ we get, $\zeta' \geq \frac{\zeta^2}{18l^2}$. This gives us
\[ \bbE_{i\in\bar{C}}\sum_{j=1}^l\tilde{c}^j_{i} \leq \frac{270l^2}{\zeta^2}\sum_{j=1}^lc^j \leq \frac{270l^3 c}{\zeta^2}.\]
Since $2^{-x}$ is a convex function, by Jensen's inequality we have,
\[
\bbE_{i\in\bar{C}}\alpha_{i} = \bbE_{i\in\bar{C}}2^{-\sum_{j=1}^l\tilde{c}^j_{i}}\geq 2^{-\bbE_{i\in\bar{C}}\sum_{j=1}^l\tilde{c}^j_{i}} \geq 2^{-270l^3c/\zeta^2}.
\]
Therefore there exists an $i\in\bar{C}$ such that $\alpha_i \geq 2^{-270l^3c/\zeta^2}$. This proves condition \ref{state-i-close} in Lemma \ref{lem:i-conds}.


\section{Device-independent cryptography with leakage}\label{sec:QKD}
In this section, we prove Theorems \ref{thm:leaky-qkd}, \ref{thm:leaky-re} and \ref{thm:leaky-quantum}. The protocols for all three theorems will involve the Magic Square non-local game, so we first note some properties of this game.
\subsection{Properties of the Magic Square game}
\begin{definition}
The 2-player Magic Square game, denoted by $\MS$, is as follows:
\begin{itemize}
\item Alice and Bob receive respective inputs $x \in \{0,1,2\}$ and $y \in \{0,1,2\}$ independently and uniformly at random.
\item Alice outputs $a \in \{0,1\}^3$ such that $a[0]\oplus a[1]\oplus a[2] = 0$ and Bob outputs $b \in \{0,1\}^3$ such that $b[0]\oplus b[1]\oplus b[2] = 1$.
\item Alice and Bob win the game iff $a[y] = b[x]$.
\end{itemize}
\end{definition}
The classical value of the magic square game is $\omega(\MS) = {8}/{9}$, whereas the quantum value is $\omega^*(\MS)=1$.

Rao \cite{Rao10} proved a threshold version of parallel repetition for the classical value of 2-player non-local games, which upper bounds the probability of winning $(\omega(G)+\eta)n$ copies out of $n$ parallel copies of a game $G$ by classical players. The following fact is Rao's result applied to the Magic Square game.
\begin{fact}\label{fc:MS-thres}
The probability of classical players winning $(\frac{8}{9}+\eta)n$ many games out of $n$ parallel copies of $\MS$ is ${2^{-\Omega(\eta^3n)}}$.
\end{fact}
The threshold version of the parallel repetition can be proven from a statement similar to Lemma \ref{lm:induct} in the classical case, which is an intermediate step in the classical proof as well. In principle, it is possible to prove a classical version of this lemma with leakage between the players, although we will not be doing that here. Instead we can use a simpler argument like in the proof of Lemma \ref{lem:eff-lb} to get a threshold parallel repetition theorem for the classical value of $\MS$ with leakage, using Fact \ref{fc:MS-thres}. The idea is as in Lemma \ref{lem:eff-lb} to get a protocol without leakage from a protocol with leakage. The argument is even simpler in this case because the players are classical: the leaked messages are all classical functions of their inputs and previous messages. We shall not describe the argument in too much detail, since as stated earlier it is similar to the proof of Lemma \ref{lem:eff-lb}, and we also do a similar argument in the sequential quantum case in Lemma \ref{lem:re-Hmin} later. The idea is as before for the players to share randomness and use it to guess the value of the messages round by round; if the randomness is inconsistent at any point, they record a "failure" and output a random answer; otherwise the output according to the protocol with leakage. The probability that both players output according to the protocol with leakage is $2^{-cn}$ if $cn$ bits are leaked: this means that the threshold winning probability with leakage is at most $2^{cn}$ times the threshold winning probability without leakage. We thus have the following corollary of Fact \ref{fc:MS-thres}.
\begin{cor}\label{cor:MS-cl-leaky}
The probability of classical players, between whom $cn$ bits are interactively leaked, winning $(\frac{8}{9}+\eta)n$ many games out of $n$ parallel copies of $\MS$ is ${2^{-\Omega(\eta^3n)+cn}}$.
\end{cor}

We shall use a 3-player version of the Magic Square game, defined below, in order to prove security for DIQKD and DIRE.
\begin{definition}
The 3-player variant of the Magic Square game, denoted by $\MSe$, is as follows:
\begin{itemize}
\item Alice receives inputs $x\in\{0,1,2\}, z \in \{0,1\}$ and Bob receives input $y\in\{0,1,2\}$ independently and uniformly at random; Eve receives $x', y' \in \{0,1,2\}$ and $z'\in\{0,1\}$, independently and uniformly at random.
\item Alice outputs $a \in \{0,1\}^3$ such that $a[0]\oplus a[1]\oplus a[2] = 0$, Bob outputs $b\in\{0,1\}^3$ such that $b[0]\oplus b[1]\oplus b[2] = 1$, and Eve outputs $c\in\{0,1\}$.
\item Alice, Bob and Eve win the game iff
\[ (x=x')\land(y=y')\land(a[y]=c)\land((a[y]=b[x])\lor(z=z')).\]
\end{itemize}
\end{definition}

\begin{fact}[\cite{JMS17}]\label{fc:MSE-w}
There is a constant $0 < \nu < 1$ such that $\omega^*(\MSe) = \frac{1}{9}(1-\nu)$.
\end{fact}
The above fact is a consequence of Proposition 4.1 in \cite{JMS17}. The game considered in the statement of this proposition in \cite{JMS17} is different: they consider a 6-player game between Alice, Bob, $\text{Alice}'$, $\text{Bob}'$, Charlie and $\text{Charlie}'$. Here we have given Charlie's role to Alice, and merged $\text{Alice}'$, $\text{Bob}'$ and $\text{Charlie}'$ into Eve (this is later done in the analysis in \cite{JMS17} anyway). Doing this makes no difference in the proof of the game's winning probability as given in \cite{JMS17}. Alternatively, the fact can be seen as a consequence of Lemma 2 in \cite{Vid17}. The game considered in \cite{Vid17} does not include Eve having inputs $x', y',z'$, which are required to be equal to $x, y, z$. Suppose the probability of winning Vidick's game is $(1-\nu')$. Since the probability of $z=z'$ is $\frac{1}{2}$, the probability of winning the version of the game where the $z=z'$ condition is present but not $x=x',y=y'$, is $(1-\frac{\nu'}{2})$. Further, since the probability of $x=x'$ and $y=y'$ is $\frac{1}{9}$, probability of winning $\MSe$ including the $x=x', y=y'$ condition is then $\frac{1}{9}(1-\frac{\nu'}{2})$.

Now Corollary \ref{cor:rand-V} has the following consequence for the parallel-repeated $\MSe$ game in the interactive leakage model.
\begin{cor}\label{cor:MSE-dpt}
There exists a constant $\beta>0$ such that if the total communication in the interactive leakage model is at most $cn$ for some $c < 1$, with $\nu$ being the constant from Fact \ref{fc:MSE-w}, then the probability of winning $\MSe$ in a random subset of size $t$ out of $n$ instances is at most
\[ \left(\frac{1-\nu + \beta(\sqrt{c} + \sqrt{t/n})}{9}\right)^t.\]
\end{cor}

\subsection{Parallel DIQKD with leakage}
We recall Theorem \ref{thm:leaky-qkd}.
\qkd*

First we formally define security for a QKD protocol, following \cite{AFRV19}. For simplicity, here we are using the same parameter for correctness, completeness and soundness, but in principle different parameters could be used.
\begin{definition}\label{def:qkd}
An interactive protocol between two honest parties Alice and Bob that results in either the parties aborting, or in Alice and Bob outputting $l$-bit secret keys $\tK^\A$ and $\tK^\B$ respectively is said to be $\lambda$-correct and $\lambda$-secure iff
\begin{enumerate}
\item (Correctness) $\Pr[\tK^\A\neq \tK^\B] \leq \lambda$;
\item (Soundness) If $\clE$ is the event that the protocol does not abort, and $\rho_{\tK^\A E}$ is the quantum state comprising of Alice's final key and the side information of the adversary Eve (which includes all classical communication between Alice and Bob) conditioned on not aborting, then
\[ \Pr[\clE]\norm{\rho_{\tK^\A E} - \frac{\Id}{2^l}\otimes\rho_{E}}_1 \leq \lambda;\]
\item (Completeness) The honest implementation of the protocol aborts with probability at most $\lambda$.
\end{enumerate}
\end{definition}

Protocol \ref{prot:QKD} realizing Theorem \ref{thm:leaky-qkd} is given below. It is a parallel protocol and makes use of the following equipment:
\begin{enumerate}[(i)]
\item Boxes $(\clB^\A,\clB^\B)$ with Alice and Bob respectively, whose honest behaviour is to play $n$ i.i.d. instances of $\MS$ $\delta$-noisily, i.e., each copy of $\MS$ is won with probability $1-\delta$;
\item Private sources of randomness for both Alice and Bob;
\item A public authenticated channel between Alice and Bob.
\end{enumerate}
The $\delta$ in the description of the protocol is the $\delta$ from item (i) above, i.e., the noise level in the honest behaviour of the boxes. We shall specify how the parameters $\alpha, \gamma$ in the protocol description are picked later. The description given here is not a full QKD protocol; after the steps described, error correction and privacy amplification need to be performed on the raw keys, giving the final keys. However, these are very standard, and we refer the interested reader to e.g. \cite{PR14, AFRV19} for details. The error correction step essentially ensures that the correctness condition (condition 1) in Definition \ref{def:qkd} is satisfied, as long as $n=\Omega\left(\frac{1}{\delta^2\alpha\gamma}\log(1/\lambda)\right)$.

\begin{algorithm}[!h]
\caption{Parallel DIQKD protocol (with parameters $\alpha, \gamma, \delta$)}
\label{prot:QKD}
\begin{algorithmic}[1]
\State Alice chooses $x_1\ldots x_n \in \{0,1,2\}^n$ uniformly at random from private randomness, inputs it into her box $\clB^\A$, and records the output $a_1\ldots a_n$ \;
\State Bob chooses $y_1\ldots y_n \in \{0,1,2\}^n$ uniformly at random from private randomness, inputs it into his box $\clB^\B$, and records the output $b_1\ldots b_n$ \;
\State Alice chooses $S \subseteq [n]$ of size $\alpha n$, $T \subseteq S$ of size $\gamma |S|$  uniformly at random from private randomness \; \label{step:bef-comm}
\State Alice sends $(S,T,x_S,a_T)$ to Bob using the public channel \;
\State Bob sends $y_S$ to Alice using the public channel \;
\State Bob tests if $(a_i[0]\oplus a_i[1]\oplus a_i[2]=0)\land(b_i[0]\oplus b_i[1]\oplus b_i[2]=1)\land(a_i[y_i] = b_i[x_i])$ for at least $(1-2\delta)|T|$ many $i$-s in $T$ \;
\If{the test fails}
\State Bob aborts the protocol \;
\Else
\State Alice sets $(K^\A)_{i\in S} = a_i[y_i]$ and Bob sets $(K^\B)_{i\in S} = b_i[x_i]$ as their respective raw keys
\EndIf
\end{algorithmic}
\end{algorithm}

We need to prove Protocol \ref{prot:QKD} satisfies conditions 2 and 3 of Definition \ref{def:qkd} in the presence of leakage. Note that the honest implementation of the devices does not involve leakage, so the proof of condition 3 is the same as in the no-leakage case, which was shown in \cite{JMS17, Vid17}. In order to show condition 2, it is enough to lower bound the (smoothed) conditional Renyi-2 entropy of $K^\A$ conditioned on everything held by Eve, which includes her quantum side information, as well as copies of whatever Alice and Bob communicate through the classical channel. In fact, by the Leftover Hashing Lemma, if condition 2 of Definition \ref{def:qkd} is to be satisfied with parameter $\lambda$, then the number of bits of secret key extracted from the protocol is given by this smoothed Renyi-2 entropy, with a factor of $(\log(1/\lambda) - \log(1/\Pr[\clE]) + O(1))$ subtracted \cite{Ren-th, PR14, TL17}.\footnote{The factor of $\log(1/\Pr[\clE])$, where $\clE$ is the event that the protocol does not abort, is thus added to the entropy bound. This is because condition 2 of Definition \ref{def:qkd} can tolerate a multiplicative factor of $\Pr[\clE]$ in front of the trace distance. For details, we refer the reader to \cite{PR14, AFRV19}.}

Most security proofs for QKD work with the conditional min-entropy (which is smaller) instead of the conditional Renyi-2 entropy. However, the Leftover Hashing Lemma used to prove condition 2 for QKD works fine with the Renyi-2 entropy as well; see e.g. Theorem 5.5.1 in \cite{Ren-th}. In our case, we shall indeed lower bound the conditional min-entropy for Protocol \ref{prot:QKD}, but later for our sequential protocols, we shall use Renyi-2 entropy instead. However, since we allow for smoothing, whether we use $\sfH_2$ or $\sfH_\infty$ does not make a lot of difference, due to Fact \ref{fc:H2-Hmin}.

To prove the security of Protocol \ref{prot:QKD}, we shall prove the following lemma.
\begin{lemma}\label{thm:qkd-Hmin}
Let $\rho_{K^\A K^\B X_SY_SA_TST\tE}$ be the state of Alice's and Bob's raw keys and Eve's side information conditioned on not aborting in Protocol \ref{prot:QKD} (where $\tE$ is Eve's quantum register and $X_SY_SA_TST$ is the communication through the public channel, which Eve also has access to). If the total communication in the interactive leakage model is $cn$ for some $c < 1$ and $\Pr[\clE] \geq 2\cdot 2^{-8\delta^2 \alpha \gamma n}$, then the state $\rho$ satisfies
\[ \sfH_\infty^\eps(K^\A|X_SY_SA_TST\tE)_\rho - \sfH^\eps_0(K^\A|K^\B)_\rho \geq  \alpha\left(\nu - \beta(\sqrt{c}+ \sqrt{\alpha}) - 4\delta - h(4\delta) - \gamma\right)n -1 - \log(1/\Pr[\clE]),\]
where $\clE$ is the event that the protocol does not abort, $\eps = \frac{2\cdot 2^{-8\delta^2\alpha \gamma n}}{\Pr[\clE]}$, $\beta, \nu$ are constants in $(0,1)$ (given by Fact \ref{fc:MSE-w} and Corollary \ref{cor:MSE-dpt}), and $h$ is the binary entropy function. Moreover, when $(\clB^\A,\clB^\B)$ have their honest $\delta$-noisy behaviour, then $\Pr[\clE] \geq 1 - 2^{-2\delta^2\gamma \alpha n}$.
\end{lemma}

Before proving the lemma, we shall briefly explain the significance of the quantity we lower bound in it. Firstly, if $\Pr[\clE] \leq 2\cdot 2^{-8\delta^2\alpha \gamma n}$, condition 2 of Definition \ref{def:qkd} is saitsfied automatically, for $n = \Omega\left(\frac{1}{\delta^2\alpha\gamma}\log(1/\lambda)\right)$. Thus we only need to lower bound the Renyi-2 entropy in the case that $\Pr[\clE]$ is larger than this bound. Alice and Bob need to do the error correction step on the raw keys $K^\A$ and $K^\B$ that they produce at the end of the protocol as described. The (one-way) error correction procedure will go like this: Alice will send a message $M$ of length $\sfH_0^\eps(K^\A|K^\B)_\rho + \log(1/\lambda) + O(1)$ to Bob, and based on thismessage and Bob's raw key $K^\B$, Bob will produce a guess $\tK^\A$, which will be equal to $K^\A$ with probability at least $1-\lambda$, as long as $n=\Omega\left(\frac{1}{\delta^2\alpha\gamma}\log(1/\lambda)\right)$ \cite{Ren-th}. Throughout the whole procedure, all the information communicated via the classical channel is $STX_SA_T$, $Y_S$ and $M$.
Therefore, using Fact \ref{fc:Hinf-ch-rule}, the Renyi-2 entropy of $K^\A$ conditioned on everything held by Eve after error correction is given by,
\begin{align*}
\sfH^\eps_2(K^\A|X_SY_SA_TSTM\tE)_\rho & \geq  \sfH^\eps_2(K^\A|X_SY_SA_TST\tE)_\rho - \log|M| \\
 & \geq \sfH^\eps_\infty(K^\A|X_SY_SA_TST\tE)_\rho - \sfH_0(K^\A|K^\B)_\rho - \log(1/\lambda) - O(1).
\end{align*}
As already stated, the number of bits of key that we can extract from this protocol after error correction, with $\lambda$ security, is $\sfH^\eps_2(K^\A|X_SY_SA_TSTM\tE)_\rho - \log(1/\lambda) + \log(1/\Pr[\clE]) - O(1)$. Since Lemma \ref{thm:qkd-Hmin} lower bounds $\sfH^\eps_\infty(K^\A|X_SY_SA_TST\tE)_\rho - \sfH_0(K^\A|K^\B)_\rho$, it thus gives us the quantity $r^{\mathrm{QKD}}_{\mathrm{par}}$ defined in Theorem \ref{thm:leaky-qkd}:
\begin{equation}\label{eq:r_par}
r^{\mathrm{QKD}}_{\mathrm{par}}(\delta, c) = \alpha\left(\nu - \beta(\sqrt{c}+ \sqrt{\alpha}) -4\delta -h(4\delta) - \gamma\right).
\end{equation}

The parameters $\alpha$ and $\gamma$ in Protocol \ref{prot:QKD} and thus \eqref{eq:r_par} are as yet unspecified; we describe how to pick them now. For $c, \delta$ such that $\nu > \beta\sqrt{c} + 2h(4\delta)$, there exist choices of $\alpha, \gamma$ and values of $\Pr[\clE]$ for which the key rate given by Lemma \ref{thm:qkd-Hmin} is positive. We pick such $\alpha, \gamma$, and then the key rate achieved by the protocol with $cn$ leakage and $\delta$ noise is positive; $\alpha = O(\nu^2)$ is a valid choice.

In order to prove Lemma \ref{thm:qkd-Hmin}, we introduce some notation for states. Note that we have defined $\clE$ to be the abort event, but we can equivalently define it to be the event that $(a_i[0]\oplus a_i[1]\oplus a_i[2]=0)\land(b_i[0]\oplus b_i[1]\oplus b_i[2]=1)\land(a_i[y_i] = b_i[x_i])$ for at least $(1-2\delta)|T|$ many $i$-s in $T$. This way we can condition states of the protocol before Alice and Bob have communicated on $\clE$ as well, even though they cannot abort at this point. We use:
\nopagebreak

\begin{tabular}{lcp{10cm}}
$\rho_{K^\A K^\B X_SY_SA_TST\tE}$ & : & state conditioned on $\clE$ at the end of Protocol \ref{prot:QKD} \\
$\sigma_{K^\A K^\B X_SY_SA_TST\tE}$ & : & state after step \ref{step:bef-comm} in Protocol \ref{prot:QKD} \\
$\vph_{K^\A K^\B X_SY_SA_TST\tE}$ & : & state after step \ref{step:bef-comm} in Protocol \ref{prot:QKD} conditioned on $\clE$.
\end{tabular}

First we shall prove some lemmas about the states $\sigma$ and $\vph$, and then use them to get the final min-entropy bound on $\rho$.

\begin{lemma}\label{lem:virt-Hmin-2}
If the total communication in the interactive leakage model is at most $cn$ for some $c < 1$, and $\Pr[\clE] \geq 2\cdot 2^{-8\delta^2\alpha \gamma n}$, then
\[ \sfH^\eps_\infty(K^\A|X_SY_SS\tE)_\vph \geq \alpha\left(\nu - \beta(\sqrt{c}+\sqrt{\alpha}) - 4\delta\right)n - 1 - \log(1/\Pr[\clE])\]
for $\eps = 2\cdot 2^{-8\delta^2\alpha \gamma n}/\Pr[\clE]$.
\end{lemma}
\begin{proof}
Consider the $\MSe_{\text{rand}}^{\alpha n/n}$ game being played on the state shared by Alice, Bob and Eve (with $S$ being the random subset of size $\alpha n$, and $\MSe_{\text{rand}}^{\alpha n/n}$ being won if the instances in the random subset $S$ are won) in Protocol \ref{prot:QKD}. The variables $X_i, Y_i, A_i, B_i$ in the protocol are straightforwardly Alice and Bob's inputs and outputs in the game, and $C_i$ is Eve's guess for Alice's raw key bit. The variables $X'_i, Y'_i, Z_i, Z'_i$ do not actually exist in the protocol. As far as $X'_i, Y'_i$ are concerned, they are uncorrelated with $X_i, Y_i$, which are Alice and Bob's inputs, so we can treat them as inputs to Eve, since Eve can always have private randomness that is uncorrelated with Alice and Bob. Conditioning on $X'_i=X_i$ and $Y'_i=Y_i$ then corresponds to conditioning on Eve actually knowing Alice and Bob's inputs at those locations. The bits $Z_i$ and $Z'_i$ are only relevant at locations $i$ where $A_i[Y_i] \neq B_i[X_i]$. We use these variables in order to be able to apply Corollary \ref{cor:MSE-dpt} in our security proof. In order to apply the corollary, it is important that $\MSe$ be won on all the coordinates in $S$, even if $A_i[Y_i]\neq B_i[Z_i]$ at those locations; $Z_i, Z'_i$ are introduced so that it is still possible for the game to be won at these locations with $Z'_i=Z_i$.

Let $U_i$ be the indicator variable of the event that $X'_iY'_iC_i = X_iY_iA_i[Y_i]$, $V_i$ be the indicator variable for the event that $Z'_i = Z_i$ and $W_i$ be the indicator variable for the event that $(A_i[0]\oplus A_i[1]\oplus A_i[2]=0)\land(B_i[0]\oplus B_i[1]\oplus B_i[2]=1)\land(A_i[Y_i]=B_i[X_i])$ for $i\in S$. From Fact \ref{fc:guess-prob},
\begin{align*}
\Pr_\sigma\left[\prod_{i \in S}U_i \land (\lnot W_i \implies V_i)\right] = \Pr\left[\text{Win } \MSe^{\alpha n/n}_{\text{rand}}\right] \leq \left(\frac{1-\nu+\beta(\sqrt{c}+\sqrt{\alpha})}{9}\right)^{\alpha n}
\end{align*}
where we have used Corollary \ref{cor:MSE-dpt} along with the upper bound on communication in the last line.

Since $X_iY_i$ and $X'_iY'_i$ are uniformly random on a set of support size 9, we have that $X'_S=X_S$ and $Y'_S=Y_S$ with probability $(\frac{1}{9})^{|S|}$. Since this event is included in $\prod_{i\in S}U_i$, we have that,
\begin{align*}
\Pr_\sigma\left[\prod_{i \in S}U_i \land (\lnot W_i \implies V_i)\middle|X'_S=X_S\land Y'_S=Y_S\right] & \leq \left(\frac{1-\nu+\beta(\sqrt{c}+\sqrt{\alpha})}{9}\right)^{\alpha n}\cdot 9^{\alpha n} \\
& = \left(1-\nu+\beta(\sqrt{c}+\sqrt{\alpha})\right)^{\alpha n}.
\end{align*}

Let $\vph'$ denote $\sigma$ conditioned on the following event, which we call $\clE'$:
\[ \left(\sum_{i \in T}W_i \geq (1-2\delta)|T|\right) \land\left(\sum_{i\in S}W_i \geq (1-4\delta)|S|\right).\]
That is, $\clE'$ is a conjunction of $\clE$ and another event, and $\vph'$ is $\vph$ conditioned on this further event (since $\vph$ is $\sigma$ conditioned on $\clE$). We have that,
\[ \Pr_{\vph'}\left[\prod_{i \in S}U_i \land (\lnot W_i \implies V_i)\middle|X'_S=X_S\land Y'_S=Y_S\right] \leq \frac{\left(1-\nu+\beta(\sqrt{c}+\sqrt{\alpha})\right)^{\alpha n}}{\Pr[\clE']}. \]
By Fact \ref{fc:serfling}, $\Pr[\clE'] \geq \Pr[\clE] - 2^{-8\delta^2 \alpha\gamma n}$, which gives us $\Vert\vph - \vph'\Vert_1 \leq \frac{2\cdot 2^{-8\delta^2\alpha\gamma n}}{\Pr[\clE]} = \eps$. In $\vph'$, $A_i[Y_i]$ and $B_i[X_i]$ differ in at most $4\delta|S|$ many places in $S$; let us call this set $S'$. We note that $\Pr[\prod_{i\in S\setminus S'}W_i]=1$ in $\vph'$. Therefore,
\begin{align*}
\Pr_{\vph'}\left[\prod_{i\in S}U_i\land\prod_{i\in S'}V_i\middle|X'_S=X_S\land Y'_S=Y_S\right] & = \Pr_{\vph'}\left[\prod_{i\in S}U_i\land\prod_{i\in S'}V_i\land \prod_{i\in S\setminus S'}W_i\middle|X'_S=X_S\land Y'_S=Y_S\right] \\
 & = \Pr_{\vph'}\left[\prod_{i \in S}U_i \land (\lnot W_i \implies V_i)\middle|X'_S=X_S\land Y'_S=Y_S\right].
\end{align*}
Since $Z_i, Z'_i$ are uniformly random bits that are independent of every other variable, we have that $\Pr[\prod_{i\in S'}V_i]=\left(\frac{1}{2}\right)^{|S'|} \geq \left(\frac{1}{2}\right)^{4\delta \alpha n}$. This gives us
\begin{align*}
\Pr_{\vph'}\left[\prod_{i\in S}U_i\middle|X'_S=X_S\land Y'_S=Y_S\right] & = \Pr_{\vph'}\left[\prod_{i\in S}U_i\middle|X'_S=X_S\land Y'_S=Y_S\land\prod_{i\in S'}V_i\right] \\
& \leq \frac{\left(1-\nu+\beta(\sqrt{c}+\sqrt{\alpha})\right)^{\alpha n}}{\Pr[\clE']}\cdot 2^{4\delta \alpha n}.
\end{align*}
The above probability is Eve's guessing probability for $A_i[Y_i]$ in $\vph'$ conditioned on knowing $X_S, Y_S$ (and also having access to $S$, which Eve always knows, and Eve's quantum side information $\tE$). Thus by Fact \ref{fc:guess-prob} we have,
\begin{align*}
\sfH_\infty(K^\A |X_SY_SS\tE)_{\vph'} & \geq \alpha n\cdot \log\left(\frac{1}{1 - \nu + \beta(\sqrt{c}+\sqrt{\alpha})}\right) - \log(1/\Pr[\clE']) - 4\delta \alpha n \\
 & \geq \alpha\left(\nu - \beta(\sqrt{c}+\sqrt{\alpha})\right) - \log(1/\Pr[\clE']) - 4\delta \alpha n
\end{align*}
Using the lower bound on $\Pr[\clE]$ from the lemma statement, and the lower bound for $\Pr[\clE']$ in terms of $\Pr[\clE]$, we have that $\Pr[\clE'] \geq \Pr[\clE]/2$. Finally, since $\vph$ and $\vph'$ are $\eps$ apart, we have,
\[ \sfH^\eps_\infty(K^\A |X_SY_SS\tE)_{\vph} \geq \alpha\left(\nu - \beta(\sqrt{c}+\sqrt{\alpha}) - 4\delta\right)n - 1 - \log(1/\Pr[\clE]). \qedhere \]
\end{proof}

\begin{proof}[Proof of Lemma \ref{thm:qkd-Hmin}]
First we shall condition the conditional min-entropy bound from Lemma \ref{lem:virt-Hmin-2} further on $(T,A_T)$. Among these, $T$ is independent of $K^\A$, so conditioning on it makes no difference. $A_T$ is contained in $K^\A$, and uniformly random in $\{0,1,2\}^{|T|}$, with $T$ being of size $\alpha\gamma n$. Hence,
\[ \sfH^\eps_\infty(K^\A|X_SY_SA_TST\tE)_\vph \geq \alpha\left(\nu - \beta(\sqrt{c}+\sqrt{\alpha}) - 4\delta\right)n - 1- \log(1/\Pr[\clE]) - \alpha\gamma n.\]
Now notice that in $\rho$, $X_SY_SSTA_T$ is revealed to Eve, so she may do some operations on her side depending on these. $\rho$ is thus related to $\vph$ by some local operations on the registers $X_SY_SSTA_T\tE$. Hence by Fact \ref{fc:local-Hmin},
\[ \sfH^\eps_\infty(K^\A|X_SY_SA_TST\tE)_\rho \geq \alpha\left(\nu - \beta(\sqrt{c}+\sqrt{\alpha}) - 4\delta - \gamma\right)n - 1- \log(1/\Pr[\clE]).\]

Finally, to bound $\sfH^\eps_0(K^\A|K^\B)_\rho$, we consider the state $\rho'$, which is conditioned on the event $\clE'$ as defined in the proof of Lemma \ref{lem:virt-Hmin-2} instead of $\clE$ like $\rho$. They satisfy $\Vert\rho - \rho'\Vert_1 \leq \frac{2\cdot 2^{-8\delta^2\alpha\gamma n}}{\Pr[\clE]}$. The number of strings $K^\B$ of length that can differ from a given value of $K^\A$ in at most $4\delta|S|$ places is at most $2^{h(4\delta)|S|}$, which gives us $\sfH_0(K^\B|K^\A)_{\rho'} \leq h(4\delta)\alpha n$. Putting everything together we get,
\begin{align*}
\sfH^\eps_\infty(K^\A|X_SY_SA_TST\tE)_\rho - \sfH^\eps_0(K^\B|K^\A)_\rho & \geq \alpha\left(\nu - \beta(\sqrt{c}+\sqrt{\alpha}) - 4\delta - \gamma - h(4\delta)\right)n - 1- \log(1/\Pr[\clE]).
\end{align*}

To lower bound $\Pr[\clE]$ in the honest case when each instance of $\MS$ is won with probability $1-\delta$, we use the Chernoff bound. Letting $W_i$ denote the indicator variable for $A_i[Y_i] = B_i[X_i]$, the $W_i$-s are i.i.d. in this case, and the expected value of each $W_i$ is $1-\delta$. Hence
\[ \Pr[\lnot\clE] = \Pr\left[\sum_{i\in T} W_i < (1-2\eps)|T|\right] \leq 2^{-2\delta^2\gamma \alpha n}. \qedhere\]
\end{proof}

\subsection{Sequential DIRE and DIQKD with leakage}
We recall Theorem \ref{thm:leaky-re}.
\re*

First, we formally define  $\lambda$-secure $m\to l$ randomness expansion.
\begin{definition}\label{def:RE}
A protocol that takes as input a uniformly random $m$-bit string $R$ is called a $m\to l$ randomness expansion protocol if the protocol either aborts or returns a random $l$-bit string $Z$. The protocol is $\lambda$-secure iff
\begin{enumerate}
\item  (Soundness) If $\clE$ is the event that the protocol does not abort, and $\rho_{ZR\tE}$ is the quantum state comprising of the input and output of the protocol and the adversary Eve's side information conditioned on not aborting, then
\[ \Pr[\clE]\norm{\rho_{ZR\tE} - \frac{\Id}{2^l}\otimes\rho_{R\tE}}_1 \leq \lambda;\]
\item (Completeness) The honest implementation of the protocol aborts with probability at most $\lambda$.
\end{enumerate}
\end{definition}

Protocols \ref{prot:seq-QKD} and \ref{prot:RE} referred to in Theorem \ref{thm:leaky-re} are given below. They are sequential protocols, which make use of the following equipment:
\begin{itemize}[(i)]
\item Two sets of boxes $\clB^\A_1\ldots \clB^\A_n$ and $\clB^\B_1\ldots \clB^\B_n$, which are capable of taking inputs and providing outputs sequentially. The honest behaviour of boxes $\clB^\A_i$ and $\clB^\B_i$ is to play a copy of MS $\delta$-noisly, and independent of the other boxes;
\item A uniform source of input randomness.
\end{itemize}
Additionally, Protocol \ref{prot:seq-QKD} makes use of an authenticated classical channel between Alice and Bob just like Protocol \ref{prot:QKD}.

We shall describe Protocol \ref{prot:RE} as Alice doing things with boxes $\clB^\A_1\ldots \clB^\A_n$, and Bob doing things with $\clB^\B_1\ldots \clB^\B_n$, although in the actual DIRE setup, it is only one party playing the roles of both Alice and Bob.
In both protocol descriptions, it should be understood that the $(i+1)$-th round's input to Bob's box is entered after the $i$-th round's outputs are obtained from Alice's box, and vice versa. As before, the protocols as described here are not complete --- error correction and privacy amplification need to be performed in the QKD protocol, and a step analogous to privacy amplification (using randomness extractors) needs to be performed in the RE protocol. 
\begin{algorithm}[!h]
\caption{Sequential DIQKD protocol (with parameters $\gamma, \delta$)}
\label{prot:seq-QKD}
\begin{algorithmic}[1]
\For{$i=1$ to $n$ sequentially}
\State Alice chooses $x_i \in \{0,1,2\}$ uniformly at random from private randomness, inputs it into her box $\clB^\A_i$, and records the output $a_i$ \;
\State Bob chooses $y_i \in \{0,1,2\}$ uniformly at random from private randomness, inputs it into his box $\clB^\B_i$, and records the output $b_i$ \;
\EndFor
\State Alice chooses $T \subseteq [n]$ of size $\gamma n$  uniformly at random from private randomness \;
\State Alice sends $(T,x,a_T)$ to Bob using the public channel \;
\State Bob sends $y$ to Alice using the public channel \;
\State Bob tests if $(a_i[0]\oplus a_i[1]\oplus a_i[2]=0)\land(b_i[0]\oplus b_i[1]\oplus b_i[2]=1)\land(a_i[y_i] = b_i[x_i]=1)$ for at least $(1-2\delta)|T|$ many $i$-s in $T$ \; \label{step:test-re}
\If{the test fails}
\State Bob aborts the protocol \;
\Else
\State Alice sets $K^\A = (a_i[y_i])_{i=1}^n$ and Bob sets $K^\B = (b_i[y_i])_{i=1}^n$ as their raw outputs.
\EndIf
\end{algorithmic}
\end{algorithm}

\begin{algorithm}[!h]
\caption{Sequential DIRE protocol (with parameter $\delta$)}
\label{prot:RE}
\begin{algorithmic}[1]
\For{$i=1$ to $n$ sequentially}
\State Alice chooses $x_i \in \{0,1,2\}$ uniformly at random from seed randomness, inputs it into her box $\clB^\A_i$, and records the output $a_i$ \;
\State Bob chooses $y_i \in \{0,1,2\}$ uniformly at random from seed randomness, inputs it into his box $\clB^\B_i$, and records the output $b_i$ \;
\EndFor
\State Alice and Bob test if $(a_i[0]\oplus a_i[1]\oplus a_i[2]=0)\land(b_i[0]\oplus b_i[1]\oplus b_i[2]=1)\land(a_i[y_i] = b_i[x_i]=1)$ for at least $(1-2\delta)n$ many $i$-s in $[n]$ \; \label{step:test-re}
\If{the test fails}
\State The protocol is aborted \;
\Else
\State $K = (a_i[y_i])_{i=1}^n$ is set as the raw output.
\EndIf
\end{algorithmic}
\end{algorithm}
Note that for randomness expansion, Alice and Bob's winning condition for Magic Square is tested on the whole of $[n]$ instead of a subset $T$. This is because the testing data does not need to be communicated over the public channel in this case, and thus does not need to be subtracted from the key rate. Using randomness extractors, in order to show condition 1 of Definition \ref{def:RE}, it is enough to bound the min-entropy of $K$ given the randomness used in the protocol, which is $XY$ (these can be generated using $n\log 9$ uniformly random bits), as well as Eve's side information $\tE$. The length of the final uniformly random string extracted is given by this min-entropy (again excluding the $\log(1/\Pr[\clE])$ factor).

Note that it is in principle possible to use the Leftover Hashing Lemma for randomness expansion as well, in which case it would have been fine to lower bound the Renyi-2 entropy rather than the min-entropy. However, the Leftover Hashing Lemma uses a lot of extra random bits. This is not a problem in QKD, where private randomness is a free resource, but in randomness expansion we want to minimize the amount of initial randomness as much as possible. The most efficient quantum-proof randomness extractors which use $\polylog(n)$-length random seeds require a lower bound on the min-entropy rather than the Renyi-2 entropy \cite{DPVR12}, and so we shall be working with that. The overall randomness used in the RE protocol are the $n\log 9$ bits for $XY$, and the $\polylog(n)$ bits of seed randomness for the extractor.

We prove the following lemmas about Protocols \ref{prot:seq-QKD} and \ref{prot:RE}. Lemma \ref{lem:seq-QKD-H2} straightforwardly gives $r^{\mathrm{QKD}}_{\mathrm{seq}}(\delta, c) = \log(1/(1-\nu)) -4\delta -h(4\delta) - \gamma - c$ in Theorem \ref{thm:leaky-re}. Picking  $\eps=O(\lambda)$ (where $\lambda$ is the security parameter desired for Definition \ref{def:RE})  in Lemma \ref{lem:re-Hmin} gives $r^{\mathrm{RE}}_{\mathrm{seq}}(\delta, c) = \log(1/(1-\nu)) - 2\delta - c$  in Theorem \ref{thm:leaky-re}. This extra factor of $\log(1/\eps) = \log(1/\lambda) + O(1)$ is what leads the final key rate for RE in Theorem \ref{thm:leaky-re} to have an additive $2\log(1/\lambda)$ instead of $\log(1/\lambda)$.
\begin{lemma}\label{lem:seq-QKD-H2}
Let $\rho_{K^\A K^\B XYA_TTE}$ be the state of Alice's and Bob's raw keys and Eve's side information conditioned on not aborting in Protocol \ref{prot:seq-QKD}. If the total communication in the sequential interactive model is $cn$ for some $c<1$, and $\Pr[\clE] \geq 2\cdot 2^{8\delta^2\gamma n}$, then the state $\rho$ satisfies
\[ \sfH_2^\eps(K^\A|XYT\tE)_\rho - \sfH_0^\eps(K^\A|K^\B)_\rho \geq \log\left(\frac{1}{1-\nu}\right) n - (4\delta + h(4\delta)+c)n - 1 - \log(1/\Pr[\clE]), \]
where $\clE$ is the event that the protocol does not abort, and $\eps = \frac{2\cdot 2^{-8\delta^2\alpha\gamma n}}{\Pr[\clE]}$. Moreover, when $\clB^\A_1\ldots \clB^\A_n$ and $\clB^\B_1\ldots \clB^\B_n$ have their $\delta$-noisy behaviour, then $\Pr[\clE] \geq 1 - 2^{-2\delta^2 \gamma n}$.
\end{lemma}
\begin{lemma}\label{lem:re-Hmin}
Let $\rho_{KXY\tE}$ be the state of the raw output, seed randomness and Eve's side information conditioned on not aborting in Protocol \ref{prot:RE}. If the total communication in the sequential interactive leakage model is $cn$ for some $c < 1$, then for any $\eps \in (0,1)$, the state $\rho$ satisfies
\[ \sfH^\eps_\infty(K|XY\tE)_\rho \geq \log\left(\frac{1}{1-\nu}\right) n - (2\delta +c)n - \log(1/\Pr[\clE]) -\log(2/\eps),\]
where $\clE$ is the event that the protocol does not abort. Moreover, when $\clB^\A_1\ldots \clB^\A_n$ and $\clB^\B_1\ldots \clB^\B_n$ have their $\delta$-noisy behaviour, then $\Pr[\clE] \geq 1 - 2^{-2\delta^2 n}$.
\end{lemma}
We shall only give the proof of Lemma \ref{lem:re-Hmin} because the only parts of Protocol \ref{prot:seq-QKD} that are different from Protocol \ref{prot:RE} are similar to Protocol \ref{prot:QKD} instead. However, the following point is worth highlighting: in Lemma \ref{lem:seq-QKD-H2}, the smoothing parameter $\eps$ in $\sfH_2^\eps$ is due to the fact that we test the winning condition on a subset and then generalize to the whole of $[n]$: this comes from the Serfling bound and we cannot freely pick this parameter. In Lemma \ref{lem:re-Hmin}, we first lower bound the unsmoothed $\sfH_2$ --- smoothing is not required because we test on the whole of $[n]$ and do not need to apply the Serfling bound --- and then we convert from $\sfH_2$ to $\sfH^\eps_\infty$. We can freely pick the $\eps$ we use for this smoothing, and we shall pick it to be $O(\lambda)$ in order to get $\lambda$ soundness in Definition \ref{def:RE}.

\begin{proof}[Proof of Lemma \ref{lem:re-Hmin}]
To prove this lemma, we are going to lower bound $\sfH_2(K|XY\tE)_\rho$ by upper bounding the winning probability of $n$ copies of $\MSe$ being won when the games are played sequentially, with $cn$ sequential leakage. First, note that we can assume Eve holds a purification of Alice and Bob's quantum state, and since $\sfH_2$ is invariant under local isometries by Fact \ref{fc:local-Hmin}, for the purposes of computing $\sfH_2$, we can assume Eve holds the canonical purification. If we let $\sigma_{XYK\tE}$ denote the state at the end of the protocol without conditioning on $\clE$, then the $K$ register of this state is obtained by doing a measurement on Alice and Bob's part of the initial shared state, and then recording the outcome in $K$ (there were actually separate measurements on Alice and Bob's parts of the state, and $K$ is actually only the outcome of Alice's measurement --- but we can think of this as the coarse-grained outcome of a joint measurement on Alice and Bob's parts). Moreover, we can also think of the $XY$ registers as being obtained by Alice and Bob by doing a measurement in the computational basis on a state of the form $\sum_{xy}\sqrt{\sfP_{XY}(xy)}\ket{xy}$; this state is already pure, and its canonical purification is just another copy of it, which we shall assume is included in $\tE$. By Fact \ref{fc:H2-prob}, $\sfH_2(KXY|\tE)_\sigma$ is then the log of the inverse of Eve's probability of guessing $XYK$ by doing the same measurements that Alice and Bob did to obtain them, which means sequential measurements.

Now we would like to interpret the probability of winning $n$ sequential copies of $\MSe$ (witn $cn$ leakage) as a guessing probability for Eve, doing sequential measurements. The event $C_i=A_i[Y_i]$ in the game winning condition is obviously the event that Eve guesses $K_i$ --- this is not necessarily by performing the same measurement as Alice, but since we are interested in lower bounding $\sfH_2$, considering the event that Eve learns $K$ by doing arbitrary sequential measurements is fine. The event $X_i=X_i', Y_i=Y_i'$ (for each $i$) in the winning condition of $\MSe$ is actually the same event as Eve doing the same measurements as Alice and Bob on her copy of the state $\sum_{xy}\sqrt{\sfP_{XY}(xy)}\ket{xy}$ in her purification and getting the same outcomes as Alice and Bob (where we identify $X'$ and $Y'$ with Eve's outcomes). We deal with the $(A_i[Y_i]=B_i[X_i])\lor(Z_i=Z'_i)$ condition in winning the game in a similar way as we did in Lemma \ref{lem:virt-Hmin-2}; Alice and Eve will obtain $Z$ and $Z'$ by measuring, and we shall count Eve's guessing probability at locations where $A_i[Y_i]\neq B_i[X_i]$. We can also go from the guessing probability in $\sigma$ to the guessing probability in $\rho$ with a factor of $\Pr[\clE]$ as before. Since $A_i[Y_i]\neq B_i[X_i]$ in at most $2\delta n$ many locations in $\rho$, we have as before,
\[ \sfH_2(KXY|\tE)_\rho \geq -\log\left(\frac{\Pr[\clE]\cdot 2^{2\delta n}}{\Pr_\sigma[\text{Win $n$ sequential copies of } \MSe]}\right).\]
What remains for us to show is that the sequential winning probability of $n$ copies of $\MSe$ with $cn$ bits of leakage is $\left(\frac{1-\nu}{9}\right)^n\cdot 2^{cn}$. If we can show this, then the above expression gives us the lower bound on $\sfH_2$ we need, after removing the $XY$ in the first argument of $\sfH_2$ by multiplying the probability with $9^n$ as before (we have $9^n$ here instead of $9^{\alpha n}$ since there is no subset $S$ of size $\alpha n$). Applying Fact \ref{fc:H2-Hmin}, we then get the lemma.

First we shall upper bound the winning probability of $n$ copies of $\MSe$ being won sequentially without any leakage. Note that when the games are played sequentially, the outputs (and inputs) of the 1st to $i$-th games are uncorrelated with the inputs of the $(i+1)$-th game, although the shared state during the $(i+1)$-th game depends on the inputs and outputs of the previous games. We can consider the shared state during the $(i+1)$-th game conditioned on any values of the inputs and outputs of the previous games. Regardless of the values we conditioned on, a fresh copy of $\MSe$ is played in the $(i+1)$-th round, and the winning probability of this is at most $\frac{1-\nu}{9}$.\footnote{We could not make this argument in the parallel case because the outputs of games 1 to $i$ are correlated with the inputs of the $(i+1)$-th game, so the shared state for the $(i+1)$-th game would depend on the inputs of the $(i+1)$-th game if we condition on some values of the inputs and outputs of the previous games.} Inducting from the 1st to the $n$-th game we thus have that the probability of winning $n$ copies of $\MSe$ sequentially without leakage of at most $\left(\frac{1-\nu}{9}\right)^n$.

Now to deal with leakage, we shall provide an argument similar to that of Lemma \ref{lem:eff-lb}. We shall take a sequential protocol $n$ copies of $\MSe$ with $cn$ leakage and convert it to a protocol without leakage, which has a winning probability that is smaller by a factor of $2^{-cn}$; the extra thing we need to make sure of here is that the protocol without leakage is also sequential. This is automatically true if we do a simulation similar to that in the proof of Lemma \ref{lem:eff-lb}, where the players share randomness and use it to guess the messages leaked between the inputs of the $i$-th game being entered and the outputs being produced by the devices, for each $i$. There is no option of outputting $\bot$ in this setting, so if at any point the players notice that the shared randomness is not consistent with their input and measurement outcomes, then they produce random outputs. Suppose we have done the simulation up to the $i$-th round, and we condition on particular values of the inputs and outputs for the rounds up to $i$, as well as the event that the shared randomness matches the leaked messages up to the $i$-th round. The input distribution for the $(i+1)$-th game as well as the block of shared randomness that is supposed to be used for the $(i+1)$-th game are unaffected by this conditioning, so effectively a fresh copy of $\MSe$ is being played as the $(i+1)$-th game. If $c^{i+1}$ bits are leaked in the $(i+1)$-th game and its winning probability with the leakage if $p_{i+1}$, then the winning probability without leakage is at least $p_{i+1}\cdot 2^{-c^{i+1}}$ (here we are not counting the probability of the game being won when the shared randomness do not match the leaked bits and the players output randomly). Thus, if the overall winning probability with leakage is $p$, then the winning probability without leakage is at least $p\cdot 2^{-cn}$. Since we know that the winning probability without leakage is at most $\left(\frac{1-\nu}{9}\right)^n$, this means that $p$ is at most $\left(\frac{1-\nu}{9}\right)^n\cdot 2^{cn}$.

We have $\Pr[\clE] \geq 1-2^{-2\delta^2n}$ by the Chernoff bound in the honest case here; there is no factor of $\gamma$ in the exponent since the testing is done on the whole set instead of $T$.
\end{proof}
We note that in the above proof, we upper bounded the guessing probability of Eve for a single $A_i[Y_i]$ while doing the same measurement as Alice, by her overall best guessing probability for $A_i[Y_i]$. If we do not do this, we could potentially get the improvement mentioned in Section \ref{sec:QKDRE-results}, by upper bounding the probability of winning $n$ sequential games under this constraint.  

\subsection{Proof of quantumness with two players and leakage}
We recall Theorem \ref{thm:leaky-quantum}.
\quantum*

We first provide a security definition for proof of quantumness with two provers who are allowed to communicate a bounded amount.
\begin{definition}
A proof of quantumness with two provers is an interactive protocol between a verifier and two provers; the amount of interaction between the verifier and the provers separately is not bounded, but the amount of interaction between the two provers is. At the end of the protocol, the verifier outputs either $\top$ (indicating acceptance) or $\bot$ (indicating rejection). We say the protocol has correctness and soundness parameter $\lambda$ against $C(\lambda) = O(\log(1/\lambda))$ leakage iff
\begin{enumerate}
\item (Correctness) There exists a $\polylog(1/\lambda)$-time strategy that two quantum provers who share entanglement but do not interact can implement so that the verifier outputs $\top$ with probability at least $1-\lambda$.
\item (Soundness) For any strategy that two classical provers who do not share entanglement but communicate at most $C(\lambda)$ bits can implement, the verifier outputs $\top$ with probability at most $\lambda$.
\end{enumerate}
\end{definition}
Protocol \ref{prot:quantum} that realizes Theorem \ref{thm:leaky-quantum} is described below. The protocol involves one round of communication from the verifier to each prover, and one round of communication from each prover to the verifier. The verifier can communicate to the two provers simultaneously or in any order, and the two provers can communicate back simultaneously or in any order, which is why we call this a 2-round protocol. The protocol requires the prover to have access to private randomness. The parameter $\delta$ used in the protocol is a noise parameter --- it allows quantum provers who implement a noisy version of the ideal strategy to be accepted by the verifier.
\begin{algorithm}[!h]
\caption{Proof of quantumness protocol with two provers (with parameter $\delta$)}
\label{prot:quantum}
\begin{algorithmic}[1]
\State The verifier chooses $x_1\ldots x_n \in \{0,1,2\}^n$ and $y_1\ldots y_n \in \{0,1,2,\}^n$ uniformly at random \;
\State The verifier sends $x_1\ldots x_n$ to Prover 1, and $y_1\ldots y_n$ to Prover 2 \;
\State Prover 1 and Prover 2 send $a_1\ldots a_n$ and $b_1\ldots b_n$ respectively to the verifier
\State The verifier outputs $\top$ if $(a_i[0]\oplus a_i[1]\oplus a_i[2]=0)\land(b_i[0]\oplus b_i[1]\oplus b_i[2]=1)\land(a_i[y_i] = b_i[x_i])$ for at least $(1-2\delta)n$ many $i$-s in $[n]$, and outputs $\bot$ otherwise.
\end{algorithmic}
\end{algorithm}
\begin{lemma}
For two quantum provers who share an entangled state that plays $n$ independent copies of the Magic Square game with at most $\delta$ noise, and provide $a_1\ldots a_n$ and $b_1\ldots b_n$ as the outputs of the Magic Square game on inputs $x_1\ldots x_n$ and $y_1\ldots y_n$ respectively, the verifier outputs $\top$ with probability at least $1-2^{-2\delta^2n}$. On the other hand, for two classical provers who do not share entanglement but communicate at most $cn$ bits, the verifier outputs $\top$ with probability at most $2^{-\Omega(1/9 - 2\delta)^3n + cn}$.
\end{lemma}
\begin{proof}
The correctness property for two quantum provers with an entangled state that plays $n$ copies of the Magic Square game $\delta$-noisily follows from the Chernoff bound. The soundness property for classical players follows from Corollary \ref{cor:MS-cl-leaky}.
\end{proof}


\section*{Acknowledgements}
We thank Ernest Tan and Tony Metger for helpful discussions on the security of DIQKD with leakage, in particular on the difficulties of applying known sequential proof techniques, and for pointing out references \cite{SPM13,TZB+20,TZWP20}.

This work was done in part while S.K. was at the Centre for Quantum Technologies (CQT), National University of Singapore. Research at CQT is supported by the National Research Foundation, including under NRF RF Award No. NRF-NRFF2013-13, the Prime Minister's Office, Singapore and the Ministry of Education, Singapore, under the Research Centres of Excellence program and by Grant No. MOE2012-T3-1-009 and in part by the NRF2017-NRF-ANR004 {\em VanQuTe} Grant. S. K. is currently funded by the NSERC Canada Discovery Grants Program and Fujitsu Labs America; research at the Institute for Quantum Computing (IQC) is supported by Innovation, Science and Economic Development (ISED) Canada.


\bibliographystyle{alpha}
\bibliography{two-way-DIQKD}

\newcommand{\etalchar}[1]{$^{#1}$}
\begin{thebibliography}{TZCBB{\etalchar{+}}20}

\bibitem[AA11]{AA10}
Scott Aaronson and Alex Arkhipov.
\newblock The computational complexity of linear optics.
\newblock In {\em 43rd Annual ACM Symposium on Theory of Computing}, STOC '11,
  page 333–342, 2011.

\bibitem[AAB{\etalchar{+}}19]{AAB+19}
Frank Arute, Kunal Arya, Ryan Babbush, Dave Bacon, Joseph~C. Bardin, Rami
  Barends, Rupak Biswas, Sergio Boixo, Fernando G. S.~L. Brandao, David~A.
  Buell, Brian Burkett, Yu~Chen, Zijun Chen, Ben Chiaro, Roberto Collins,
  William Courtney, Andrew Dunsworth, Edward Farhi, Brooks Foxen, Austin
  Fowler, Craig Gidney, Marissa Giustina, Rob Graff, Keith Guerin, Steve
  Habegger, Matthew~P. Harrigan, Michael~J. Hartmann, Alan Ho, Markus Hoffmann,
  Trent Huang, Travis~S. Humble, Sergei~V. Isakov, Evan Jeffrey, Zhang Jiang,
  Dvir Kafri, Kostyantyn Kechedzhi, Julian Kelly, Paul~V. Klimov, Sergey Knysh,
  Alexander Korotkov, Fedor Kostritsa, David Landhuis, Mike Lindmark, Erik
  Lucero, Dmitry Lyakh, Salvatore Mandr{\`{a}}, Jarrod~R. McClean, Matthew
  McEwen, Anthony Megrant, Xiao Mi, Kristel Michielsen, Masoud Mohseni, Josh
  Mutus, Ofer Naaman, Matthew Neeley, Charles Neill, Murphy~Yuezhen Niu, Eric
  Ostby, Andre Petukhov, John~C. Platt, Chris Quintana, Eleanor~G. Rieffel,
  Pedram Roushan, Nicholas~C. Rubin, Daniel Sank, Kevin~J. Satzinger, Vadim
  Smelyanskiy, Kevin~J. Sung, Matthew~D. Trevithick, Amit Vainsencher, Benjamin
  Villalonga, Theodore White, Z.~Jamie Yao, Ping Yeh, Adam Zalcman, Hartmut
  Neven, and John~M. Martinis.
\newblock Quantum supremacy using a programmable superconducting processor.
\newblock {\em Nature}, 574(7779):505--510, 2019.

\bibitem[ABJO21]{ABJO21}
Divesh Aggarwal, Naresh Boddu, Rahul Jain, and Maciej Obremski.
\newblock {Quantum Measurement Adversary}.
\newblock \url{https://arxiv.org/abs/2106.02766}, 2021.

\bibitem[ABJT20]{ABJT18}
Anurag Anshu, Mario Berta, Rahul Jain, and Marco Tomamichel.
\newblock {Partially Smoothed Information Measures}.
\newblock {\em IEEE Transactions on Information Theory}, 66(8):5022--5036,
  2020.

\bibitem[AFDF{\etalchar{+}}18]{ADF+18}
Rotem Arnon-Friedman, Fr\'{e}d\'{e}ric Dupuis, Omar Fawzi, Renato Renner, and
  Thomas Vidick.
\newblock {Practical device-independent quantum cryptography via entropy
  accumulation}.
\newblock {\em Nature Communications}, 9(1):459, 2018.

\bibitem[AFRV19]{AFRV19}
Rotem Arnon-Friedman, Renato Renner, and Thomas Vidick.
\newblock {Simple and Tight Device-Independent Security Proofs}.
\newblock {\em {SIAM} Journal on Computing}, 48(1):181--225, 2019.

\bibitem[BB84]{BB84}
Charles~H. Bennett and Gilles Brassard.
\newblock {Quantum cryptography: Public key distribution and coin tossing}.
\newblock In {\em {Proceedings of International Conference on Computers,
  Systems and Signal Processing}}, page 175, 1984.

\bibitem[BBCR13]{BBCR13}
Boaz Barak, Mark Braverman, Xi~Chen, and Anup Rao.
\newblock {How to Compress Interactive Communication}.
\newblock {\em SIAM Journal on Computing}, 42(3):1327--1363, 2013.

\bibitem[BCM{\etalchar{+}}21]{BCM+18}
Zvika Brakerski, Paul Christiano, Urmila Mahadev, Umesh Vazirani, and Thomas
  Vidick.
\newblock A cryptographic test of quantumness and certifiable randomness from a
  single quantum device.
\newblock {\em Journal of the ACM}, 68(5), 2021.

\bibitem[BJS10]{BJS10}
Michael~J. Bremner, Richard Jozsa, and Dan~J. Shepherd.
\newblock Classical simulation of commuting quantum computations implies
  collapse of the polynomial hierarchy.
\newblock {\em Proceedings of the Royal Society A: Mathematical, Physical and
  Engineering Sciences}, 467(2126):459--472, August 2010.

\bibitem[BKVV20]{BKVV20}
Zvika Brakerski, Venkata Koppula, Umesh Vazirani, and Thomas Vidick.
\newblock {Simpler Proofs of Quantumness}.
\newblock In {\em 15th Conference on the Theory of Quantum Computation,
  Communication and Cryptography (TQC 2020)}, volume 158 of {\em Leibniz
  International Proceedings in Informatics (LIPIcs)}, pages 8:1--8:14. Schloss
  Dagstuhl--Leibniz-Zentrum f{\"u}r Informatik, 2020.

\bibitem[BR11]{BR11}
Mark Braverman and Anup Rao.
\newblock {Information Equals Amortized Communication}.
\newblock In {\em Proceedings of the 2011 IEEE 52nd Annual Symposium on
  Foundations of Computer Science (FOCS '11)}, page 748–757, 2011.

\bibitem[BRdW08]{BRW08}
Avraham {Ben-Aroya}, Oded Regev, and Ronald de~Wolf.
\newblock {A Hypercontractive Inequality for Matrix-Valued Functions with
  Applications to Quantum Computing and LDCs}.
\newblock In {\em Proceedings of the 49th Annual IEEE Symposium on Foundations
  of Computer Science, FOCS '08}, pages 477--486, 2008.

\bibitem[BRWY13a]{BRWY13a}
Mark Braverman, Anup Rao, Omri Weinstein, and Amir Yehudayoff.
\newblock {Direct Product via Round-Preserving Compression}.
\newblock In {\em Automata, Languages, and Programming}, pages 232--243, 2013.

\bibitem[BRWY13b]{BRWY13b}
Mark Braverman, Anup Rao, Omri Weinstein, and Amir Yehudayoff.
\newblock {Direct Products in Communication Complexity}.
\newblock In {\em Proceedings of the 2013 IEEE 54th Annual Symposium on
  Foundations of Computer Science (FOCS '13)}, page 746–755, 2013.

\bibitem[BVY17]{BVY17}
Mohammad Bavarian, Thomas Vidick, and Henry Yuen.
\newblock {Hardness Amplification for Entangled Games via Anchoring}.
\newblock In {\em Proceedings of the 49th Annual ACM SIGACT Symposium on Theory
  of Computing}, STOC '17, page 303–316, 2017.

\bibitem[BYJKS02]{BJKS02}
Ziv Bar-Yossef, T.~S. Jayram, Ravi Kumar, and D.~Sivakumar.
\newblock {An Information Statistics Approach to Data Stream and Communication
  Complexity}.
\newblock In {\em Proceedings of the 43th Annual IEEE Symposium on Foundations
  of Computer Science, FOCS '02}, pages 209--218, 2002.

\bibitem[CSUU08]{CSUU08}
Richard Cleve, William Slofstra, Falk Unger, and Sarvagya Upadhyay.
\newblock {Perfect Parallel Repetition Theorem for Quantum XOR Proof Systems}.
\newblock {\em Computational Complexity}, 17(2):282--299, 2008.

\bibitem[CSWY01]{CSWY01}
Amit Chakrabarti, Yaoyun Shi, Anthony Wirth, and Andrew Yao.
\newblock {Informational Complexity and the Direct Sum Problem for Simultaneous
  Message Complexity}.
\newblock In {\em Proceedings of the 42nd Annual IEEE Symposium on Foundations
  of Computer Science, FOCS '01}, pages 270--278, 2001.

\bibitem[DFR20]{DFR17}
Fr{\'{e}}d{\'{e}}ric Dupuis, Omar Fawzi, and Renato Renner.
\newblock {Entropy Accumulation}.
\newblock {\em Communications in Mathematical Physics}, 379(3):867--913, 2020.

\bibitem[dGdW02]{dGdW02}
Mart de~Graaf and Ronald de~Wolf.
\newblock On quantum versions of the yao principle.
\newblock In {\em STACS 2002}, pages 347--358, 2002.

\bibitem[DPVR12]{DPVR12}
Anindya De, Christopher Portmann, Thomas Vidick, and Renato Renner.
\newblock Trevisan's extractor in the presence of quantum side information.
\newblock {\em SIAM Journal on Computing}, 41(4):915--940, 2012.

\bibitem[DSV15]{DSV15}
Irit Dinur, David Steurer, and Thomas Vidick.
\newblock {A Parallel Repetition Theorem for Entangled Projection Games}.
\newblock {\em Computational Complexity}, 24(2):201–254, 2015.

\bibitem[HJMR10]{HJMR10}
Prahladh Harsha, Rahul Jain, David McAllester, and Jaikumar Radhakrishnan.
\newblock {The Communication Complexity of Correlation}.
\newblock {\em IEEE Transactions on Information Theory}, 56(1):438--449, 2010.

\bibitem[Hol07]{Hol09}
Thomas Holenstein.
\newblock {Parallel Repetition: Simplifications and the No-Signaling Case}.
\newblock In {\em Proceedings of the Thirty-Ninth Annual ACM Symposium on
  Theory of Computing}, STOC '07, page 411–419, 2007.

\bibitem[JK21]{JK20}
Rahul Jain and Srijita Kundu.
\newblock {A Direct Product Theorem for One-Way Quantum Communication}.
\newblock In {\em Proceedings of the 36th IEEE Annual Computational Complexity
  Conference (CCC 2021)}, pages 27:1--27:28, 2021.

\bibitem[JMS20]{JMS17}
Rahul Jain, Carl~A. Miller, and Yaoyun Shi.
\newblock {Parallel Device-Independent Quantum Key Distribution}.
\newblock {\em IEEE Transactions on Information Theory}, 66(9):5567--5584,
  2020.

\bibitem[JN12]{JN12}
Rahul Jain and Ashwin Nayak.
\newblock {Short Proofs of the Quantum Substate Theorem}.
\newblock {\em IEEE Transactions on Information Theory}, 58(6):3664--3669,
  2012.

\bibitem[JPY14]{JPY14}
Rahul Jain, Attila Pereszl\'{e}nyi, and Penghui Yao.
\newblock {A Parallel Repetition Theorem for Entangled Two-Player One-Round
  Games under Product Distributions}.
\newblock In {\em 2014 IEEE 29th Conference on Computational Complexity (CCC
  '14)}, pages 209--216, 2014.

\bibitem[JPY16]{JPY16}
Rahul Jain, Attila Pereszl\'{e}nyi, and Penghui Yao.
\newblock {A Direct Product Theorem for Two-Party Bounded-Round Public-Coin
  Communication Complexity}.
\newblock {\em Algorithmica}, 76(3):720–748, 2016.

\bibitem[JRS02]{JRS02}
Rahul Jain, Jaikumar Radhakrishnan, and Pranab Sen.
\newblock {The Quantum Communication Complexity of the Pointer Chasing Problem:
  The Bit Version}.
\newblock In {\em FSTTCS 2002: Foundations of Software Technology and
  Theoretical Computer Science}, volume 2556 of {\em Lecture Notes in Computer
  Science}, pages 218--229, 2002.

\bibitem[JRS05]{JRS03}
Rahul Jain, Jaikumar Radhakrishnan, and Pranab Sen.
\newblock {Prior Entanglement, Message Compression and Privacy in Quantum
  Communication}.
\newblock In {\em 20th Annual IEEE Conference on Computational Complexity (CCC
  '05)}, pages 285--296, 2005.

\bibitem[JRS09]{JRS09}
Rahul Jain, Jaikumar Radhakrishnan, and Pranab Sen.
\newblock {A Property of Quantum Relative Entropy with an Application to
  Privacy in Quantum Communication}.
\newblock {\em Journal of the ACM}, 56(6), 2009.

\bibitem[JY12]{JY12}
Rahul Jain and Penghui Yao.
\newblock {A Strong Direct Product Theorem in Terms of the Smooth Rectangle
  Bound}.
\newblock \url{http://arxiv.org/abs/1209.0263}, 2012.

\bibitem[KLVY22]{KLVY22}
Yael Kalai, Alex Lombardi, Vinod Vaikuntanathan, and Lisa Yang.
\newblock {Quantum Advantage from Any Non-local Game}.
\newblock \url{https://arxiv.org/abs/2203.15877}, 2022.

\bibitem[KMCVY22]{KCVY21}
Gregory~D. Kahanamoku-Meyer, Soonwon Choi, Umesh~V. Vazirani, and Norman~Y.
  Yao.
\newblock Classically verifiable quantum advantage from a computational bell
  test.
\newblock {\em Nature Physics}, 18(8):918--924, August 2022.

\bibitem[KN96]{KN96}
Eyal Kushilevitz and Noam Nisan.
\newblock {\em {Communication Complexity}}.
\newblock Cambridge University Press, 1996.

\bibitem[KRS09]{KRS09}
Robert Konig, Renato Renner, and Christian Schaffner.
\newblock The operational meaning of min- and max-entropy.
\newblock {\em IEEE Transactions on Information Theory}, 55(9):4337--4347,
  2009.

\bibitem[KRT10]{KRT10}
Julia Kempe, Oded Regev, and Ben Toner.
\newblock {Unique Games with Entangled Provers are Easy}.
\newblock {\em {SIAM} Journal on Computing}, 39(7):3207--3229, 2010.

\bibitem[K{\v{S}}dW07]{KSW07}
Hartmut Klauck, Robert {\v{S}}palek, and Ronald de~Wolf.
\newblock {Quantum and Classical Strong Direct Product Theorems and Optimal
  Time-Space Tradeoffs}.
\newblock {\em SIAM Journal on Computing}, 36(5):1472--1493, 2007.

\bibitem[KT20]{KT20}
Srijita Kundu and Ernest Y.-Z. Tan.
\newblock {Composably secure device-independent encryption with certified
  deletion}.
\newblock \url{https://arxiv.org/abs/2011.12704}, 2020.

\bibitem[LLL{\etalchar{+}}21]{LLL+21}
Yong~(Alexander) Liu, Xin~(Lucy) Liu, Fang~(Nancy) Li, Haohuan Fu, Yuling Yang,
  Jiawei Song, Pengpeng Zhao, Zhen Wang, Dajia Peng, Huarong Chen, Chu Guo,
  Heliang Huang, Wenzhao Wu, and Dexun Chen.
\newblock Closing the "quantum supremacy" gap: Achieving real-time simulation
  of a random quantum circuit using a new sunway supercomputer.
\newblock In {\em Proceedings of the International Conference for High
  Performance Computing, Networking, Storage and Analysis}, SC '21, 2021.

\bibitem[LLR12]{LLR12}
Sophie Laplante, Virginie Lerays, and J{\'{e}}r{\'{e}}mie Roland.
\newblock {Classical and Quantum Partition Bound and Detector Inefficiency}.
\newblock In {\em Automata, Languages, and Programming}, pages 617--628, 2012.

\bibitem[LS09]{LS08}
Troy Lee and Adi Shraibman.
\newblock Lower bounds in communication complexity.
\newblock {\em Foundations and Trends® in Theoretical Computer Science},
  3(4):263--399, 2009.

\bibitem[LS{\v{S}}08]{LSS08}
Troy Lee, Adi Shraibman, and Robert {\v{S}}palek.
\newblock {A Direct Product Theorem for Discrepancy}.
\newblock In {\em Proceedings of the 23rd Annual IEEE Conference on
  Computational Complexity, CCC '08}, pages 71--80, 2008.

\bibitem[MFSR22]{MFSR22}
Tony Metger, Omar Fawzi, David Sutter, and Renato Renner.
\newblock {Generalized entropy accumulation}.
\newblock \url{https://arxiv.org/abs/2203.04989}, 2022.

\bibitem[PAB{\etalchar{+}}09]{PAB+09}
Stefano Pironio, Antonio Ac\'in, Nicolas Brunner, Nicolas Gisin, Serge Massar,
  and Valerio Scarani.
\newblock {Device-independent quantum key distribution secure against
  collective attacks}.
\newblock {\em New Journal of Physics}, 11(4):045021, 2009.

\bibitem[PM13]{PM13}
Stefano Pironio and Serge Massar.
\newblock Security of practical private randomness generation.
\newblock {\em Physical Review A}, 87:012336, Jan 2013.

\bibitem[PR14]{PR14}
Christopher Portmann and Renato Renner.
\newblock {Cryptographic security of quantum key distribution}.
\newblock \url{https://arxiv.org/abs/1409.3525v1}, 2014.

\bibitem[Rao08]{Rao10}
Anup Rao.
\newblock Parallel repetition in projection games and a concentration bound.
\newblock In {\em Proceedings of the Fortieth Annual ACM Symposium on Theory of
  Computing}, STOC '08, page 1–10. Association for Computing Machinery, 2008.

\bibitem[Raz92]{Raz92}
Alexander~A. Razborov.
\newblock {On the Distributional Complexity of Disjointness}.
\newblock {\em Theoretical Computer Science}, 106(2):385--390, 1992.

\bibitem[Raz95]{Raz95}
Ran Raz.
\newblock {A Parallel Repetition Theorem}.
\newblock In {\em Proceedings of the Twenty-Seventh Annual ACM Symposium on
  Theory of Computing}, page 447–456, 1995.

\bibitem[Ren05]{Ren-th}
Renato Renner.
\newblock {\em {Security of Quantum Key Distribution}}.
\newblock PhD thesis, ETH Z\"urich, 2005.

\bibitem[She18]{She12}
Alexander~A. Sherstov.
\newblock {Compressing Interactive Communication Under Product Distributions}.
\newblock {\em SIAM Journal on Computing}, 47(2):367--419, 2018.

\bibitem[Sio58]{Sion58}
Maurice Sion.
\newblock {On general minimax theorems}.
\newblock {\em Pacific Journal of Mathematics}, 8(1):171--176, 1958.

\bibitem[SPM13]{SPM13}
Jonathan Silman, Stefano Pironio, and Serge Massar.
\newblock Device-independent randomness generation in the presence of weak
  cross-talk.
\newblock {\em Phys. Rev. Lett.}, 110:100504, 2013.

\bibitem[TL17]{TL17}
Marco Tomamichel and Anthony Leverrier.
\newblock {A largely self-contained and complete security proof for quantum key
  distribution}.
\newblock {\em {Quantum}}, 1:14, 2017.

\bibitem[Tom16]{Tom16}
Marco Tomamichel.
\newblock {\em {Quantum Information Processing with Finite Resources}}.
\newblock Springer International Publishing, 2016.

\bibitem[Tsi87]{Tsi87}
B.~S. Tsirelson.
\newblock Quantum analogues of the bell inequalities. the case of two spatially
  separated domains.
\newblock {\em Journal of Soviet Mathematics}, 36(4):557--570, 1987.

\bibitem[TZCBB{\etalchar{+}}20]{TZB+20}
Armin Tavakoli, Emmanuel Zambrini~Cruzeiro, Jonatan Bohr~Brask, Nicolas Gisin,
  and Nicolas Brunner.
\newblock Informationally restricted quantum correlations.
\newblock {\em {Quantum}}, 4:332, 2020.

\bibitem[TZCWP20]{TZWP20}
Armin Tavakoli, Emmanuel Zambrini~Cruzeiro, Erik Woodhead, and Stefano Pironio.
\newblock {Informationally restricted correlations: a general framework for
  classical and quantum systems}, 2020.

\bibitem[Vid17]{Vid17}
Thomas Vidick.
\newblock {Parallel DIQKD from parallel repetition}.
\newblock \url{https://arxiv.org/abs/1703.08508}, 2017.

\bibitem[VV19]{VV19}
Umesh Vazirani and Thomas Vidick.
\newblock Fully device independent quantum key distribution.
\newblock {\em Communications of the ACM}, 62(4):133, 2019.

\bibitem[VW08]{VW08}
Emanuele Viola and Avi Wigderson.
\newblock {Norms, {XOR} Lemmas, and Lower Bounds for Polynomials and
  Protocols}.
\newblock {\em Theory of Computing}, 4(7):137--168, 2008.

\bibitem[Yao77]{Yao79}
Andrew Chi-Chih Yao.
\newblock {Probabilistic computations: Toward a unified measure of complexity}.
\newblock In {\em 18th Annual Symposium on Foundations of Computer Science
  (SFCS 1977)}, pages 222--227, 1977.

\bibitem[Yue16]{Yuen16}
Henry Yuen.
\newblock {A Parallel Repetition Theorem for All Entangled Games}.
\newblock In {\em 43rd International Colloquium on Automata, Languages, and
  Programming (ICALP '16)}, volume~55 of {\em Leibniz International Proceedings
  in Informatics (LIPIcs)}, pages 77:1--77:13, 2016.

\bibitem[YZ22]{YZ22}
Takashi Yamakawa and Mark Zhandry.
\newblock Verifiable quantum advantage without structure.
\newblock In {\em IEEE 63rd Annual Symposium on Foundations of Computer Science
  (FOCS)}, pages 69--74, 2022.

\end{thebibliography}

\appendix

\section{Proof of Yao's lemma}\label{ap:yao}
In order to prove this lemma, we shall use Sion's minimax theorem \cite{Sion58}. Sion's minimax theorem is true for quasisaddle and semicontinuous functions in general. However, we shall only need to apply it to continuous and saddle functions, so we shall state it for such below. A function $f:\clX\times\clY\to\bbR$ is called \emph{saddle} if $f(\cdot,y)$ is convex as a function of $x$ for each fixed $y\in\clY$, and $f(x,\cdot)$ is concave as a function of $y$ for each fixed $x\in\clX$.
\begin{fact}[\cite{Sion58}]
Suppose $\clX$ and $\clY$ are convex spaces, one of which is compact, and $f:\clX\times\clY\to\bbR$ is a continuous (in both arguments) and saddle function. Then,
\[ \inf_{x\in\clX}\sup_{y\in\clY}f(x,y) = \sup_{y\in\clY}\inf_{x\in\clX}f(x,y).\]
\end{fact}

The set $\clY$ that we shall be considering a supremum over in our application of Sion's minimax lemma is going to be the set of probability distributions over the inputs to a communication protocol. Since the set of inputs is finite, this set is compact. It is also obviously convex.

A quantum protocol $\clP$ running on input $x=x^1\ldots x^l$ gives rise to a distribution $\sfP_{A|\clP,x}$ on the outputs. Let $\sfP_{A|\clP}$ denote the tuple $(\sfP_{A|\clP,x})_{x\in\clX^1\times\ldots\times\clX^l}$, which is the entire output distribution of $\clP$. Let $\bbP^\Q_t$ denote the set of $\sfP_{A|\clP}$ for quantum protocols $\clP$ with at most $t$ communication. Given two protocols $\clP_1$ and $\clP_2$ between $l$ players, with at most $r$ rounds and $t$ communication, their convex combination can be implemented by the players sharing an appropriate superposition of $\ket{0^l}$ and $\ket{1^l}$ as entanglement, and implementing the steps of $\clP_1$ or $\clP_2$ controlled on their part of the shared entanglement being $\ket{0}$ or $\ket{1}$. The resultant protocol $\clP$ has at most $r$ rounds and $t$ communication. Therefore, the set $\bbP^\Q_t$ is convex. This will be the set $\clX$ we take infimum over in Sion's minimax theorem.

For a protocol $\clP$ and input $x$, let $\err_\sfV(\clP,x)$ denote the error that $\clP$ makes on input $x$ for predicate $\sfV$. With some abuse of notation, we shall also use $\err_\sfV(\clP,\mu)$ to denote the average error of $\clP$ over an input distribution $\mu$. Note that $\err_\sfV(\clP,x)$ and $\err_\sfV(\clP,\mu)$ can be defined using just $\sfP_{A|\clP}$, and with abuse of notation we shall use $\err_\sfV(\sfP_{A|\clP},x)$ and $\err_\sfV(\sfP_{A|\clP},\mu)$ to denote these. The set of all quantum protocols which communicate at most $t$ qubits has a somewhat complicated structure, but the infimum of $\err_\sfV(\clP,\mu)$ over this set is the same as the infimum of $\err_\sfV(\sfP_{A|\clP},\mu)$ over $\bbP^\Q_t$.

For a fixed $\clP$, $\err_\sfV(\clP,\cdot)$ (or $\err_\sfV(\sfP_{A|\clP},\cdot)$) is linear\footnote{Here by `linear' we mean a function that is both convex and concave.} in the second argument, because the error for a convex combination of distributions is just the convex combination of the errors for the individual distributions. Similarly, the error for a convex combination of algorithms is the convex combination of errors for the individual algorithms, i.e., $\err_\sfV(\cdot,\mu)$ is linear in the first argument. Therefore $\err_f$ is saddle, and is also continuous in both arguments.

The function $\err_\sfV$, the set $\bbP^\Q_t$ and the set of probability distributions satisfy the conditions of Sion's minimax theorem. Therefore we can say,
\[
\inf_{\sfP_{A|\clP}\in\bbP^\Q_t}\sup_\mu\, \err_\sfV(\sfP_{A|\clP},\mu) = \sup_\mu\inf_{\sfP_{A|\clP}\in\bbP^\Q_t}\,\err_\sfV(\sfP_{A|\clP},\mu).
\]
It is clear that the supremum over $\mu$ on the left-hand side of the above equation is the same as the supremum over inputs $x$. Now consider $t=\Q_\eps(\sfV)-1$. By the definition of $\Q_\eps(\sfV)$, the left-hand sides of the equation must then be strictly greater than $\eps$. Moreover, the sup on the right-hand side is a max, since the set of distributions is compact. Hence, there is a distribution $\mu$ such that all algorithms making less than $\Q_\eps(\sfV)$ queries must make more than $\eps$ error over inputs from $\mu$, i.e.,
\[ \max_\mu\Q_\eps(\sfV,\mu) \geq \Q_\eps(\sfV). \]
Moreover $\Q_\eps(\sfV) \geq \max_\mu\Q_\eps(\sfV,\mu)$ also holds, since an algorithm that makes $\eps$ error in the worst case also makes at most $\eps$ error averaged over $\mu$. This completes the proof of the lemma.

\section{Proof of Lemma \ref{int-holevo}}\label{ap:int-holevo}
We shall do induction on the number of rounds. Let $c_i$ be the communication in the $i$-th round and $\clR^j = \{j, j+l, \ldots \}$ denote the set of rounds in which the $j$-th player communicates, so that $\sum_{i \in \clR^j}c_i = c^j$. Let $M_i$ be the message register of the $i$-th round, $E_i$ be the memory register the party who communicates in the $i$-th round holds after sending their message. For $i \in \clR^j$, the registers held by the $j$-th party at the beginning of the $i$-th round are messages $M^j_{i-l+1}\ldots M^j_{i-1}$ from other parties in the $(i-l+1)$-th to $(i-1)$-th rounds, which we shall jointly denote by $N^j_{i-1}$, and their memory register $E_{i-l}$ which they have retained from the $(i-l)$-th round. We shall denote all other (non-input) registers held by parties other than the $j$-th party at the beginning of the $I$-th round by $F^{-j}_{i-1}$. Since $i \in \clR^j$, clearly $F^{-j}_i = F^{-j}_{i-1}M_i$. Using $X$ to denote $X^1\ldots X^l$ and similar notation for $\tX$, we shall call the shared state including the input purifications at the beginning of the the $i$-th round 
\[ \ket{\sigma^i}_{X\tX N^j_{i-1}E_{i-l}F^{-j}_{i-1}} = \sum_x\sqrt{\sfP_X(x)}\ket{xx}_{X\tX}\ket{\sigma^i}_{N^j_{i-1}E_{i-l}F^{-j}_{i-1}}.\]

For the base case $i=1$, communication is zero. Since $\sfP_{X^1\ldots X^l}$ is a product distribution, $\sigma^1_{X^jX^{-j}\tX^{-j} F^{-j}_0}$ is product between $X^j$ and the other registers, $F^{-j}_0$ being simply the other parties' parts of the initial shared entangled state, which is independent of the inputs. So the condition trivially holds. For the induction step, we shall assume the condition
\[ \sfD_\infty\left(\sigma^i_{X^jX^{-j}\tX^{-j}F^{-j}_{i-1}}\middle\Vert\sigma^i_{X^j}\otimes\rho^i_{X^{-j}\tX^{-j}F^{-j}_{i-1}}\right) \leq 2\sum_{\substack{i' \in \clR^j, \\ i' < i}}c_{i'}\]
holds at the beginning of the $i$-th round, where $i \in \clR^j$, for some state $\rho^i_{X^{-j}\tX^{-j}F^{-j}_{i-1}}$, and see how it changes in the $i$-th to $(i+l-1)$-th rounds.

In the $i$-th round, the $j$-th party applies a unitary on the $X^jN^j_{i-1}E_{i-l}$ registers, getting registers $X^jM_iE_i$. By Fact \ref{dim-ub}, there exists a state $\tilde{\rho}^{i+1}_{M_i}$ such that
\[ \sfD_\infty\left(\sigma^{i+1}_{X^jX^{-j}\tX^{-j}F^{-j}_{i-1}M_i}\middle\Vert\sigma^{i+1}_{X^jX^{-j}\tX^{-j} F^{-j}_{i-1}}\otimes\tilde{\rho}^{i+1}_{M_i}\right) \leq 2c_i.\]
Now note that the marginal states $\sigma^i_{X^jX^{-j}\tX^{-j}F^{-j}_{i-1}}$ and $\sigma^{i+1}_{X^jX^{-j}\tX^{-j} F^{-j}_{i-1}}$ are exactly the same, since the unitary relating $\ket{\sigma^i}$ and $\ket{\sigma^{i+1}}$ does not act on $X^{-j}\tX^{-j} F^{-j}_{i-1}$ at all, and only uses $X^j$ as a control register. Hence we have,
\begin{align*}
& \sfD_\infty\left(\sigma^{i+1}_{X^jX^{-j}\tX^{-j} F^{-j}_{i-1}}\otimes\tilde{\rho}^{i+1}_{M_i}\middle\Vert\sigma^{i+1}_{X^j}\otimes\rho^i_{X^{-j}\tX^{-j} F^{-j}_{i-1}}\otimes\tilde{\rho}^{i+1}_{M_i}\right) \\
& = \sfD_\infty\left(\sigma^i_{X^jX^{-j}\tX^{-j} F^{-j}_{i-1}}\otimes\tilde{\rho}^{i+1}_{M_i}\middle\Vert\sigma^i_{X^j}\otimes\rho^i_{X^{-j}\tX^{-j} F^{-j}_{i-1}}\otimes\tilde{\rho}^{i+1}_{M_i}\right) \\
& = \sfD_\infty\left(\sigma^i_{X^jX^{-j}\tX^{-j} F^{-j}_{i-1}}\middle\Vert\sigma^i_{X^j}\otimes\rho^i_{X^{-j}\tX^{-j} F^{-j}_{i-1}}\right) \\
& \leq 2\sum_{\substack{i' \in \clR^j, \\ i' < i}}c_{i'}.
\end{align*}
Now using Fact \ref{fc:Sinfty-tri} we can say,
\begin{align*}
& \sfD_\infty\left(\sigma^{i+1}_{X^jX^{-j}\tX^{-j}F^{-j}_{i-1}M_i}\middle\Vert\sigma^{i+1}_{X^j}\otimes\rho^i_{X^{-j}\tX^{-j} F^{-j}_{i-1}}\otimes\tilde{\rho}^{i+1}_{M_i}\right) \\
& \leq \sfD_\infty\left(\sigma^{i+1}_{X^jX^{-j}\tX^{-j}F^{-j}_{i-1}M_i}\middle\Vert\sigma^{i+1}_{X^jX^{-j}\tX^{-j} F^{-j}_{i-1}}\otimes\tilde{\rho}^{i+1}_{M_i}\right) \\
& \quad + \sfD_\infty\left(\sigma^{i+1}_{X^jX^{-j}\tX^{-j} F^{-j}_{i-1}}\otimes\tilde{\rho}^{i+1}_{M_i}\middle\Vert\sigma^{i+1}_{X^j}\otimes\rho^i_{X^{-j}\tX^{-j} F^{-j}_{i-1}}\otimes\tilde{\rho}^{i+1}_{M_i}\right) \\
& \leq 2c_i + 2\sum_{\substack{i' \in \clR^j, \\ i' < i}}c_{i'} \\
& = 2\sum_{\substack{i' \in \clR^j, \\ i' \leq i}}c_{i'}.
\end{align*}
Hence the condition holds at the beginning of the $(i+1)$-th round with $\rho^{i+1}_{X^{-j}\tX^{-j} F^{-j}_{i-1}M_i} = \rho^i_{X^{-j}\tX^{-j} F^{-j}_{i-1}}\otimes\tilde{\rho}^{i+1}_{M_i}$.

In the $(i+1)$-th round, the $(j+1)$-th player applies a unitary on the $X^{j+1}N^{j+1}_iE_{i-l+1}$ registers, getting registers $X^jM^1_{i+1}\ldots M^j_{i+1}\ldots M^l_{i+1}E_{i+1}$, of which they send $M^j_{i+1}$ to the $j$-th player. So after this round, the registers held by the $j$-th player are $E_iM^j_{i+1}$, and $F^{-j}_{i+1}$ does not include $M^j_{i+1}$. By Fact \ref{fc:u-inv} we have that,
\begin{align*}
\sfD_\infty\left(\sigma^{i+2}_{X^jX^{-j}\tX^{-j} M^j_{i+1}F^{-j}_i}\middle\Vert \sigma^{i+2}_{X^j}\otimes \rho^{i+2}_{X^{-j}\tX^{-j} M^j_{i+1}F^{-j}_i}\right) & = \sfD_\infty\left(\sigma^{i+1}_{X^jX^{-j}\tX^{-j}F^{-j}_{i-1}M_i}\middle\Vert\sigma^{i+1}_{X^j}\otimes \rho^{i+1}_{X^{-j}\tX^{-j} F^{-j}_{i-1}M_i}\right)\\
& \leq 2\sum_{\substack{i' \in \clR^j, \\ i' \leq i}}c_{i'}
\end{align*}
where $\rho^{i+2}$ is the state obtained by applying the $(j+1)$-th player's unitary in the $(i+1)$-th round to $\rho^{i+1}$. From this we can trace out the $M^j_{i+1}$-th register to show that
\[ \sfD_\infty\left(\sigma^{i+2}_{X^jX^{-j}\tX^{-j} F^{-j}_{i+1}}\middle\Vert \sigma^{i+2}_{X^j}\otimes \rho^{i+2}_{X^{-j}\tX^{-j} F^{-j}_{i+1}}\right)\leq 2\sum_{\substack{i' \in \clR^j, \\ i' \leq i}}c_{i'}.\]

The bound is similarly unchanged in the rounds $i+2, \ldots, i+l-1$. Hence we can say that at the beginning of the next round $i+l$ in which the $j$-th party communicates, it holds that
\[ \sfD_\infty\left(\sigma^{i+l}_{X^jX^{-j}\tX^{-j}F^{-j}_{i+l-1}}\middle\Vert\sigma^{i+l}_{X^j}\otimes\rho^{i+l}_{X^{-j}\tX^{-j}F^{-j}_{i+l-1}}\right) \leq 2\sum_{\substack{i' \in \clR^j, \\ i' < i+l}}c_{i'}.\]


\end{document}